\RequirePackage{fix-cm}

\documentclass{article}

\usepackage[english]{babel}
\usepackage[utf8]{inputenc}
\usepackage[T1]{fontenc}
\usepackage{lmodern}
\rmfamily
\DeclareFontShape{T1}{lmr}{b}{sc}{<->ssub*cmr/bx/sc}{}
\DeclareFontShape{T1}{lmr}{bx}{sc}{<->ssub*cmr/bx/sc}{}

\usepackage[a4paper, margin=1in]{geometry}
\usepackage{graphicx}
\usepackage{todonotes}
\usepackage{color}

\definecolor{darkgreen}{rgb}{0,0.5,0}
\usepackage{hyperref}
\hypersetup{
    unicode=false,          % non-Latin characters in Acrobat’s bookmarks
    colorlinks=true,        % false: boxed links; true: colored links
    linkcolor=red,          % color of internal links (change box color with linkbordercolor)
    citecolor=darkgreen,    % color of links to bibliography
    filecolor=magenta,      % color of file links
    urlcolor=cyan           % color of external links
}

\usepackage{amsthm}
\usepackage{amsmath}
\usepackage{amssymb}
\usepackage{amsfonts}
\usepackage{mathrsfs}
\usepackage{mathtools}
\usepackage{verbatim}
\usepackage{footnote}
\usepackage[linesnumbered,boxed, vlined]{algorithm2e}
\usepackage{lineno}
\usepackage{caption}
\usepackage{framed}
\usepackage[framemethod=tikz]{mdframed}
\global\mdfdefinestyle{myframe}{leftmargin=.75in,rightmargin=.75in,linecolor=black,linewidth=0.5pt,innertopmargin=10pt,innerbottommargin=10pt} 
\usepackage{enumerate}

\usepackage{tikzsymbols}
\usepackage{thmtools,thm-restate}
\usepackage{nicefrac}
\usepackage{bold-extra}

\usepackage{tikz}
\tikzstyle{vertex}=[circle, draw,fill=black, inner sep=0pt, minimum size=6pt]

\usepackage[capitalize, nameinlink]{cleveref}
\usepackage[section]{placeins}
\crefname{theorem}{Theorem}{Theorems}
\Crefname{lemma}{Lemma}{Lemmas}
\Crefname{claim}{Claim}{Claims}
\Crefname{observation}{Observation}{Observations}

\newcommand{\unwaugpath}{\textsc{Unw-3-Aug-Paths}\xspace}

\newtheorem{theorem}{Theorem}[section]
\newtheorem{lemma}[theorem]{Lemma}

\newtheorem{definition}[theorem]{Definition}

\newtheorem{obsv}[theorem]{Observation}
\newtheorem{fact}[theorem]{Fact}

\newtheorem*{remark}{Remark}

\newcommand{\advntg}{\alpha}
\renewcommand{\ge}{\geqslant}
\renewcommand{\le}{\leqslant}
\renewcommand{\geq}{\geqslant}
\renewcommand{\leq}{\leqslant}

\newcommand{\eps}{\varepsilon}
\newcommand{\EE}{\mathbb{E}}
\newcommand{\prob}[1]{\Pr\left[ #1 \right]}

\DeclareMathOperator{\poly}{poly}
\DeclareMathOperator{\polylog}{polylog}
\newcommand{\rb}[1]{\left( #1 \right)}

\newcommand{\opt}{\operatorname{OPT}}
\newcommand{\optWM}{M^{\star}}
\newcommand{\optwgt}{w(\optWM)}
\newcommand{\optM}{\optWM}

\newcommand{\cA}{\mathcal{A}}
\newcommand{\cC}{\mathcal{C}}
\newcommand{\cD}{\mathcal{D}}

\newcommand{\cP}{\mathcal{P}}
\newcommand{\cT}{\mathcal{T}}
\newcommand{\cW}{\mathcal{W}}
\newcommand{\cX}{\mathcal{X}}
\newcommand{\tcC}{\tilde{\cC}}

\newcommand{\Layered}{\mathcal{L}}
\newcommand{\VLayered}{V_{\Layered}}
\newcommand{\ELayered}{E_{\Layered}}
\newcommand{\bbR}{\mathbb{R}}
\newcommand{\bbN}{\mathbb{N}}
\newcommand{\directedS}{\vec{S}}
\newcommand{\directedP}{\vec{P}}

\newcommand{\arc}[1]{\vec{#1}}

\newcommand{\cCbad}{\cC_{\operatorname{bad}}}

\newcommand{\Nmatching}[1]{{#1}^{M}}
\newcommand{\cApair}{\cA^{(\tau^A, \tau^B)}}
\newcommand{\partupple}{(L, R, A, B)}

\newcommand{\machines}{\Gamma}

\newcommand{\eqdef}{\stackrel{\text{\tiny\rm def}}{=}}

\newcommand{\layerepspower}{\eps^{12}}
\newcommand{\gain}{w^{+}}

\newcommand{\algfontsize}{\small}

\newcommand{\UnwBipMatching}{\textsc{Unw-Bip-Matching}\xspace}
\newcommand{\MainAlg}{\textsc{Main-Alg}\xspace}

\begin{document}

\title{Weighted Matchings via Unweighted Augmentations}

 \author{
   Buddhima Gamlath \\
   \texttt{buddhima.gamlath@epfl.ch}
     \and
   Sagar Kale \\
   \texttt{sagar.kale@epfl.ch}
     \and
	 Slobodan Mitrovi\'{c} \\
	 \texttt{slobo@mit.edu}
		\and
   Ola Svensson \\
   \texttt{ola.svensson@epfl.ch}
 }\date{}

\maketitle

\begin{abstract}
  We design a generic method for reducing the task of finding weighted matchings
  to that of finding short augmenting paths in \emph{unweighted} graphs. This
  method enables us to provide efficient implementations for approximating weighted matchings  in the streaming model and in the massively parallel computation (MPC) model.
	
  In the context of streaming with random edge arrivals, our techniques yield a
  $(\nicefrac{1}{2}+c)$-approximation algorithm thus breaking the natural barrier of
  $\nicefrac{1}{2}$. For multi-pass streaming and the MPC model, we show that \emph{any} algorithm computing a $(1-\delta)$-approximate unweighted matching in bipartite graphs can be translated into
  an algorithm that computes a $(1-\eps(\delta))$-approximate maximum weighted
  matching. Furthermore, this translation incurs only a constant factor (that
  depends on $\eps> 0$) overhead in the complexity. Instantiating
  this with the current best multi-pass streaming and  MPC algorithms for unweighted matchings yields the following results for maximum weighted matchings:
  \begin{itemize}
    \item A $(1-\eps)$-approximation streaming algorithm that uses $O_\eps(1)$ passes\footnote{We use $O_\eps(f(n))$ to denote a function that is $O(f(n))$ when the parameter $\eps$ is a constant.} and $O_\eps(n \poly (\log n))$ memory. This is the first  $(1-\eps)$-approximation streaming algorithm for weighted matchings that uses a constant number of passes (only depending on $\eps$). %changed to "streaming algorithm" to avoid a really bad overfull hbox.  -Sagar.
    \item A  $(1 - \eps)$-approximation algorithm in the MPC model that uses $O_\eps(\log \log n)$
  rounds, $O(m/n)$ machines per round, and $O_\eps(n \poly(\log n))$ memory per machine. 
  %with $O(n \log{n})$ memory per machine. 
  This improves
  upon the previous best approximation guarantee of $(\nicefrac{1}{2}-\eps)$ for weighted
  graphs.
  \end{itemize}
%  Prior to our work, the best MPC algorithm in the considered regime achieved an approximation guarantee of $(\nicefrac{1}{2}-\epsilon)$ and there was no known algorithm in the streaming model that used a constant number of passes (only depending on $\epsilon$) for weighted matchings.
\end{abstract}

\section{Introduction}

The maximum  matching problem is a classic problem in combinatorial optimization. For polynomial-time computation,
efficient algorithms exist both for the unweighted (cardinality) version and
the weighted version. However, in other models of computation, the weighted
version turns out to be significantly harder, and better algorithms are known in
the unweighted case. In fact, in some settings such as online algorithms, the
weighted version is \emph{provably} much harder than the unweighted case. In
other models, such as streaming and massively parallel computation (MPC), no
such results are known. Instead the performance gap  in the algorithms for
unweighted and weighted matchings seems to arise due to a lack of techniques.
The goal of this paper is to address this by developing new techniques for
weighted matchings. % to address  these gaps in our understanding. 

In the (semi-)streaming model the edges of the graph arrive one-by-one and the
algorithm is restricted to use memory that is almost linear in the number of vertices.
For unweighted graphs, the very basic greedy algorithm guarantees to return
a $(\nicefrac{1}{2})$-approximate maximum matching. It is a major open problem to improve upon this
factor when the order of the stream is adversarial. In the random-edge-arrival
setting --- where the edges of the stream are presented in a random order --- algorithms that are more advanced than the greedy algorithm overcome this barrier~\cite{Konrad2012}.
In contrast, for weighted graphs  a $(\nicefrac{1}{2} - \eps)$-approximation algorithm was given only recently  for adversarial streams~\cite{Paz2017, Ghaffari2017}, and here we give the first algorithm that breaks the natural ``greedy'' barrier of $\nicefrac{1}{2}$ for random-edge-arrival streams:
\begin{theorem}
  There is a $(\nicefrac{1}{2} + c)$-approximation algorithm for finding weighted matchings in the streaming model with random-edge-arrivals, where $c>0$ is an absolute constant.
  \label{thm:streaming}
\end{theorem}
As we elaborate below, the result is achieved via a general approach that {reduces} the task of finding weighted matchings to that of finding (short) \emph{unweighted} augmenting paths. This allows us to incorporate some of the ideas present in the streaming algorithms for unweighted matchings to achieve our result. Our techniques, perhaps surprisingly, also simplify the previous algorithms for finding unweighted matchings, and give an improved guarantee for general graphs. % in that special case.  
%Moreover, the approach to reduce to unweighted augmenting paths is rather versatile and it also allows us to give better algorithms 

The idea to reduce to the problem of finding unweighted augmenting paths is rather versatile, and we use it to obtain a general reduction from weighted matchings to unweighted matchings as our second main result. We give implementations of this reduction in the models of multi-pass streaming and MPC that incur only a constant factor overhead in the complexity. In multi-pass streaming, the algorithm is (as for single-pass) restricted to use memory that is almost linear in the number of vertices and the complexity is measured in terms of the number of passes that the algorithm requires over the data stream.  In MPC, parallel computation is  modeled by parallel machines with sublinear memory (in the input size) and data can be transferred between machines only between two rounds of computation (see \cref{sec:prelim} for a precise definition). The complexity of an algorithm in the MPC model, also referred to as the round complexity, is then measured as the number of (communication) rounds used.

Both the streaming model and the MPC model, which encompasses many of today's most successful parallel computing paradigms such as MapReduce and Hadoop, are motivated by the need for devising efficient algorithms for large problem instances. As data and the size of instances keep growing, this becomes ever more relevant and a large body of recent work has been devoted to these models. For the matching problem, McGregor~\cite{McGregor2005} gave the first streaming algorithm for approximating unweighted matchings within a factor $(1-\eps)$ that runs in  a constant number of passes (depending only on $\eps$); the dependency on $\eps$ was more recently improved  for  bipartite graphs~\cite{ahnguha,Eggert2012}.  McGregor's techniques for unweighted matchings have been very influential. In particular,  his general reduction technique can be used to transform any $O(1)$-approximation unweighted matching algorithm that uses $R$ MPC rounds into a $(1-\eps)$ approximation unweighted matching algorithm that uses $O_{\eps}(R)$ rounds in the MPC model.   This together with  
a sequence of  recent papers~\cite{assadi2017coresets, czumaj2018round, ghaffari2018improved}, that give constant-factor approximation algorithms for 
{unweighted} matchings with improved round complexity, %achieving a guarantee of $1-\eps$ 
culminated in algorithms that find $(1-\eps)$-approximate maximum unweighted matchings in $O_{\eps}(\log \log n)$ rounds. 
However, as McGregor's techniques apply to only unweighted matchings, it was not known how to achieve an analogous result in the presence of weights.  In fact, McGregor raised as an open question whether his result can be generalized to weighted graphs.  Our result answers this in the affirmative and gives a reduction that is \emph{lossless} with respect to the approximation guarantee while only increasing the complexity by a constant factor.
Moreover, our reduction is to bipartite graphs. %Added by Sagar.
Instantiating this with the aforementioned streaming and  MPC algorithms for unweighted matchings yields the following\footnote{Throughout the paper, we denote by $n$ the number of vertices and by $m$ the number of edges.}: %(the $\tilde \Theta (\cdot)$ notation hides $\log(n)$ factors):
    \begin{theorem}\label{thm:mpc}
    	There exists an algorithm that in expectation finds a $(1-\eps)$-approximate weighted matching that can be implemented
    	\begin{enumerate}
    	\item in $O_\eps(U_M)$ rounds, $O(m/n)$ machines per round, and $O_\eps(n \poly(\log n))$ memory per machine, where $U_M$ is the number of rounds used by a $(1-\delta)$-approximation algorithm for \emph{bipartite} unweighted matching using $O(m/n)$ machines per round, and $O_\delta(n \poly(\log n))$ memory per machine in the MPC model, and
      \item in $O_{\eps}(U_S)$ passes and $O_\eps(n \poly (\log n))$ memory, where $U_S$ is the number of passes used by a $(1-\delta)$-approximation algorithm for \emph{bipartite} unweighted matching using $O_\delta(n \poly (\log n))$ memory, in the multi-pass streaming model,
          \end{enumerate}
          where $\delta = \eps^{28 + 900/\eps^2}$.
          Using the algorithm of Ghaffari et al.~\cite{ghaffari2018improved} or that of Assadi et al.~\cite{assadi2017coresets}, we get that $U_M = O_\eps(\log\log n)$ and using the algorithm of Ahn and Guha~\cite{ahnguha}, we get that $U_S = O_{\eps}(1)$.  %}O(\log\log(1/\delta)/\delta^2) = O((1/\eps)^{56+1800/\eps^2}\log(1/\eps))$.
    % 
    % 
    % 
    %   In the MPC model, for any $\eps > 0$,  we can find a $(1-\eps)$-approximate maximum weight matching in $O(\log \log n)$ rounds in expectation with $\tilde \Theta (n)$ memory per machine. 
    \end{theorem} 
%Naturally our result for weighted matchings match the current best result for unweighted ones~\cite{assadi2017coresets, czumaj2018round, ghaffari2018improved} (since our result reduces the problem of finding a weighted matching to that of finding an unweighted one). It therefore allows us to tap into the existing unweighted matchings algorithms: in a sequence of papers unweighted matchings algorithms 
%
%
%for unweighted matchings were given with improved round complexity. The current 
Prior to this, the best known results for computing a $(1-\eps)$-approximate
\emph{weighted} matching required super constant  $\Omega(\log n)$  many passes over the stream in the streaming model~\cite{ahnguha} and $\Omega(\log n)$ rounds~\cite{ahn2018access} in the MPC model.  We
remark that if we allow for memory $\tilde \Theta(n^{1+\nicefrac{1}{p}})$ per
machine in the MPC model, then~\cite{ahn2018access} gave an algorithm that uses only a constant
number of rounds (depending on $p$). Achieving a similar result with near linear memory per machine is a major open question in the MPC literature; our results show that it is sufficient to concentrate on unweighted graphs as any progress on such graphs gives analogous progress in the weighted setting.
%\stodo{The claim above is correct under the condition that there is $\Theta(n \polylog{n})$ memory per machine. Do we want to say that?}
%
%Interestingly, let us note that the factor $2$ was again
%a natural barrier for matchings: it was previously known that any constant-factor approximation
%algorithm for unweighted matchings can be used to obtain
%a $(\nicefrac{1}{2}-\eps)$-approximation for weighted matchings (compared to our tight $(1-\eps)$-factor). This together with
%a sequence~\cite{assadi2017coresets, czumaj2018round, ghaffari2018improved} of
%\emph{unweighted} matching algorithms %achieving a guarantee of $1-\eps$
%with improved round complexity implied the previous  best
%$(\nicefrac{1}{2}-\eps)$-approximation algorithm for weighted matchings when  restricted
%to $O(\log \log n)$ MPC rounds.
We now give an outline of our approach.

\subsection{Outline of Our Approach}

Let $M$ be a matching in a graph $G=(V,E)$ with edge-weights $w: E \to \mathbb{R}$. Recall that an alternating path $P$ is a path in $G$ that alternates between edges in $M$ and in $E\setminus M$.  If the endpoints of $P$ are unmatched vertices or incident to edges in $M \cap P$, then removing the $M$-edges in $P$ and adding the other edges of $P$ gives a new matching. In other words, $M\Delta P = (M \setminus (P\cap M)) \cup P\setminus M$ is a new matching. We say that we updated $M$ using the alternating path $P$, and we further say that $P$ is augmenting if $w(M\Delta P) > w(M)$ where we used the notation $w(F) = \sum_{e\in F} w(e)$ for a subset of edges $F\subseteq E$.  Also recall that an alternating cycle $C$ is a cycle that alternates between edges in $M$ and in $E\setminus M$, and $M\Delta C$ is also a matching.  We say that $C$ is augmenting if $w(M\Delta C) > w(M)$.
Now a well-known structural result regarding approximate matchings is the following:
\begin{center}
\begin{minipage}{0.9\textwidth}
\begin{mdframed}[hidealllines=true, backgroundcolor=gray!15]
  \begin{fact}
  For any $\ell \in \mathbb{N}$, if there is no augmenting path or cycle of length at most $2\ell -1$, then $M$ is a $(1-\nicefrac{1}{\ell})$-approximate matching. 
  \label{fact:shortaug}
\end{fact}
\end{mdframed}
\end{minipage}
\end{center}
In particular, this says that in order to find a $(1-\eps)$-approximate matching it is sufficient to find augmenting paths or cycles of length $O(1/\eps)$. This is indeed the most common route used to design efficient algorithms  for finding approximate matchings: in the streaming model with random-edge-arrivals,~\cite{Konrad2012}  finds augmenting paths of length $\leq 3$ and the MPC algorithms~\cite{assadi2017coresets, czumaj2018round} find augmenting paths of length $O(1/\eps)$. However, those approaches work only for unweighted graphs. The high level reason being that it is easy to characterize the augmenting paths in the unweighted setting: they simply must start and end in unmatched vertices. Such a simple classification of augmenting paths is not available in the weighted setting and the techniques of those papers do not apply.
Nevertheless, we propose a general framework to overcome this obstacle that allows us to tap into the results and techniques developed for unweighted matchings. Informally, we reduce the problem of finding  augmenting paths in the weighted setting to the unweighted setting. 

The high level idea is simple: Consider the example depicted on the left in \cref{fig:introex}. The current matching $M$ consists of a single edge $\{c,d\}$ that is depicted by a solid line. The weights are written next to the edges and so $w(M) = 5$ (the edges $E\setminus M$ are dashed). The maximum matching consists of $\{a,c\}, \{d,f\}$ and has weight $8$. Furthermore, there are several alternating paths of length $3$ that are also augmenting. However, it is important to note that we cannot simply apply an algorithm for finding unweighted augmenting paths. Such an algorithm may find the alternating path $P= b,c,d,e$ which is augmenting in the unweighted sense but $w(M \Delta P) < w(M)$. To overcome this, we apply a \emph{filtering} technique that we now explain in our simple example: First ``guess'' lower bounds on the weights of the edges incident to $c$ and $d$ in an augmenting path. Let $\tau_c$ and $\tau_d$ be those lower bounds. We then look for augmenting paths in the unweighted graph that keeps only those unmatched edges incident to $c$ and $d$ whose weights are above the guessed thresholds. Then to guarantee that an unweighted augmenting path that an algorithm finds is also an augmenting path in the weighted sense, we always set $\tau_c$ and $\tau_d$ such that $\tau_c + \tau_d > w(\{c,d\})$.   In the center and right part of \cref{fig:introex} we depict two \emph{unweighted} graphs obtained for different values of $\tau_c$ and $\tau_d$ (in the center with $\tau_c=\tau_d = 3$ and to the right with $\tau_c=2, \tau_d = 4$). Note that in both examples any unweighted augmenting path is also augmenting with respect to the weights.

\newcommand{\lvertex}[3]{\node[vertex, label=left:{\small #1}] (#1) at (#2,#3)  {}; }
\newcommand{\rvertex}[3]{\node[vertex, label=right:{\small #1}] (#1) at (#2,#3)  {}; }
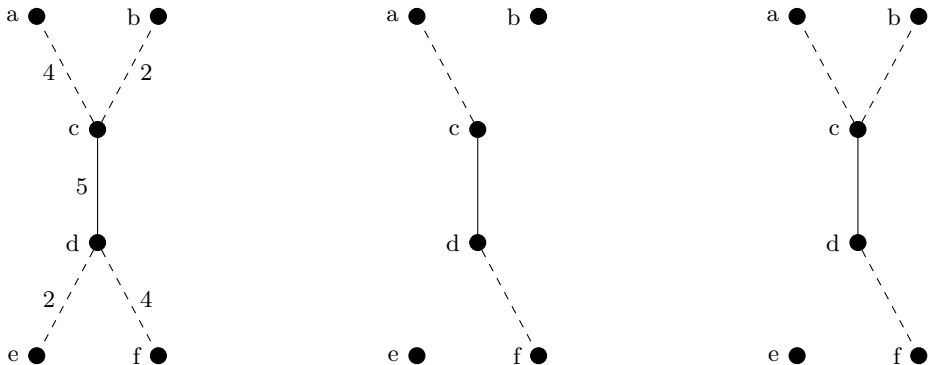
\begin{figure}[ht]
  \begin{center}
  \begin{tikzpicture}
    \lvertex{c}{0}{0}
    \lvertex{f}{0.8}{-3.0}
    \lvertex{e}{-0.8}{-3.0}
    \lvertex{d}{0}{-1.5}
    \lvertex{b}{0.8}{1.5}
    \lvertex{a}{-0.8}{1.5}

    \draw (c) edge node[left] {\small $5$} (d) edge[dashed]  node[right=1pt] {\small $2$} (b) edge[dashed] node[left=1pt] {\small $4$} (a);
    \draw (d) edge[dashed] node[right=1pt] {\small 4} (f) edge[dashed] node[left=1pt] {\small 2} (e);

    \begin{scope}[xshift=5cm]
    \lvertex{c}{0}{0}
    \lvertex{f}{0.8}{-3.0}
    \lvertex{e}{-0.8}{-3.0}
    \lvertex{d}{0}{-1.5}
    \lvertex{b}{0.8}{1.5}
    \lvertex{a}{-0.8}{1.5}

    \draw (c) edge  (d)  edge[dashed]   (a);
    \draw (d) edge[dashed]  (f);
    \end{scope}
    \begin{scope}[xshift=10cm]
    \lvertex{c}{0}{0}
    \lvertex{f}{0.8}{-3.0}
    \lvertex{e}{-0.8}{-3.0}
    \lvertex{d}{0}{-1.5}
    \lvertex{b}{0.8}{1.5}
    \lvertex{a}{-0.8}{1.5}

    \draw (c) edge  (d) edge[dashed]   (b) edge[dashed]   (a);
    \draw (d) edge[dashed]  (f);
    \end{scope}
  \end{tikzpicture}
\end{center}
  \caption{A simple illustration of the filtering technique.}
  \label{fig:introex}
\end{figure}

While the implementation of the basic idea is simple in the above case, there are several challenges in general. Perhaps the most obvious one is that,  for weighted matchings, $M$ may be a perfect matching but still far from optimal. And a perfect matching obviously has no unweighted augmenting paths! On a very high level, we overcome this issue  by dropping edges in $M$ while making sure to set the guessed lower bounds (the $\tau$'s) so as to guarantee that any unweighted augmenting path is also a weighted augmenting path (even when taking the dropped edges into account). 

In what follows, we describe in more detail the implementation of the above basic idea.  We start with the simpler case, single-pass streaming with random edge arrivals, where we look only for augmenting paths of length $3$. We then describe the technically more involved multi-pass streaming and MPC algorithms that consider long augmenting paths and cycles.

\subsubsection{Single-pass Streaming with Random Edge Arrivals}
In contrast to unweighted graphs where the basic greedy algorithm gives a $(\nicefrac{1}{2})$-approximation, it was only very recently that  a $(\nicefrac{1}{2}- \eps)$-approximation streaming algorithm was given for weighted matchings~\cite{Paz2017}. The algorithm of Paz and Schwartzman is based on the local ratio technique, which we now describe\footnote{The description of the local-ratio technique is adapted from a recent grant proposal submitted to the Swiss National Science Foundation by the last author.}. On an input graph $G=(V,E)$ with edge-weights $w: E \rightarrow \mathbb{R}$, the following simple local-ratio algorithm is known to return a $(\nicefrac{1}{2})$-approximate weighted matching: Initially, let $S = \emptyset$ and $\alpha_v = 0$ for all $v\in V$. For each $e=\{u,v\} \in E$  in an \emph{arbitrary} order:
\begin{itemize}
  \item if $\alpha_u + \alpha_v < w(e)$, add $e$ to $S$  and increase $\alpha_u$ and $\alpha_v$ by $w(e) - \alpha_u - \alpha_v$.
    %Let $S = \emptyset$ and $y_v = 0$ for all $v\in V$.
  %\item For $\{u,v\}\in E$:
\end{itemize}
Finally, obtain a matching $M$ by running the basic greedy algorithm  on the edges in $S$ in the reverse order (i.e., by starting with the edge  last added  to $S$).

Since the above algorithm returns a $(\nicefrac{1}{2})$-approximate matching irrespective of the order in which the edges are considered (in the for loop), it may appear immediate to use it in the streaming setting. The issue is that, if the edges arrive in an adversarial order, we may add \emph{all} the edges to $S$. For dense graphs,  this would lead to a  memory consumption of $\Omega(n^2)$ instead of the wanted memory usage $O(n \poly(\log n))$ which is (roughly) linear in the output size. The main technical challenge in~\cite{Paz2017} is to limit the number of edges added to $S$; this is why that algorithm obtains a $(\nicefrac{1}{2} - \eps)$-approximation, for any $\eps >0$, instead of a $(\nicefrac{1}{2})$-approximation.

McGregor and Vorotnikova observed that the technical issue in~\cite{Paz2017} disappears if we assume that edges arrive in a uniformly random order\footnote{Sofya Vorotnikova presented this result in the workshop ``Communication Complexity and Applications, II (17w5147)'' at the Banff International Research Station held in March 2017.}. Indeed, we can then use basic probabilistic techniques (see, e.g., the ``hiring problem'' in~\cite{CLRS}) to show that the expected (over the random arrival order) number of edges added to $S$ is $O(n\log n)$. Even better, here we show that, in expectation, the following adaptation still adds only $O(n\log n)$ edges to $S$: update the vertex potentials (the $\alpha_v$'s) only for, say, $1\%$ of the stream and then, in the remaining $99\%$ of the stream, add \emph{all} edges $\{u,v\}$ for which $\alpha_u + \alpha_v < w(\{u,v\})$ to $S$ (without updating the vertex potentials). This adaptation allows us to prove the following structural result:
\begin{center}
\begin{minipage}{0.95\textwidth}
\begin{mdframed}[hidealllines=true, backgroundcolor=gray!15]
  In a random-edge-arrival stream, either the local-ratio algorithm  already obtains a (close) to $(\nicefrac{1}{2})$-approximate matching $M$ after seeing a small fraction of the stream (think $1\%$), or the set $S$ (in the adaptation that freezes vertex potentials) contains a better than $(\nicefrac{1}{2})$-approximation in the end of the stream. 
\end{mdframed}
\end{minipage}
\end{center}
The above allows us to concentrate on the case when we have a (close) to $(\nicefrac{1}{2})$-approximate matching $M_0$ after seeing only $1\%$ of the stream. We can thus use the remaining $99\%$ to find enough augmenting paths to improve upon the initial $(\nicefrac{1}{2})$-approximation.
It is here that our \emph{filtering} technique is used to reduce the task of finding weighted augmenting paths to unweighted ones. By~\cref{fact:shortaug}, it is sufficient to consider very short augmentations  to improve upon an approximation guarantee of $\nicefrac{1}{2}$. Specifically, the considered augmentations are of two types:
\begin{enumerate}
  \item Those consisting of a single  edge  $\{u,v\}$ to add satisfying $w(\{u,v\}) > w(M_0(u)) + w(M_0(v))$, where $w(M_0(x))$ denotes the weight of the edge of $M_0$ incident to vertex $x$ (and $0$ if no such edge exists)\footnote{To make sure that the weight of the matching increases significantly by an augmentation, the strict inequality needs to be satisfied with a slack. We avoid this technicality  in the overview.}.  
  \item Those consisting of two new edges $o_1$ and $o_2$ that form a path or a cycle $(e_1, o_1, e_2, o_2, e_3)$ with at most three edges $e_1, e_2, e_3\in M_{0}$ and $w(o_1) + w(o_2) > w(e_1) + w(e_2) + w(e_3)$, i.e.,  adding $o_1, o_2$ and removing $e_1, e_2, e_3$ increases the weight of the matching.
\end{enumerate}
For concreteness, consider the graph in~\cref{fig:streaming}. The edges in $M_0$ are solid and dashed edges are yet to arrive in the stream. An example of the first type of augmentations is to add $\{e,h\}$ (and remove $\{e,f\}$ and $\{g,h\}$) which results in a gain because $w(\{e,h\}) = 2 > 1 + 0 = w(M_0(e)) + w(M_0(h))$.  Two examples of the second type of augmentations are the path  $(\{b,a\}, \{a, d\}, \{d, c\}, \{c,f\}, \{f,e\})$ and the cycle $(\{e,f\}, \{f,h\}, \{h,g\}, \{g,e\})$.  
\newcommand{\avertex}[3]{\node[vertex, label=above:{\small #1}] (#1) at (#2,#3)  {}; }
\newcommand{\bvertex}[3]{\node[vertex, label=below:{\small #1}] (#1) at (#2,#3)  {}; }
\newcommand{\aevertex}[4]{\node[vertex, label=above:{\small #1}] (#2) at (#3,#4)  {}; }
\newcommand{\bevertex}[4]{\node[vertex, label=below:{\small #1}] (#2) at (#3,#4)  {}; }
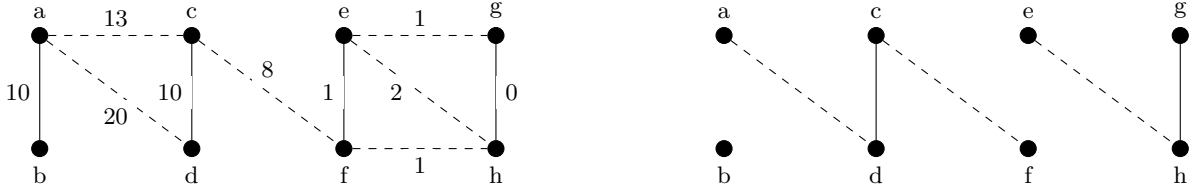
\begin{figure}[ht]
  \begin{center}
  \begin{tikzpicture}
    \avertex{a}{0}{0}
    \bvertex{b}{0}{-1.5}
    \avertex{c}{2}{0}
    \bvertex{d}{2}{-1.5}
    \avertex{e}{4}{0}
    \bvertex{f}{4}{-1.5}
    \avertex{g}{6}{0}
    \bvertex{h}{6}{-1.5}
 %   \avertex{i}{8}{0}
  %  \bvertex{j}{8}{-1.5}

    \draw (a) edge node[left,fill=white] {\small $10$} (b) edge[dashed]  node[below=3pt,fill=white] {\small $20$} (d) edge[dashed]  node[above,fill=white] {\small $13$} (c); 
    \draw (c) edge node[left,fill=white] {\small $10$} (d) edge[dashed]  node[above=2pt,fill=white] {\small $8$} (f);
    \draw (e) edge node[left,fill=white] {\small $1$} (f) edge[dashed] node[above,fill=white] {\small $1$} (g)edge[dashed] node[left=3pt,fill=white] {\small $2$} (h);
    \draw (f) edge[dashed] node[below,fill=white] {\small $1$} (h);
    \draw (g) edge node[right,fill=white] {\small $0$} (h);
%    \draw (i) edge node[fill=white] {\small $0$} (j);
    \begin{scope}[xshift=9cm]
      \avertex{a}{0}{0}
      \bvertex{b}{0}{-1.5}
      \avertex{c}{2}{0}
      \bvertex{d}{2}{-1.5}
      \avertex{e}{4}{0}
      \bvertex{f}{4}{-1.5}
      \avertex{g}{6}{0}
      \bvertex{h}{6}{-1.5}
   %   \avertex{i}{8}{0}
    %  \bvertex{j}{8}{-1.5}

      \draw (a)  edge[dashed]   (d); 
      \draw (c) edge  (d) edge[dashed]   (f);
      \draw (e) edge[dashed]  (h);
      \draw (g) edge  (h);
  %    \draw (i) edge node[fill=white] {\small $0$} (j);
    \end{scope}
  \end{tikzpicture}
\end{center}

\caption{On the left, an example of a weighted graph with matching $M_0$ (solid edges) is shown. On the right, the unweighted graph obtained in the filtering step with $M'_0 = \{\{c,d\}, \{g,h\}\}$ is shown.}
  \label{fig:streaming}
\end{figure}

The augmentations of the first type are easy to find in a greedy manner. 
%: those edges are  included in the set $S$ in our adaptation of the  local-ratio algorithm that freezes the vertex potentials. 
For the second type, we now describe how to use our filtering technique to reduce the problem to that of finding length three \emph{unweighted} augmenting paths.  Let \unwaugpath be a streaming algorithm for finding such unweighted augmenting paths. 
%Indeed, although they are leng augmenting paths, we call them 
%$3$-augmentations because of this reduction. we reduce the problem of finding those to
%the problem of finding length three unweighted augmenting paths. that we implement as follows in this setting.  
We first initialize \unwaugpath with a (random) matching $M'_0$ obtained by including each edge in $M_0$ with probability $\nicefrac{1}{2}$. As we explain shortly, $M_0'$ corresponds to the edges $e_2$ from the second type of augmenting paths.
Then, at the arrival of an edge  $\{u,v\}$, it is forwarded as an unweighted edge to \unwaugpath if 
\begin{align*}
  w(\{u,v\}) > \tau_u + \tau_v,\quad  \mbox{where} \qquad \tau_x = \begin{cases} w(M_0(x))/2 & \mbox{if $x$ is incident to an edge in $M'_0$,} \\ w(M_0(x)) & \mbox{otherwise.}
  \end{cases}
\end{align*}
For an example of the  forwarded edges for a specific $M'_0$, see the right part of~\cref{fig:streaming}. 

Note that the $\tau$-values are set so that any augmenting path found by \unwaugpath will also improve the matching in the weighted graph\footnote{We remark that there may be short augmentations that are beneficial in the weighted sense that are never present in the graph forwarded to \unwaugpath regardless of the choice of $M'_0$. An example would be $\{e_1, o_1, e_2, o_2, e_3\}$ with $w(e_1) = w(e_2) = w(e_3) = 10$ and $w(o_1) = 20, w(o_2) = 14$. In this case, $o_2$ is not forwarded to \unwaugpath due to the filtering if $e_2 \in M'_0, e_1, e_3 \not \in M'_0$; and, in the other choices of $M'_0$, $\{o_1, e_2, o_2\}$ is not a length three unweighted augmenting path.  However, as we prove in~\cref{sec:randarrival}, those augmentations are safe to ignore in our goal to beat the approximation guarantee of $\nicefrac{1}{2}$.} . Indeed, suppose that \unwaugpath finds the length three augmenting path $\{o_1, e_2, o_2\}$ where $e_2 \in M'_0$. Let $e_1$ and $e_3$ be the other edges in $M_0$ incident to $o_1$ and $o_2$ (if they exist). Then, by the selection of the $\tau$-values, we have
\begin{align*}
  w(o_1) + w(o_2) > (w(e_1) + w(e_2)/2) + (w(e_2)/2) + w(e_3)) = w(e_1) + w(e_2) + w(e_3)\,,
\end{align*}
as required. Hence, the $\tau$-values are set so as to guarantee that the augmenting paths will improve the weighted matching if applied.  

The reason for the random selection of $M'_0$ is to make sure that any such  beneficial weighted augmenting path $\{e_1, o_1, e_2, o_2, e_3\}$ is present as an unweighted augmenting path $\{o_1, e_2, o_2\}$ in the graph given to \unwaugpath with probability at least $1/8$. This guarantees that there will be (in expectation) many length three unweighted augmenting paths corresponding to weighted augmentations (assuming the initial matching $M_0$ is no better than $(\nicefrac{1}{2})$-approximate).  

This completes the high level description of our single-pass streaming algorithm except for the following omission: all unweighted augmenting paths are equally beneficial  while their weighted contributions may differ drastically. This may result in a situation where \unwaugpath returns a constant-fraction of the unweighted augmenting paths that have little value in the weighted graph.  The solution is simple: we partition $M'_0$ into  weight classes by geometric grouping,  run \unwaugpath for each weight class in parallel, and then select vertex-disjoint augmenting paths in a greedy fashion starting with the augmenting paths in the largest weight class. This ensures that many unweighted augmenting paths also translates into a significant improvement of the weighted matching. The formal and complete description of these techniques are given in~\cref{sec:randarrival}.

%%%% \begin{itemize}
%%%%   \item[(i)] exactly one of $o$'s end points, say $u$, is incident to an edge $e$ in $M'_0$ and 
%%%%   \item[(ii)]  $w(o) > w(M_0(u))/2 + w(M_0(v))$.
%%%% \end{itemize}

\subsubsection{Multi-pass streaming and MPC}
\label{sec:intro-mpcandstreaming}

In our approach for single-pass streaming, it was crucial to have an algorithm (local-ratio with frozen vertex potentials) that allowed us to reduce the problem to that of finding augmenting paths to a  matching $M_0$ that is already (close) to $\nicefrac{1}{2}$-approximate. 
This is because, in a single-pass streaming setting, we can find a limited amount of augmenting paths leading to a limited improvement over the initial matching. 

In multi-pass streaming and MPC, the setting is somewhat different. 
On the one hand, the above difficulty disappears because we can repeatedly find augmentations. 
In fact, we can even start with the empty matching.   
On the other hand, we now aim for the much stronger approximation guarantee of $(1-\eps)$ for any fixed $\eps >0$. 
This results in a more complex filtering step as  we now need to find augmenting paths and cycles of arbitrary length (depending on $\eps$). We remark that the challenge of finding long augmenting cycles is one of the difficulties that appears in the weighted case where previous techniques do not apply~\cite{McGregor2005,ahnguha}. We overcome this and other challenges by giving a general reduction to the unweighted matching problem, which can be informally stated as follows:
\begin{center}
\begin{minipage}{0.9\textwidth}
\begin{mdframed}[hidealllines=true, backgroundcolor=gray!15]
  Let $M$ be the current matching and $\optWM$ be an optimal matching of maximum weight.
  %and \UnwBipMatching be an $(1-\delta(\eps))$-approximation algorithm for the unweighted matching problem on bipartite graphs. 
  If $w(M) < (1-\eps) \optwgt$ then an $(1-\delta(\eps))$-approximation algorithm for the unweighted matching problem on bipartite graphs can be used to find a collection of vertex-disjoint augmentations that in expectation increases the weight of $M$ by $\Omega_\eps(\optwgt)$.
\end{mdframed}
\end{minipage}
\end{center}
The reduction itself is efficient and can easily be implemented both in the multi-pass streaming and MPC models by incurring only a constant overhead in the complexity. Using the best-known approximation algorithms for the unweighted matching problem on bipartite graphs in these models then yields \cref{thm:mpc} by repeating the above $f(\eps)$ times after starting with the empty matching $M= \emptyset$. %as follows: Start with the empty matching $M= \emptyset$ and repeat the above reduction $f(\eps)$ times.

We now present the main ideas of our reduction (the formal proof is given in~\cref{sec:mpctot}). We start with a structural statement for weighted matchings similar to~\cref{fact:shortaug}: 

\begin{center}
\begin{minipage}{0.9\textwidth}
\begin{mdframed}[hidealllines=true, backgroundcolor=gray!15]
Suppose the current matching $M$  satisfies $w(M) \leq (1-\eps)\optwgt$. Then there must exist a collection $\cC$ of short (each consisting of $O(1/\eps)$ edges) vertex-disjoint augmenting paths and cycles with total gain $\Omega(\eps^2) \cdot \optwgt$. Moreover, each augmentation $C\in \cC$ has gain at least $\Omega(\eps^2 w(C))$, i.e., proportional to its total weight.
\end{mdframed}
\end{minipage}
\end{center}
Our goal now is to find a large fraction of these short weighted augmentations. For this, we first reduce the problem to that of finding such augmentations $C$ with $w(C) \approx W$ for some fixed $W$. This is similar to the concept of weight classes mentioned in the previous section and corresponds to the notion of augmentation classes in~\cref{sec:mpctot}. Note that, by standard geometric grouping, we can reduce the number of choices of $W$ to be at most logarithmic. We can thus afford to run our algorithm for all choices of $W$ in parallel and then greedily select the augmentations starting with those of the highest weight augmentation class.
%as done in the single-pass random edge-arrival setting. 

\newcommand{\layerellipse}[2]{\draw (#1,#2) ellipse (0.65cm and 1.5cm);} 
\newcommand{\layerellipsebig}[2]{\draw (#1,#2) ellipse (0.9cm and 1.7cm);} 
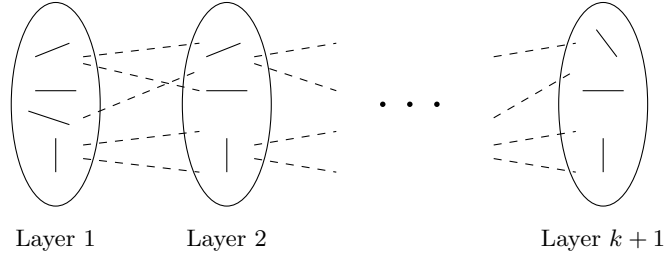
\begin{figure}[t]
  \begin{center}
    \begin{tikzpicture}[scale=0.9]
    \layerellipse{0}{0}
    \layerellipse{2.5}{0}
    \layerellipse{8}{0}
    \node at (0, -2) {\small Layer $1$};
    \node at (2.5, -2) {\small Layer $2$};
    \node at (8, -2) {\small Layer $k+1$};

    %%% FIRST LAYER %%%
    \draw (-0.3,0.2) -- (0.3,0.2);
    \draw (-0.3,0.7) -- (0.2,0.9);
    \draw (-0.4,-0.1) -- (0.2,-0.3);
    \draw (0.0,-1.0) -- (0.0,-0.5);

    %%% SECOND LAYER %%%
    \draw (2.2,0.2) -- (2.8,0.2);
    \draw (2.2,0.7) -- (2.7,0.9);
    \draw (2.5,-1.0) -- (2.5,-0.5);

    %%% LAST LAYER %%%
    \draw (7.7,0.2) -- (8.3,0.2);
    \draw (8.2,0.7) -- (7.9,1.1);
    \draw (8,-1.0) -- (8,-0.5);

    %%% BETWEEN 1st and 2nd LAYER %%%
    \begin{scope}[yshift=0.2cm]
      \draw[dashed] (0.4, 0.5) -- (2.1, 0.7);
      \draw[dashed] (0.4, 0.4) -- (2.1, 0.0);
      \draw[dashed] (0.4, -0.4) -- (2.1, 0.3);
      \draw[dashed] (0.4, -0.8) -- (2.1, -0.6);
      \draw[dashed] (0.4, -1.0) -- (2.1, -1.2);
    \end{scope}

    %%% From 2nd LAYER %%%
    \begin{scope}[yshift=0.2cm]
      \draw[dashed] (2.9, 0.5) --  (4.1, 0.7);
      \draw[dashed] (2.9, 0.4) --  (4.1, 0.0);
%     \draw[dashed] (2.9, -0.4) -- (4.1, 0.3);
      \draw[dashed] (2.9, -0.8) -- (4.1, -0.6);
      \draw[dashed] (2.9, -1.0) -- (4.1, -1.2);
    \end{scope}

    %%% DOTS %%%
    \node at (5.25, 0) {\Huge \ldots};

    %%% To last LAYER %%%
    \begin{scope}[yshift=0.2cm]
      \draw[dashed] (6.4, 0.5) --  (7.6, 0.7);
     %\draw[dashed] (6.4, 0.4) --  (7.6, 0.0);
      \draw[dashed] (6.4, -0.4) -- (7.6, 0.3);
      \draw[dashed] (6.4, -0.8) -- (7.6, -0.6);
      \draw[dashed] (6.4, -1.0) -- (7.6, -1.2);
    \end{scope}
  \end{tikzpicture}
\end{center}

\caption{The layered graph consisting of $k+1$ layers. The solid edges inside the layers are subsets of $M$ and the dashed edges between layers are subsets of $E\setminus M$.}
  \label{fig:layerintro}
\end{figure}

Now, for each augmentation class (i.e., for each choice of $W$), we give
a reduction from finding weighted augmentations to finding unweighted ones by
constructing a set of tailored graphs. This construction resembles some of the ideas used in the construction of \cite{McGregor2005}, but they are not the same. The intuition behind our construction is as follows. Suppose
that, for a fixed $W$,  we aim to find augmenting paths of length $2k+1$ in the input graph
$G=(V,E)$. Then, as depicted in~\cref{fig:layerintro}, we construct a new layered graph $\Layered$ consisting of
$k+1$ layers of vertices,
(each layer is a copy of $V$), % Sagar: I prefer this line uncommented that was commented earlier.  This will help understand the figure that comes later.
where the edge set
of each layer consists of a subset of the edges in the current matching $M$ and
the edges between layers are subsets of $E\setminus M$.  The construction of $\Layered$ is
so that if we consider an alternating path $ C= (e_1, o_1,e_2, o_2, \ldots, e_k,
  o_k, e_{k+1})$  in $\Layered$ where $e_i \in M$ is an edge in layer $i$ and $o_i$ is an edge between layer $i$ and $i+1$, then, assuming they all correspond to distinct edges in $G$, we can augment  $M$ with $C$ to obtain the new matching $M \Delta C$.  Moreover,  the augmentation improves the matching, i.e., satisfies $w(M\Delta C) > w(M)$, if
  \begin{align}
    \sum_{i=1}^k w(o_i) > \sum_{i=1}^{k+1} w(e_i)\,.
    \label{eq:benefitintro}
  \end{align}

  To ensure that any alternating path in the unweighted graph $\Layered$ satisfies~\eqref{eq:benefitintro} we use our filtering technique. 
  %Let $\Tau = \{ k\layerepspower W: k =0, 1, \ldots, 1/\layerepspower\}$ be the set of multiples of $\layerepspower$ between $0$ and $W$. 
  For each layer $i=1, \ldots, k+1$, we have a parameter $\tau_i^A$ that filters the edges in that layer: we keep an edge $e\in M$ in layer $i$ only if $w(e)$ rounded up to the closest multiple of $\layerepspower W$ equals $\tau_i^A W$. Similarly, we have a parameter $\tau_i^B$ for each $i=1,\ldots, k$, and we keep an edge $e\in E\setminus M$ between layer $i$ and $i+1$ only if $w(e)$ rounded down to the closest  multiple of  $\layerepspower W$ equals $\tau_i^B W$.  Now by considering only those $\tau$-values satisfying $\sum \tau_{i}^B > \sum \tau_i^A$, we ensure that any augmenting path that is found improves the matching, i.e., \eqref{eq:benefitintro} holds. Moreover, the rounding of edge-weights in the filtering step still keeps large (by weight) fraction of the augmentations in the original graph as the rounding error, which is less than $\layerepspower W$ for each edge, is very small compared to the length and total gain of the structural augmentations that we are looking for.   It is thus enough to find the augmentations corresponding to each fixation of $k$ and $\tau$-values. To bound the number of choices, note that we may assume that each $\tau$-value is such that $\tau\cdot W$ is a multiple of $\layerepspower W$ between $0$ and $W$. Hence, as we need to consider augmentations of length $O(1/\eps)$ only, we have, for a fixed $\eps> 0$ and $W$, that the total number of choices of $k$ and $\tau$-values is a constant. They  can thus all be considered in parallel. 
For each of these choices, we use the approximation algorithm for unweighted matchings to find a  $(1 - \delta(\eps))$-approximate maximum unweighted matching in the corresponding layered graph and take the symmetric difference with the initial matched edges to find the desired unweighted augmentations.  These augmentations are then translated back to weighted augmentations in the original graph. 

Note that, unlike McGregor's layered graphs, our layered graphs allow edges (both matched and unmatched) to be repeated in different layers, which is crucial in identifying weighted augmenting cycles. Furthermore, edges in each layer are filtered with respect to a given edge-weight arrangement, that ensures that the augmenting paths in our layered graphs correspond to weighted augmentations with positive gain. These differences result from the different purposes of the two constructions: McGregor's construction aims to find unweighted augmenting paths efficiently, whereas our purpose is to reduce weighted augmentations to unweighted ones.

While, on a high level, this completes the description of our reduction, there are many interesting technical challenges to overcome. In the remaining part of this overview, we highlight two of these challenges.

\paragraph{Translating augmenting paths in layered graph to the original graph} From our high level description of the layered graph $\Layered$, there is no guarantee that an augmenting path in it corresponds to an augmentation with a positive gain in the original graph $G$. First, there is no reason that an augmenting path in $\Layered$ visits the layers from left-to-right as intended. In the formal definition of layered graphs (see \cref{sec:finding-short-augmentations}), we take care of this and make sure\footnote{To be completely accurate, the edges $e_1$ and $e_{k+1}$ may not appear in the alternating path: $e_1$ does not appear if the vertex incident to $o_1$ in the first layer is not incident to a filtered edge in $M$; the case of $e_{k+1}$ is analogous.} that any unweighted augmenting path in $\Layered$ corresponds to an alternating path of the form $(e_1, o_1,e_2, o_2, \ldots, e_k,
  o_k, e_{k+1})$, where $e_i \in M$ is an edge in layer $i$ and $o_i$ is an edge between layer $i$ and $i+1$. Intuitively, such an alternating path can be made an unweighted augmenting path by discarding the matching edges of the first and last layers. However, a  second and more challenging issue is that such an alternating path (going from the left to the right layer)  may contain repeated edges and thus do not correspond to an augmentation in $G$. An example of this phenomena is as follows:

\newcommand{\llvertex}[4]{\node[vertex, label=left:{\small #1}] (#2) at (#3,#4)  {}; }
\newcommand{\rrvertex}[4]{\node[vertex, label=right:{\small #1}] (#2) at (#3,#4)  {}; }
  \begin{center}
  \begin{tikzpicture}
    \lvertex{a}{0}{0}
    \rvertex{b}{1.5}{0}
    \lvertex{c}{0}{-1.2}
    \rvertex{d}{1.5}{-1.2}
    \lvertex{e}{0}{-2.4}
    \rvertex{f}{1.5}{-2.4}
 %   \avertex{i}{8}{0}
  %  \bvertex{j}{8}{-1.5}

    \draw (a) edge node[above,fill=white] {\small $1$} (b); % edge[dashed]  node[below=3pt,fill=white] {\small $20$} (d) edge[dashed]  node[above,fill=white] {\small $13$} (c); 
    \draw (b) edge[dashed] node[above left = -1pt and -1pt, fill=white] {\small $2$} (c) edge[dashed] node[right, fill=white] {\small $2$} (d); 
    \draw (c) edge node[above,fill=white] {\small $1$} (d);
    \draw (d) edge[dashed] node[above left = -1pt and -1pt, fill=white] {\small $2$} (e); 
    \draw (e) edge node[above,fill=white] {\small $1$} (f);

    \begin{scope}[xshift=5cm, yshift=-0.2cm]
      \layerellipsebig{0.5}{-1}
      \layerellipsebig{3}{-1}
      \layerellipsebig{5.5}{-1}
      \llvertex{}{a1}{0}{0}
      \rrvertex{}{b1}{1}{0}
      \llvertex{}{c1}{0}{-1}
      \rrvertex{}{d1}{1}{-1}
      \llvertex{}{e1}{0}{-2}
      \rrvertex{}{f1}{1}{-2}

      \llvertex{}{a2}{2.5}{0}
      \rrvertex{}{b2}{3.5}{0}
      \llvertex{}{c2}{2.5}{-1}
      \rrvertex{}{d2}{3.5}{-1}
      \llvertex{}{e2}{2.5}{-2}
      \rrvertex{}{f2}{3.5}{-2}

      \llvertex{}{a3}{5}{0}
      \rrvertex{}{b3}{6}{0}
      \llvertex{}{c3}{5}{-1}
      \rrvertex{}{d3}{6}{-1}
      \llvertex{}{e3}{5}{-2}
      \rrvertex{}{f3}{6}{-2}

      \draw (a1) edge[ultra thick]  (b1);
      \draw (a2) edge  (b2);
      \draw (a3) edge[ultra thick]  (b3);
      \draw (b1) edge[dashed, ultra thick]  (c2) edge[dashed]  (d2); 
      \draw (b2) edge[dashed]  (c1) edge[dashed]  (d1); 
      \draw (b2) edge[dashed]  (c3) edge[dashed]  (d3); 
      \draw (b3) edge[dashed]  (c2) edge[dashed, ultra thick]  (d2); 
      \draw (c1) edge  (d1);
      \draw (c2) edge[ultra thick]  (d2);
      \draw (c3) edge  (d3);
      \draw (d2) edge[dashed]  (e1); 
      \draw (d1) edge[dashed]  (e2); 
      \draw (d3) edge[dashed]  (e2); 
      \draw (d2) edge[dashed]  (e3); 
      \draw (e1) edge  (f1);
      \draw (e2) edge  (f2);
      \draw (e3) edge  (f3);
    \end{scope}
%    \draw (i) edge node[fill=white] {\small $0$} (j);
  \end{tikzpicture}
\end{center}
Here, we depict the weighted graph on the left and the ``incorrect'' layered graph to the right with $\tau_1^AW = \tau_2^AW = \tau_3^AW =1$ and $\tau_B^1W = \tau_B^2W = 2$. The weighted graph has an augmentation that adds $\{b,c\}, \{d,e\}$ and removes $\{a,b\}, \{c,d\}, \{e,f\}$ and improves the weight of the matching by one. This augmentation is also present in the layered graph. However, an equally good augmentation in that graph from an unweighted perspective corresponds to the alternating path depicted in bold. In the original graph the bold edge set corresponds to the non-simple path $a-b - c-d-b-a$. Such a non-simple path clearly does not correspond to an augmentation and, even worse, there is no augmentation with a positive gain in the support $\{a,b\}, \{b,c\}, \{c,d\}, \{d, b\}$ of the considered path.   

Our main idea to overcome this issue is as follows. We first select a random bipartition $L$ and $R$ of the vertex set of $G$. Then between two layers $i$ and $i+1$, we keep only those edges  that go from an $R$-vertex in layer $i$ to an $L$-vertex in layer $i+1$. We emphasize that the edges going from an L-vertex to an R-vertex between two layers are not kept.
For example, if we let $L= \{a,c,e\}$ and $R=\{b, d, f\}$ in the considered example then the layered graph (with the same $\tau$-values) becomes:
\begin{center}
  \begin{tikzpicture}

    \begin{scope}[xshift=5cm, yshift=-0.2cm]
      \layerellipsebig{0.5}{-1}
      \layerellipsebig{3}{-1}
      \layerellipsebig{5.5}{-1}
      \llvertex{}{a1}{0}{0}
      \rrvertex{}{b1}{1}{0}
      \llvertex{}{c1}{0}{-1}
      \rrvertex{}{d1}{1}{-1}
      \llvertex{}{e1}{0}{-2}
      \rrvertex{}{f1}{1}{-2}

      \llvertex{}{a2}{2.5}{0}
      \rrvertex{}{b2}{3.5}{0}
      \llvertex{}{c2}{2.5}{-1}
      \rrvertex{}{d2}{3.5}{-1}
      \llvertex{}{e2}{2.5}{-2}
      \rrvertex{}{f2}{3.5}{-2}

      \llvertex{}{a3}{5}{0}
      \rrvertex{}{b3}{6}{0}
      \llvertex{}{c3}{5}{-1}
      \rrvertex{}{d3}{6}{-1}
      \llvertex{}{e3}{5}{-2}
      \rrvertex{}{f3}{6}{-2}

      \draw (a1) edge[]  (b1);
      \draw (a2) edge  (b2);
      \draw (a3) edge[]  (b3);
      \draw (b1) edge[dashed]  (c2); 
      \draw (b2) edge[dashed]  (c3); 
      \draw (c1) edge  (d1);
      \draw (c2) edge  (d2);
      \draw (c3) edge  (d3);
      \draw (d1) edge[dashed]  (e2); 
      \draw (d2) edge[dashed]  (e3); 
      \draw (e1) edge  (f1);
      \draw (e2) edge  (f2);
      \draw (e3) edge  (f3);
    \end{scope}
%    \draw (i) edge node[fill=white] {\small $0$} (j);
  \end{tikzpicture}
\end{center}
In this example, the remaining alternating path that visits all layers (in the formal proof we further refine the layered graph to make sure that these are the only paths that are considered) corresponds to the   augmentation in $G$. However, in general, an alternating path may still not correspond to a simple path and an augmentation in $G$ since it may contain repetitions. However, the bipartition and the refinement of the layered graph can be seen to introduce an ``orientation'' of the edges in $G$. This together with standard Eulerian techniques of directed graphs allow us to prove that  any alternating path in the layered graph  can be decomposed into a collection of alternating even-length cycles and an alternating path in $G$, one of which is also augmenting. Finally, let us remark that the idea to consider a bipartition $L$ and $R$ of the vertex set of $G$ and to allow only those edges that are from an $R$-vertex to an $L$-vertex between consecutive layers has the additional benefit  that the layered graph becomes bipartite. This is the reason that our reduction is from weighted matchings in general graphs to unweighted matchings in bipartite graphs. % to that of approximating the unweighted matching problem in bipartite graphs.

\paragraph{Finding augmenting cycles}
In the unweighted setting, matching algorithms do not have to consider cycles because alternating cycles cannot augment an existing matching.
In contrast, algorithms for the weighted setting (at least the ones that try to iteratively improve an initial matching) have to somehow deal with augmenting cycles; weighted graphs can have perfect (unweighted) matchings whose weights are not close to the optimal and that can be improved only through augmenting cycles. 
For example, consider a 4-cycle with edge weights  $(3, 4, 3, 4)$, where the edges of weight $3$ form an initial perfect matching of weight $6$, but the optimal matching consists of edges of weight $4$ and has a total weight of $8$. The only way to augment the weight here is to consider the whole cycle.
The crucial property of our reduction is its ability to transform not only weighted augmenting paths, but also weighted augmenting \emph{cycles} of the original graph into augmenting paths in the layered graphs. 

Before explaining our solution, let us take a closer look at the above $4$-cycle example.  Let the edges of the $4$-cycle be $(e_1, o_1, e_2, o_2)$ where $\{e_1, e_2\}$ is the current matching. Note that the cycle can  be represented as an alternating path $(e_1, o_2, e_2, o_2, e_1)$ in the layered graph using three layers (consisting of the three edges of the matching with $e_1$ repeated once). However, such a representation of the augmenting cycle cannot be captured by our filtering technique due to the constraint $\sum_{i}\tau_i^B > \sum_{i}\tau_i^A$ which ensures that any alternating path in the layered graph can be translated into a weighted augmentation. The reason being that for $(e_1, o_2, e_3, o_2, e_1)$ to be present in the layered graph we would need $\tau_i^A W= 3$ for $i=1,2,3$, and $\tau_i^B W = 4$ for $i=1,2$ which would contradict the above inequality.
This approach  is therefore not sufficient to find augmenting cycles and achieve a $(1 -\eps)$ approximation guarantee. 
Specifically, the issue is due to the fact that we account for the edge weight of $e_1$ twice in the filtering process, once for $o_1$ and once more for $o_2$.
To overcome this issue, consider the $4$-cycle with more general weights $2, 2 + \eps, 2, 2+\eps$, where taking $o_1, o_2$ in place of $e_1, e_2$ gives an $\eps/2$ fractional gain in weight. What we need is to make sure that, even if we account for the same edge $e_1$ (or $e_2$) twice, the alternating path we get in the layered graph (``corresponding'' to the cycle) is still gainful. 
For this, we blow-up the cycle length by repeating the same cycle $O(1/\eps)$ times.
I.e., we consider the cycle $$\underbrace{e_1, o_1, e_2, o_2}_{\text{instance $1$}}, \underbrace{e_1, o_1, e_2, o_2}_{\text{instance $2$}}, \dots,  \underbrace{e_1, o_1, e_2, o_2}_{\text{instance $c/\eps$}}, e_1.$$
Since we have repeated the $o_i$ edges many times, their gains add up so that it can account for the weight of considering $e_1$ one additional time. The considered cycle of length $4$  is thus present as a ``repeated'' alternating path in the layered graph (with the appropriate $\tau$-values and bipartition) consisting of $O(1/\eps)$ layers. In general, to make sure that we can find augmenting cycles of length $O(1/\eps)$ we will consider the layered graph with up to $O(1/\eps^2)$ layers.

\subsection{Further Related Work}

 There is a large body of work devoted to (semi-)streaming
algorithms for the maximum matching problem. For unweighted graphs, the basic
greedy approach yields a $(\nicefrac{1}{2})$-approximation, and for weighted
graphs~\cite{Paz2017} recently gave
a $(\nicefrac{1}{2}-\eps)$-approximation based on the local ratio
technique. These are the best known algorithms that take a single pass over an
adversarially ordered stream. Better algorithms are known if the stream is randomly ordered or if the algorithm can take multiple passes through the stream. In  the random-edge-arrival case,~\cite{Konrad2012} first improved upon the approximation guarantee of $\nicefrac{1}{2}$ in the unweighted case. Our results give better guarantees in that setting and also applies to the weighted setting. When considering multi-pass algorithms,~\cite{McGregor2005} gave a $(1-\eps)$-approximation algorithm using $(1/\eps)^{O(1/\eps)}$ passes. Complementing this,~\cite{ahnguha}  gave a deterministic $(1-\eps)$-approximation algorithm using $O(\log(n) \poly(1/\eps))$ passes. 
%Interestingly, the techniques in McGregor's multi-pass streaming result~\cite{McGregor2005} inspired some parts of our MPC algorithm.
As for hardness results,~\cite{Kapralov:2013} showed that no algorithm can achieve
a better approximation guarantee than $(1-\nicefrac{1}{e})$ in the adversarial
single pass streaming setting. 

The study of algorithms for matchings in models of parallel computation dates back to the eighties. A seminal work of Luby~\cite{Luby86} shows how to construct a maximal independent set in $O(\log{n})$ PRAM rounds. When this algorithm is applied to the line graph of $G$, it outputs a maximal matching of $G$. Similar results, also in the context of PRAM, were obtained in~\cite{AlonBI86, II86, IsraeliS86}.

Perfect maximum matchings were also a subject of study in the context of PRAM. In~\cite{lovasz1979determinants} it is shown that the decision variant is in RNC. That implies that there is a PRAM algorithm that in $\poly \log{n}$ rounds decides whether a graph has a perfect matching or not. \cite{karp1986constructing} were the first to prove that constructing perfect matchings is also in RNC. In~\cite{mulmuley1987matching}  the same result was proved, and they also introduced the isolation lemma that had a great impact on many other problems.

In \cite{KarloffSV10,goodrich2011sorting} it was shown that it is often possible to simulate one PRAM in $O(1)$ MPC rounds with $O(n^{\alpha})$ memory per machine, for any constant $\alpha > 0$. This implies that the aforementioned PRAM results lead to $O(\log{n})$ MPC round complexity algorithms for computing maximal matchings. \cite{LattanziMSV11} developed an algorithm that computes maximal matchings in the MPC model in $O(1 / \delta)$ rounds when the memory per machine is $\Omega(n^{1 + \delta})$, for any constant $\delta > 0$. In the regime of $\tilde{O}(n)$ memory per machine, the algorithm given in \cite{LattanziMSV11} requires $\tilde{O}(\log{n})$ MPC rounds of computation. Another line of work focused on improving this round complexity. Namely, \cite{czumaj2018round} and \cite{assadi2017coresets, ghaffari2018improved} show how to compute a constant-factor approximation of maximum unweighted matching in $O((\log \log{n})^2)$ and $O(\log \log{n})$ MPC rounds, respectively, when the memory per machine is $\tilde{O}(n)$. As noted in \cite{czumaj2018round}, any $\Theta(1)$-approximation algorithm for maximum unweighted matchings can be turned into a $(\nicefrac{1}{2} - \eps)$-approximation algorithm for weighted matchings by using the approach described in Section~4 of \cite{Lotker:2015}. This transformation increases the round complexity by $O(1 / \eps)$.

In the regime of $n^{\delta}$ memory per machine, for any constant $\delta \in (0, 1)$, a recent work~\cite{brandt2018matching} shows how to find maximal matchings in $O((\log \log{n})^2)$ rounds for graphs of arboricity $\poly(\log{n})$. Also in this regime, \cite{ghaffari2018sparsifying} and \cite{onak2018round} provide algorithms for constructing maximal matchings for general graphs in $\tilde{O}(\sqrt{\log{n}})$ MPC rounds. The algorithm of \cite{ghaffari2018sparsifying} requires $O(m)$ and the algorithm of \cite{onak2018round} requires $O(m + n^{1+o(1)})$ total memory.

\section{Preliminaries}
\label{sec:prelim}
We formalize the streaming and the MPC model now.

\subsubsection*{The (semi)-streaming model}
The (semi-)streaming model for graph problems was introduced by Feigenbaum et al.~\cite{fgnbm}.  In this model, the edges of the input graph arrive one-by-one in the stream, and the algorithm is allowed to use $O(n\poly\log(n))$ memory at any time and may go over the stream one (single-pass) or more (multi-pass) times.  Note that $\Omega(n\log n)$ memory is needed just to store a valid matching.

 \subsubsection*{The MPC model}
% \label{sec:MPC-model}
The MPC model was introduced in~\cite{KarloffSV10} and refined in later work~\cite{goodrich2011sorting,BeameKS13,AndoniNOY14}. In this model, the computation is performed in synchronous rounds by $\machines$ machines. Each machine has $S$ bits of memory. At the beginning of a round the data, e.g., a graph, is partitioned across the machines with each machine receiving at most $S$ bits. During a round, each machine processes the received data locally. After the local computation on all the machines is over, each machine outputs messages of the total size being at most $S$ bits. The output of one round is used to guide the computation in the next round. In this model, each machine can send messages to any other machine, as long as at most $S$ bits are sent and received by each machine.

Let $G$ be the input graph. A natural assumption is that $S \cdot \machines \in \Omega(|G|)$, i.e., it is possible to partition the entire graph across the machines. We do not assume any structure on how the graph is partitioned across the machines before the computation begins. In our work, we assume that $S \cdot \machines \in \tilde{O}(|G|)$. Furthermore, we consider the regime in which the memory per machine is nearly-linear in the vertex set, i.e., $S \in \tilde{\Theta}(|V(G)|)$.

In the rest of this work, we show how to construct a $(1 - \eps)$-approximate maximum weighted matching. Edges that are in the matching will be appropriately tagged and potentially spread across multiple machines. These tags can be used to deliver all the matching edges to the same machine in $O(1)$ MPC rounds.

\paragraph{Computation vs.~communication complexity:}
In this model, the complexity is measured by the number of rounds needed to execute a given algorithm. Although the computation complexity is de-emphasized in the MPC model, we note that our algorithms run in nearly-linear time.

%\comments{The model is imprecise.  Where is the input stored before partitioning?  Partitioning involves some computation.  What kind of computation is allowed when an MPC algorithm partitions the input?  Why can it not just solve the problem within the limits of allowed computation?  What about randomness?  We need to allow public randomness if we are to stay within the total inter-machine communication bound of $O(n \polylog n)$ because we use $\Omega(|E|)$ random bits.}

%%% Local Variables:
%%% mode: latex
%%% TeX-master: "000-main_random_order"
%%% End:

\newcommand{\wei}{w(e_i)}
\newcommand{\weia}{w(e_{i+1})}
\newcommand{\weib}{w(e_{i+2})}
\newcommand{\woi}{w(o_i)}
\newcommand{\woia}{w(o_{i+1})}

\newcommand{\R}{{\mathbb{R}}}
\newcommand{\calM}{\mathcal{M}}
\newcommand{\E}{\EE}
\newcommand{\Aug}{\operatorname{Aug}}
\newcommand{\Ind}{\operatorname{\mathbb{I}}}
\newcommand{\mwm}{\textsc{MWM}\xspace}
\newcommand{\mcm}{\textsc{MCM}\xspace}
\newcommand{\augpaths}{\textsc{Wgt-Aug-Paths}\xspace}

\newcommand{\funinit}{\textsc{Initialize}\xspace}
\newcommand{\funfeededge}{\textsc{Feed-Edge}\xspace}
\newcommand{\funfinalize}{\textsc{Finalize}\xspace}
\newcommand{\ram}{\textsc{Rand-Arr-Matching}\xspace}

\section{Weighted Matching when Edges Arrive in a Random Order}
\label{sec:randarrival}

In this section, we present a $(\nicefrac{1}{2} + c)$\nobreakdash-approximation 
(semi-)streaming algorithm for the maximum weighted matching (\mwm) problem in the
random-edge-arrival setting, where $c > 0$ is an absolute constant,
thus proving \Cref{thm:streaming}. 
Our result computes a large weighted matching using unweighted augmentations.
In that spirit, we provide the following lemma that gives us the streaming
algorithm for unweighted augmentations.

\begin{restatable}{lemma}{lemunwblackbox}
\label{lem:unw-blackbox}
  There exists an unweighted streaming algorithm \unwaugpath with the
  following properties:
  \begin{enumerate}
    \item The algorithm is initialized with a matching $M$ and a parameter
    $\beta > 0$.  Afterwards, a set $E$ of edges is fed to the algorithm one
    edge at a time.
    \item Given that $M \cup E$ contains at least $\beta |M|$ vertex disjoint
    $3$-augmenting paths, the algorithm returns a set $\Aug$ of at least
    $(\beta^2/32)|M|$ vertex disjoint $3$-augmenting paths.  The algorithm uses
    space $O(|M|)$.
  \end{enumerate}
\end{restatable}
\begin{proof}
  Since this proof is based completely on the ideas of Kale and
  Tirodkar~\cite{KaleT17}, we give it in the appendix for completeness.
  See~\Cref{sec:proofs-missing-from}.
\end{proof}

% Later in this section, we present our main result of the section, which is a 
% semi-streaming algorithm for the maximum weight matching problem in the 
% random-edge-arrival setting that beats the $(1/2)$-approximate barrier by a 
% constant margin.
% As we explain later, this algorithm essentially follows the same method
% as our unweighted algorithm.
% The most important contribution of this section is the generic reduction 
% we give from the problem of finding short weighted augmenting paths to that of
% finding short unweighted augmenting paths, which is crucial in handling the
% weighted counterpart of the third case described above.
% Owing to our reduction, we use the same unweighted algorithm \unwaugpath as a 
% black-box for finding weighted $3$-augmenting paths.

We mentioned in the introduction that, for an
effective weighted-to-unweighted reduction in the streaming model, it is important to start with a
``good'' approximate matching so that we can augment it using $3$-augmentations
afterwards. We demonstrate these ideas on unweighted matchings first (\cref{sec:an-algor-unwe}), and show that they lead to an improved approximation ratio for both general and bipartite graphs. Later, in \cref{sec:rand-arrival-weighted}, we study these ideas in the context of weighted matchings.

\subsection{Demonstration of Our Technique via Unweighted Matching}
\label{sec:an-algor-unwe}

We give an algorithm that makes one pass over a uniformly random edge stream of
a graph and computes a $0.506$-approximate maximum unweighted matching.  For
the special case of triangle-free graphs (which includes bipartite graphs), we
give a better analysis to get a $0.512$-approximation.

We denote the input graph by $G = (V,E)$, and use $\optM$ to indicate a matching
of maximum cardinality.  Assume that $\optM$ and a maximal matching $M'$ are
given.  For $i \in \{3,5,7,\ldots\}$, a connected component of $M'\cup \optM$ that
is a path of length $i$ is called an $i$-augmenting path (the component is
called nonaugmenting otherwise).  We say that an edge in $M'$ is $3$-augmentable
if it belongs to a $3$-augmenting path, otherwise we say that it is
non-$3$-augmentable.  Also, for a vertex $u$, let $N(u)$ be $u$'s neighbor set,
and for $S\subseteq E$, let $N_S(U)$ denote $u$'s neighbor set in the edges in
the graph $(V,S)$.

\begin{restatable}[Lemma 1 in~\cite{Konrad2012}]{lemma}{lemmakmm}
  \label{lem:kmmlem1}
  Let $\advntg \ge 0$, $M'$ be a maximal matching in $G$, and $\optM$ be a maximum unweighted
  matching in $G$ such that $|M'| \le (\nicefrac{1}{2} + \advntg) |\optM|$.  Then
  the number of $3$-augmentable edges in $M'$ is at least
  $(\nicefrac{1}{2} - 3\advntg)|\optM|$, and the number of non-$3$-augmentable
  edges in $M'$ is at most $4\advntg|\optM|$.
\end{restatable}
\begin{proof}
  We give the proof in the appendix for completeness.
  See~\Cref{sec:proofs-missing-from}.
\end{proof}

The algorithm is as follows.  Compute a maximal matching $M_0$ on initial $p$
(which we will set later) fraction of the stream.  Then we run three
algorithms in parallel on the remaining $(1 - p)$ fraction of the stream.  In
the first, we store all the edges into the variable $S_1$ that are among vertices
left unmatched by $M_0$.  In the end, we augment $M_0$ by adding a maximum
unweighted matching in $S_1$.  In the second, we continue to grow $M_0$ greedily to get
$M'$.  In the third, to get $3$-augmentations with respect to $M_0$, we invoke
the \unwaugpath algorithm from~\Cref{lem:unw-blackbox} that accepts a matching
$\tilde{M}$ and a stream of edges that contains $\beta$ augmenting paths of
length $3$ with respect to $\tilde{M}$. In this way we obtain a set of vertex disjoint
$3$-augmenting paths, which we then use to augment $M_0$.  We return the best of the
three algorithms

It is clear that the second and the third algorithm use $O(n\log n)$ space.
The following lemma shows that the first algorithm uses $O(n(\log n) / p)$ space.
\begin{lemma}
  \label{lem:sizes1}
	With high probability it holds that $|S_1| \in O(n (\log n) / p)$.
\end{lemma}
\begin{proof}
	Fix a vertex $v$. Define $A_{v,t}$ to be the event that after processing $t$ edges from the stream it holds: $v$ is unmatched, and at least $5 \tfrac{\log{n}}{p}$ neighbors of $v$ are still unmatched. We will show that $\prob{A_{v, pm}} \le n^{-5}$, after which the proof follows by union bound over all the vertices. We have
	\begin{align*}
		\prob{A_{v, t}} & = \prob{A_{v, t} | A_{v, t - 1}} \prob{A_{v, t - 1}} + \prob{A_{v, t} | \neg A_{v, t - 1}} \prob{\neg A_{v, t - 1}} \\
		& = \prob{A_{v, t} | A_{v, t - 1}} \prob{A_{v, t - 1}} \\
		& \le \prob{\text{$v$ is unmatched after processing $t$ edges} | A_{v, t - 1}} \prob{A_{v, t - 1}} \\
		& \le \rb{1 - \frac{5 \frac{\log{n}}{p}}{m - t + 1}} \prob{A_{v, t - 1}} \\
		& \le \rb{1 - \frac{5 \frac{\log{n}}{p}}{m}}^{t} \\
		& \le e^{-\frac{5 t \log{n}}{p m}}.
	\end{align*}
	Therefore, $\prob{A_{v, pm}} \le n^{-5}$ as desired.

\end{proof}

We divide the analysis of approximation ratio into two cases.

\subsubsection*{Case 1. $|M_0| \leq (\nicefrac{1}{2} - \advntg) |\optM|$:}
Each edge of $M_0$ can intersect with at most two edges of $\optM$, hence $S_1$
contains at least $|\optM| - 2|M_0|$ edges of $\optM$ that can be added to $M_0$ to
get a matching of size at least
$|\optM| - |M_0| \ge (\nicefrac{1}{2} + \advntg )|\optM|$.

\subsubsection*{Case 2. $|M_0| \geq (\nicefrac{1}{2} - \advntg) |\optM|$:}
If $|M_0| \geq (\nicefrac{1}{2} + \advntg) |\optM|$, we are done, so assume that
$|M_0| < (\nicefrac{1}{2} + \advntg) |\optM|$.  In the second algorithm, $M'$
is the maximal matching at the end of the stream.  If
$|M'| \ge (\nicefrac{1}{2} + \advntg )|\optM|$, we are done, otherwise, by
\Cref{lem:kmmlem1}, there are at least $(\nicefrac{1}{2} - 3\advntg )|\optM|$
$3$-augmentable edges in $M'$, i.e., there are at least
$(\nicefrac{1}{2} - 5\advntg )|\optM|$ $3$-augmentable edges in $M_0$; denote this
set of edges by $E_3$.  In expectation, for at least $(1-2p)$ fraction of $E_3$,
both the $\optM$ edges incident to them appear in the latter $(1-p)$ fraction of
the stream. This can be seen by having one indicator random variable per edge
in $E_3$ denoting whether two $\optM$ edges incident on that edge appear in the
latter $(1-p)$ fraction of the stream.  Then we condition on the event that
$uv\in E_3$, which implies that $uv$ has two $\optM$ edges, say $au$ and $vb$,
incident on it.  Since $uv$ was added to the greedy matching $M_0$, both $au$
and $vb$ must appear after $uv$.  Any of $au$ and $vb$ appears in the latter
$(1-p)$ fraction on the stream with probability $(1 - p)$ under this conditioning.  Then,
by union bound, with probability at least $(1-2p)$ both $au$ and $vb$ appear in
the latter $(1-p)$ fraction of the stream.  Then we apply linearity of
expectation over the sum of the indicator random variables.
%Sagar: the proof gets unnecessarily messy if we actually show the conditioning.
%We then will have to condition of size of $M_0$ being close to half and stuff
%like that.

Now, by \cref{lem:unw-blackbox}, using
$\beta=(\nicefrac{1}{2} -5\advntg)(1-2p)/(\nicefrac{1}{2} +\advntg) \ge
(1-2p)(1-12\advntg)$, we recover at least
$(1-2p)^2(1 - 12\advntg)^2|M_0|/32 \ge (1-4p)(1 - 24\advntg)|M_0|/32$ augmenting
paths in expectation.  Using $|M_0| \ge (\nicefrac{1}{2} -\advntg)|\optM|$, after
algebraic simplification, we get that the output size is at least
$((\nicefrac{1}{2} - \advntg) +(1-4p)(1 -26\advntg)/64)|\optM|$, i.e., at least
$(\nicefrac{1}{2} + \nicefrac{1}{64} - \nicefrac{90\advntg}{64} -p)|\optM|$.
Letting $\advntg = \nicefrac{1}{154}$ implies that our algorithm outputs a
$(\nicefrac{1}{2} + \advntg - p)$-approximate maximum unweighted matching, i.e.,
$0.506$-approximation for $p\le 0.0001$.

\begin{theorem}
  For random-order edge-streams, there is a one-pass $O(n\polylog n)$-space
  algorithm that computes a $0.506$-approximation to maximum unweighted
  matching in expectation.
\end{theorem}

\begin{remark}
  This algorithm not only demonstrates our technique, but also improves the
  current best approximation ratio of $0.503$ by Konrad et
  al.~\cite{Konrad2012}.  For bipartite graphs, recently, Konrad~\cite{konrad18}
  gave a $0.5395$-approximation algorithm.
\end{remark}

%%% Local Variables:
%%% mode: latex
%%% TeX-master: "000-main_random_order"
%%% End:

\subsection{An Algorithm for Weighted Matching}
\label{sec:rand-arrival-weighted}

% In this section, we give our $(\nicefrac{1}{2} + 4c)$\nobreakdash-approximate
% semi-streaming algorithm for maximum weight matching in the
% random-edge-arrival setting.

Now we discuss the more general weighted case.

Let $G = (V, E, w)$ be a weighted graph with $n$ vertices and $m$ edges,
and assume that the edges in $E$ are revealed to the algorithm in a uniformly 
random order. 
We further assume that the edge weights are positive integers and the maximum 
edge weight is $O(\poly(n))$.
Let $\optWM$ be a fixed maximum weighted matching in $G$.
For any matching $M$ of $G$ and a vertex $v \in V$, let $M(v)$ denote the
edge adjacent to the vertex $v$ in the matching $M$. 
If some vertex $v$ is unmatched in $M$, we assume that $v$ is connected to some 
artificial vertex with a zero-weight edge, whenever we use the notation $M(v)$.

Similarly to the algorithm in \Cref{sec:an-algor-unwe}, we start by
computing a $(\nicefrac{1}{2})$-approximate maximum weighted matching $M_0$ within the
first $p$ fraction of the edges ($p = O(1/\log n)$) using the local-ratio
technique. We recall this technique next.  We consider each incoming edge $e = (u, v)$, and
as long as it has a positive weight, we push it into a stack and subtract its
weight from each of the remaining edges incident to any of its endpoints $u$ and
$v$.  To implement this approach in the streaming setting, for each vertex $v \in V$, we
maintain a vertex potential $\alpha_v$.  The potential $\alpha_v$ tells how much
weight should be subtracted from each incoming edge that is incident to $v$.
After running the local-ratio algorithm for the first $p$ fraction of the edges,
computing $M_0$ greedily by popping the edges from the stack gives a
$(\nicefrac{1}{2})$-approximate matching $M_0$ for that portion of the stream.
This is proved using local-ratio theorem (see the work of Paz and
Schwartzman~\cite{Paz2017}). 
We also freeze the vertex potentials $\alpha_v$ at this point.

Analogous to the unweighted case, we have three possible scenarios for $M_0$:
\begin{enumerate}
\item In the best case,  $w(M_0) \ge (\nicefrac{1}{2} + 4c) \cdot \optwgt$ and we
are done.

\item The weight $w(M_0) \le (\nicefrac{1}{2} - 4c) \cdot \optwgt$, in which case we
have only seen at most $(1 - 8 c) \cdot \optwgt$ worth optimal matching edges so
far, and the rest of the stream contains at least $8 c \cdot \optwgt$ weight that can be
added \emph{on top of} $M_0$.

This corresponds to having a large fraction of unmatched vertices in the
unweighted case, where we could afford to store all the edges incident to those
vertices and compute a maximum unweighted matching that did not conflict with $M_0$.  In
the weighted case, we keep all edges $e = (u,v)$ in the second part of the
stream that satisfy $w(e) > \alpha_u + \alpha_v$, where $\alpha_u$ and
$\alpha_v$ are the frozen vertex potentials after seeing the first $p$ fraction
of the edges. Note that we continue to keep the vertex potential frozen.  (Think
of the unmatched vertices in the unweighted case as vertices with zero
potential.)  Again using the random-edge-arrival property, we show that the
number of such edges that we will have to store is small with high probability.
At the end of the stream, we use an (exact) maximum matching on those edges
together with the edges in the local-ratio stack from the first $p$ fraction of
the stream to construct a $(\nicefrac{1}{2} + 4c) \cdot \optwgt$ matching.

\item The weight of the matching $M_0$ is between
$(\nicefrac{1}{2} - 4c) \cdot \optwgt$ and $(\nicefrac{1}{2} + 4c) \cdot \optwgt$.  In
the analogous unweighted case, we did two things.  We continued to maintain a
greedy matching (on unmatched vertices), and we tried to find augmenting paths
of length three.  For the weighted case we proceed similarly: We continue to
compute a constant factor approximate matching for those edges $e = (u,v)$ such
that $w(e) > w(M_0(u)) + w(M_0(v))$, and akin to the unweighted
$3$-augmentations, we try to find the weighted $3$-augmentations.

For the latter task, we randomly choose  (guess) a set of edges from
$M_0$ that we consider as the \emph{middle} edges of weighted 
$3$-augmentations.
Here, by a weighted $3$-augmentation, we mean a quintuple of edges
$(e_1, o_1, e_2, o_2, e_3)$ that increase the weight of the matching
when the edges $e_1, e_2$, and $e_3$ are removed from $M_0$, and the edges
$o_1$ and $o_2$ are added to $M_0$.
(Although these are length five augmenting paths, we call them 
$3$-augmentations because we reduce the problem of finding those to
the problem of finding length three unweighted augmenting paths.)
We partition the chosen \emph{middle} edges into weight classes defined in
terms of geometrically increasing weights, and for each of the weight classes
we find $3$-augmentations using an algorithm that finds unweighted 
$3$-augmenting paths as a black-box.
\end{enumerate}

Before we proceed to the complete algorithm, we give an algorithm to address 
the third case described above.
In fact, this algorithm is the key contribution of this section: As the title
of this paper suggests, this algorithm improves weighted matchings via
unweighted augmentations.

\subsubsection{Finding Weighted Augmenting Paths}
\label{sec:weighted-augpaths}

Suppose that we have an initial matching $M_0$ such that
$(\nicefrac{1}{2} - 4c) \cdot \optwgt \leq w(M_0) \leq (\nicefrac{1}{2} + 4c) \cdot
\optwgt$.  In this section, we describe how to augment $M_0$ using
$3$-augmentations to get an increase of weight $8c \cdot \optwgt$ that results
in a matching of weight at least $(1/2 + 4c) \cdot \optwgt$.  
To achieve this, in a black-box manner we use the
algorithm \unwaugpath whose existence is guaranteed by
\Cref{lem:unw-blackbox}.

Let $W_i = \{e \in E : 2^{i-1} \leq w(e) < 2^{i} \}$ be the set of edges whose
weight is in the range $[2^{i-1}, 2^{i})$, and let $k$ be the index such that
$\max_{e \in E} w(e) \in W_k$.  Thus $k = O(\log n)$ (recall that the edge
weights are positive integers and the maximum edge weight is $O(\poly(n))$, and
any edge $e \in E$ belongs to exactly one $W_i$).  We refer to $W_i$'s as
weight classes.

As described earlier, we would like to find both weighted $1$-augmentations 
(i.e., single edges that could replace two incident edges in the current 
matching and give a significant gain in weight), and weighted 
$3$-augmentations.
We now give the outline of our algorithm, \augpaths, in 
\Cref{alg:augpaths} using the object-oriented notation,
and we explain its usage and intuition behind its design below. 

\begin{algorithm}[!ht]
	\DontPrintSemicolon
	\SetKwInOut{Global}{Global}
  \SetKwFunction{initialize}{}
	\SetKwFunction{feededge}{}
	\SetKwFunction{finalize}{}
  \SetKw{KwAnd}{and}
  \SetKw{KwOr}{or}
  
  \SetKwProg{function}{function}{}{\KwRet}

  \Global{Instances $\mathcal{A}_i$ of \unwaugpath for 
  		$i = 1, 2, \dots k$, a matching $M_0$, a set $\texttt{Marked}$ of 
  		\emph{marked} edges, and a $(\nicefrac{1}{4})$-approximate streaming algorithm for weighted 
  		matching algorithm \textsc{Approx-Wgt-Matching}.
  	}
  
	\function{\textsc{Initialize}\initialize{A matching $M$}}{
		 Set $M_0 = M$. \;
    	 For each $e \in M_0$, with probability $\nicefrac{1}{2}$, add $e$ to 
    	 $\texttt{Marked}$. \;		
    	 \For{$i = 1$ \KwTo $k$}{
    	 	Initialize $\mathcal{A}_i$ with the matching in 
    	 		$\texttt{Marked} \cap W_i$.\;
    	 	}
	}  
  
  \function{\textsc{Feed-Edge}\feededge{An edge $e = (u,v) \in E$}}{
  			\If{$w(e) \geq w(M_0(u)) + w(M_0(v))$}{
  					Feed $e$ to \textsc{Approx-Wgt-Matching} with weight $w'(e) =
  					w(e) -(w(M_0(u)) + w(M_0(v)))$.\;
  			} 
  			\If{$w(e) \leq (1 + \alpha)(w(M_0(u)) + w(M_0(v)))$}{ 
  				\label{algstep:small-excess}
				\If{$M_0(u) \in \texttt{Marked}$ \KwAnd 
					$M_0(v) \notin \texttt{Marked}$}{ \label{algstep:filter1}
					\If{ $w(e) \geq (1 + 2 \alpha)((\nicefrac{1}{2}) \cdot w(M_0(u)) 
						+ w(M_0(v))$}{\label{algstep:3augc1}
						Feed $e$ into $\mathcal{A}_i$ where $i$ is such 
							that $w(e) \in W_i$. \;
					}
  				}
  				\If{$M_0(v) \in \texttt{Marked}$ \KwAnd 
  					$M_0(u) \notin \texttt{Marked}$}{ \label{algstep:filter2}
					\If{ $w(e) \geq (1 + 2 \alpha)( w(M_0(u)) + (\nicefrac{1}{2}) 
						\cdot w(M_0(v))$}{\label{algstep:3augc2}
						Feed $e$ into $\mathcal{A}_i$ where $i$ is such 
							that $w(e) \in W_i$. \;
					}
  				}
  			}
  }
  
  \function{\textsc{Finalize}\finalize{}}{
		Let $M'$ be the matching computed by \textsc{Approx-Wgt-Matching}. \;
		Let $M_1$ be the matching obtained by adding edges in $M'$ to $M_0$ and
			removing the conflicting edges from $M_o$. \;
		Let $M_2$ be the matching obtained by greedily doing the non-conflicting
			augmentations returned by $\mathcal{A}_i$ for $i = k, k-1, \dots, 1$
			in that order on the initial matching $M_0$. 
			\label{algstep:augmentations} \;
	 	\KwRet $\operatorname{arg\,max}_{i \in [2]} w(M_i)$ \;
  }
  \caption{Outline of the algorithm \augpaths.}
  \label{alg:augpaths}
\end{algorithm}

\paragraph{Initialization:}
We initialize \augpaths by calling the \funinit function, passing the
initial matching $M_0$, which is the matching we compute after seeing the first
$p$ fraction of the edges in our final algorithm.
Given $M_0$, the algorithm will first independently and randomly sample a set of
edges \texttt{Marked}; these are the edges that the algorithm guesses to be the
\emph{middle} edges of $3$-augmentations.
The algorithm will later look for pairs of edges $(o_i, o_{i+1})$ such that
$(e_i, o_i, e_{i+1}, o_{i+1}, e_{i+2})$ is a weighted $3$-augmentation, where
$e_{i+1}$ is a guessed middle edge whereas $e_i$ and $e_{i+2}$ are not.  We aim
to gain at least some constant ($\alpha$ in \Cref{alg:augpaths}) fraction of the
weight of the middle edge by doing the augmentations.  To achieve this, we group
all guessed middle edges into weight classes and use dedicated instances of
\unwaugpath for each weight class.
	
\paragraph{Processing the edge stream:}
Next, a stream of edges (the rest of the stream) is fed to the algorithm using
the function \funfeededge.
The function \funfeededge does two things.  For an edge $e = (u,v)$ that has
excess weight $w'(e) = w(e) - w(M_0(u)) -w(M_0(v))$ (i.e., gain of the
corresponding $1$-augmentations), it tries to recover a matching with a large
excess weight giving a large weight increase on top of $M_0$.
On the other hand, if we do not have large matching with respect to the excess
weights, then it implies that there must be a large fraction of
$3$-augmentations \emph{by weight}.  Thus the function \funfeededge also looks
for $3$-augmentations using \unwaugpath as a black-box.
% For this we need to have a large fraction of $3$-augmentations in the stream
% not only by weight but also by \emph{by number} (see
% \Cref{lem:unw-blackbox}).  Hence we only consider edges that do not have
% too much excess weight, and we show that there exists large \emph{number} of
% $3$-augmenting paths consisting only of those edges with small excess weight
% if the maximum matching with respect to the excess weights is small.
After filtering out the edges with small excess weight, it appropriately feeds
them to the \unwaugpath instance of the correct weight class.  The filtering is
needed to ensure that for each weight class, the number of $3$-augmentations is
large compared to the number of guessed middle edges in that weight class (which
is what $\beta$ refers to in~\Cref{lem:unw-blackbox} that gives \unwaugpath).
	
\paragraph{Finalizing the matching:}
Finally at the end of the stream, we call the \funfinalize function, which uses
the initial matching together with the approximate maximum matching on excess
weights and the outputs of the \unwaugpath instances to construct the final
matching.

\subsubsection*{Analysis of the algorithm}

Assume that \augpaths is initialized with a matching $M_0$, and further assume
that $(\nicefrac{1}{2} - 4c) \cdot \optwgt \leq w(M_0) \leq (\nicefrac{1}{2} + 4c)
\cdot \optwgt$ for some $0 < c < 2^{-15}$ (we will set the exact value of $c$
later).  Let $\optWM$ be a fixed optimal weighted matching in $G$ and let
$\tilde{E} \subset E$ be a subset of edges such that $w(\optWM \cap \tilde{E})
\geq (1 - 0.001) \cdot \optwgt$ (think of $\tilde{E}$ as the edges in the second
part of the stream).  Let $\tilde{M}^\ast$ be the maximum weighted matching in
$\tilde{E}$.  By the previous assumption, we have that $w(\tilde{M}^\ast) \geq
(1 - 0.001) \cdot \optwgt$.  Assume that after the initialization, we feed the
edges of $\tilde{E}$ one at a time to \augpaths in some arbitrary order (not
necessarily random).

Let $\hat{M}$ be the matching returned by the function \funfinalize.  We show
that, under the above assumptions, the expected weight $\E[w(\hat{M})]$ of the
matching $\hat{M}$ is at least $(\nicefrac{1}{2} + 4c) \cdot \optwgt$.

Recall that in \augpaths, the output $\hat{M}$ is the maximum of two matchings
$M_1$ and $M_2$.  The matching $M_1$ is constructed by combining the output $M'$
of the $(\nicefrac{1}{4})$-approximate algorithm \textsc{Approx-Wgt-Matching} on
the excess weights $w'$ with the initial matching $M_0$. That is, $M_1$ is
obtained by adding all edges of $M'$ to $M_0$ and removing the edges that
conflict with those newly added edges from $M_0$.  The matching $M_2$ is formed
by applying the $3$-augmentations given by \unwaugpath instances to the initial
matching $M_0$.

In the construction of $M_1$, when we add an edge $e = (u,v) \in M'$ to $M_0$
and remove the two conflicting edges, the gain of weight is $w(e) - (w(M_0(u)) +
w(M_0(v)) = w'(e)$.  Thus we have $w(M_1) \geq w(M_0) + \sum_{e \in M'} w'(e)$,
and since $\hat{M}$ is the maximum of $M_1$ and $M_2$, we have the following
observation.

\begin{obsv} \label{obsv:excess-matching} If the weight of the matching $M'$
computed by \textsc{Approx-Wgt-Matching} for the excess weights $w'(e)$ is at
least $2^{-12} \optwgt$, then
$$w(\hat{M}) \geq w(M_1) \geq w(M_0) + (2^{-12}) \cdot \optwgt \geq 
(\nicefrac{1}{2} - 4c) \cdot \optwgt + (2^{-12}) \cdot \optwgt \geq (\nicefrac{1}{2} +
4c) \cdot \optwgt.$$ For the last two inequalities we use the facts that $w(M_0)
\geq (\nicefrac{1}{2} - 4c) \cdot \optwgt$ and $c < 2^{-15}$.
\end{obsv}

In light of \Cref{obsv:excess-matching}, we now assume that the
approximate maximum matching in $\tilde{E}$ with respect to the excess weights
is small.  For this case, we show that the matching $M_2$ has at least
$(\nicefrac{1}{2} + 4c) \cdot \optwgt$ weight in expectation.

Let $E_1$ be the edges in $\tilde{E}$ that satisfy the criteria of
\cref{algstep:small-excess}, namely edges $e = (u,v) \in \tilde{E}$ such
that $w(e) \leq (1 + \alpha) ( w(M_0(u)) + w(M_0(v)))$.  These are the edges
that have small excess weight.  We have the following lemma on the
$3$-augmentations that only use edges with small excess weight.

\begin{lemma} \label{lem:large-opt-with-small-excess} If the weight of the
approximate maximum matching $M'$ with respect to excess weights $w'$ is at most
$(2^{-12}) \cdot \optwgt$, then there exist a set of $3$-augmentations that only
use edges in $E_1$ such that the total weight increase of those augmentations is
at least $(0.4) \cdot \optwgt$.
\end{lemma}

\begin{proof} Consider the symmetric difference $\tilde{M}^\ast \triangle M_0$
as a collection of cycles that alternate between $\optWM$ and $M_0$ edges.
Recall that $\tilde{M}^\ast$ is the maximum matching in $\tilde{E}$, and assume
that both $\tilde{M}^\ast$ and $M_0$ are perfect matchings (with zero-weight
edges between unmatched vertices).

Without loss of generality, we assume that it is a single cycle of length $2n$
(for the case of multiple cycles, the following proof can be easily modified to
take the summations over all cycles and we can replace $n$ with the actual cycle
length).  Label the edges in the cycle as $e_1,o_1,e_2,o_2, \dots, e_n, o_n$
(assume that the indices wrap around so that $e_{n + i} = e_i$ and $o_{n + i} =
o_i$) so that the $e$-edges belong to $M_0$ and $o$-edges belong to
$\tilde{M}^\ast$.

Let $P_i$ denote the quintuple $(e_i, o_i, e_{i+1}, o_{i+1}, e_{i+2})$ of edges,
and let $g(P_i)$ denote the gain $w(o_i) + w(o_{i+1}) - w(e_i) - w(e_{i+1}) -
w(e_{i+2})$ we get by augmenting $P_i$ (i.e., by removing edges $e_i, e_{i+1},
e_{i+2}$ from $M_0$ and adding edges $o_i, o_{i+1}$ to $M_0$).  We have that
\begin{align*}
\sum_{i \in [n]}g(P_i) 
&= 2 \sum_{i \in [n]}w(o_i) - 3 \sum_{i \in [n]}w(e_i) 
	\geq 2 (1 - 0.001) \cdot \optwgt - 3 \cdot  w(M_0) \\
&\geq  (2 (1 - 0.001) - 3(\nicefrac{1}{2} + 4c)) \cdot \optwgt \geq (\nicefrac{1}{2} - 0.003) \cdot \optwgt,
\end{align*}
where the inequality follows from the assumption that $w(M_0) \leq 
(\nicefrac{1}{2} + 4c) \cdot \optwgt$.

Now let $L$ be the set of indices $i$ for which either 
$\woi \geq (1 + \alpha) (\wei + \weia)$
or $\woia \geq (1 + \alpha)(\weia + \weib)$.
Thus we have that,
$\sum_{i \in [n]} g(P_i) = \sum_{i \in [n] \setminus L} g(P_i) + 
\sum_{i \in L} g(P_i).$
Furthermore, we have \begin{align*}
\sum_{i \in L} (g(P_i) - \weia) &= \sum_{i \in L} ((\woi - \wei - \weia) + 
	(\woia - \weia -\weib)) \\
	&\leq \sum_{i \in L} ((\woi - \wei - \weia)^+ + 
	(\woia - \weia -\weib))^+ \\
	&\leq \sum_{i \in [n]} ((\woi - \wei - \weia)^+ + 
	(\woia - \weia -\weib))^+ \\
	&= 2 \sum_{i \in [n]} \underbrace{(\woi - \wei - \weia)^+}_{w'(o_i) 
	\text{ or } 0} \\
	&\leq 2 \cdot 4 \cdot (2^{-12}) \cdot \optwgt < (0.002) \cdot \optwgt.
\end{align*}
The last line above follows from the fact that $M'$ is a $4$-approximation with 
respect to the weight function $w'$, and thus any matching has
weight at most $4 \cdot w'(M')$ with respect to weights $w'$.
On the other hand, for any $i \in L$, by definition, either
$$w'(o_i) = \woi - \wei - \weia \geq \alpha (\wei + \weia) \geq \alpha \weia$$ 
or 
$$w'(o_{i+1}) = \woia - \weia - \weib \geq \alpha (\weia + \weib) \geq \alpha 
\weia.$$
Thus $\sum_{i \in L} \weia \leq (1/\alpha) \sum_{i \in [n]} w'(o_i) \leq 
 4 \cdot (2^{-12}) \cdot (1/\alpha) \cdot \optwgt \leq (0.05) \cdot \optwgt$ when 
 $\alpha = 0.02$.
Putting these together, we get
\begin{align*}
\sum_{i \in [n] \setminus L} g(P_i) &\geq (\nicefrac{1}{2} - 0.003) \optwgt - 
\sum_{i \in L} (g(P_i) - \weia) - \sum_{i \in L} \weia \\
&\geq  (\nicefrac{1}{2} - 0.003 - 0.002 - 0.05) \cdot \optwgt \geq (0.4) 
\cdot \optwgt.
\end{align*}
\end{proof}

Let $O_1 = [n] \setminus L$ where $L$ is defined as in the proof of 
\Cref{lem:large-opt-with-small-excess} so that the augmentations $P_i$
for $i \in O_1$ only uses edges with small excess weight. 
(Recall that $P_i = (e_i, o_i, e_{i+1}, o_{i+1}, e_{i+2})$ where $o_i$ and 
$o_{i+1}$ are edges in $\tilde{M}^\ast$, which is a fixed optimal matching in 
$\tilde{E}$.)
Formally,
$$O_1 = \{i \in [n]: \woi \leq (1 + \alpha)(\wei + \weia) \text{ and } 
\woia \leq (1 + \alpha) (\weia + \weib) \}.$$

Let
$O_2 = \{i \in O_1: g(P_i) \geq (\nicefrac{1}{2} + 3 \alpha)\weia + 2 \alpha
\wei + 2 \alpha \weib \}$.
% The condition on gain $g(P_i)$ for $i \in O_2$ is defined
% in a such a way that, for such $P_i$, if the algorithm \augpaths has guessed
% the middle edge (i.e., it has added $e_{i+1}$ to \texttt{Marked}),
% then it will surely select both the edges $o_i$ and $o_{i+1}$ for 
% further processing (in \Cref{algstep:filter1} or in 
% \Cref{algstep:filter2}).
% This is guaranteed by the following lemma, which also gives a lower bound on
% total gain due to $3$-augmentations in $O_2$, and an upper bound on individual
% gain in terms of the weight of the middle edge.  
We need the bounds we show in \Cref{lem:gain-bounds} below for the analysis of
$3$-augmentations.  Note that the first two parts correspond to the conditions
on \cref{algstep:3augc1,algstep:3augc2}.  To recover sufficient number of
$3$-augmentations using \unwaugpath as a black box, each weight class that has a
large fraction of augmentations by weight should also have a large fraction of
them by number.  This is because the guarantee of \Cref{lem:unw-blackbox}
is conditioned on the existence of many augmenting paths.  For this reason, we
need the upper bound on the gain of each individual augmentation as given in the
third part of the lemma.

\begin{lemma}
  For all $i \in O_2$, we have
  \begin{enumerate}
    \item $\woi \geq (1 + 2 \alpha) (\wei + (\nicefrac{1}{2}) \weia)$,
    \item $\woia \geq (1 + 2 \alpha) ((\nicefrac{1}{2})\weia + \weib)$, and
    \item $g(P_i)\le3 w(e_{i+1})$.
  \end{enumerate}
  \label{lem:gain-bounds}
\end{lemma}
\begin{proof}
For $\woi$ we have
\begin{align*}
\woi &= g(P_i)  + (\wei + \weia + \weib) - \woia \\
&\geq \left((\nicefrac{1}{2} + 3 \alpha)\weia + 2 \alpha \wei +
  2\alpha \weib  \right) 
+ \left( \wei + \weia + \weib \right) - \woia \\
&= (1 + 2 \alpha) \wei + (1 + 2\alpha) \weib + (\nicefrac{3}{2} + 3\alpha) \weia - \woia \\
&\geq (1 + 2 \alpha) \wei + (1 + 2\alpha) \weib + (\nicefrac{3}{2} + 3\alpha) \weia - 
\left( (1 + \alpha)(\weia + \weib) \right)\\
		&\geq (1 + 2 \alpha) (\wei + (\nicefrac{1}{2}) \weia).
\end{align*}
The claim on $\woia$ follows similarly.  This proves the first two parts of the
lemma.

% Now that we know the augmentations in $O_2$ are captured by \augpaths (if 
% the middle edges are guessed correctly), we show that these augmentations
% amounts to a large weight.

%\ref{lem:gain-ub} below.

Now, observe that we have
$(1 + \alpha) (\wei + \weia) \geq \woi \geq (1 + 2 \alpha)(\wei +
(\nicefrac{1}{2})\weia)$.  This implies that
\begin{align*}
(1 + \alpha) (\wei + \weia) &\geq (1 + 2 \alpha)(\wei + (\nicefrac{1}{2})\weia),
\end{align*}
which simplifies to $(\nicefrac{1}{2}) \weia \geq \alpha \wei$ or equivalently
$\wei \leq (1/(2\alpha)) \weia$.  Similarly we can show that
$\weib \leq (1/(2\alpha)) \weia$.  Thus we have that
\begin{align*}
g(P_i) &= w(o_i) + w(o_{i+1}) - ( w(e_i) + w(e_{i+1}) + w(e_{i+2})) \\
& \leq (1 + \alpha)\cdot(\wei + \weia) + (1 + \alpha) \cdot (\weia + \weib) - (\wei + \weia + \weib) \\
& = (1 + 2\alpha) w(e_{i+1}) + \alpha (\wei +  \weib) \\
& \leq  (1 + 2\alpha) w(e_{i+1}) + 2 \alpha \cdot (1/(2 \alpha)) \cdot \weia \\
& \leq 3 \cdot \weia. 
\end{align*}
\end{proof}

The guarantee of \unwaugpath holds when there exist large number of
vertex-disjoint $3$-augmenting paths.  To ensure this, we need our weighted
augmentations $P_i$ to be edge-disjoint, and for this we need the following
lemma.
%\todo[inline]{From here onward, all the technical details are finalized,
%but still need to add text so that it reads nicely.}

\begin{lemma}	\label{lem:augment-final}
There exists a set $Q \subseteq O_2$ of indices such that the augmenting paths 
$P_i$ are edge-disjoint, and	$\sum_{i \in Q} g(P_i) \geq (0.02) \cdot \optwgt$.
\end{lemma}
\begin{proof}
Since $O_2 \subseteq O_1$, we show that the gain of augmentations in $O_2$ is
also large by bounding the gain of augmentations in $O_1 \notin O_2$. 

A single $P_i$ can be in conflict with at most two other such paths, 
namely $P_{i-1}$ and $P_{i+1}$. 
Thus by greedily picking paths $P_i$ with the maximum gain that do not share 
edges with the previously picked pairs, we get at least $1/3$ fraction of the 
total gain of 	$O_2$, which is 
\begin{align*}
\nicefrac{1}{3}\sum_{i \in O_2} g(P_i) &\geq \nicefrac{1}{3}\left(\sum_{i \in O_1} g(P_i)  
	-\sum_{i \in O_1 \setminus O_2} g(P_i)\right)\\
&\geq \nicefrac{1}{3}\left(0.4 \optwgt - \sum_{i \in O_1 \setminus O_2} ((\nicefrac{1}{2} + 
3 \alpha)\weia 
	+ 2 \alpha \wei + 2\alpha \weib) \right)\\
&\geq \nicefrac{1}{3}\left(0.4 \optwgt - (\nicefrac{1}{2} + 7 \alpha) w(M_0)
                                                     \right)\\
  &\geq \nicefrac{1}{3} (0.4 - 
(0.64)(\nicefrac{1}{2} + 4c)) \cdot \optwgt\\
 &\geq (0.02) \cdot \optwgt.
\end{align*}
\end{proof}

Let $Q_1, Q_2, \dots, Q_k$ be the indices $i \in Q$ partitioned into $k$ sets 
according to the weight class of $e_{i+1}$. 
That is, $i \in Q_j$ if and only if $ i \in Q$ and $e_{i+1} \in W_j$.
For each $j$, let $Q'_j = \{ i \in Q_j: e_{i+1}$ is marked and both $e_i$ and 
$e_{i+2}$ are not marked$\}$. 
Let $N_j = M_0 \cap W_j$ be the set of edges in the initial matching $M_0$
that belong to weight class $W_j$.
Let $N'_j$ denote the subset of edges in $N_j$ that are marked.
Thus $Q_j$'s and $N_j$'s are fixed (given $M_0$) whereas $Q'_j$'s and $N'_j$'s are 
random. 
We assume that $|N_j| \geq (100 / \beta)$ for some constant $1 > \beta > 0$.
If $|N_j| < (100 / \beta)$, $|N'_j| \leq |N_j|$ is also less than $(100/ \beta)$. 
Hence we can afford to keep all the edges in the stream that are incident on 
any edge in $N'_j$ and run an offline algorithm at the end to find the maximum
set of $3$-augmenting paths that use the edges in $N'_j$ as middle edges.
Such an algorithm stores at most $4 \cdot (100/\beta) \cdot n$ edges 
(at most $2n$ edges per one end point of an edge in $N'_j$).
Thus we assume that $|N_j| \geq (100/\beta)$.

Fix $j \in \{1, 2, \dots, k\}$ and let $\Aug_j$ be the set of augmentations 
returned by the \unwaugpath instance $\mathcal{A}_j$. 
Let $\Aug'_k = \Aug_k$, and for $j = k-1, k-2, \dots, 1$, let $\Aug'_j$ be the set 
of augmentations returned by $\mathcal{A}_j$ that are not in conflict with 
any of the augmentations in $\Aug'_{j+1}, \dots, \Aug'_k$.

We now show that if we have large number of augmentations in some weight class,
the our algorithm will pick a large fraction of them, and consequently, the
unweighted algorithm will also find a large fraction of them. 
To be precise, we have the following lemma.

\begin{lemma} \label{lem:unweighted-succeeds}
Fix some $j$ such that $|Q_j| \geq 16 \beta |N_j|$. 
Then $\E[|\Aug_j|] \geq 2^{-8} \beta^2 \cdot|Q_j|$. 
\end{lemma}

\begin{proof}
  Let $B_1$ denote the event $|N'_j| \leq (1/4)|N_j|$ and $B_2$ denote the event
  $|Q'_j| \leq (1/16)|Q_j|$.
  
Each edge $e \in N_j$ appears in $N'_j$ independently with probability $\nicefrac{1}{2}$. 
Therefore, $\E[|N'_j|] = (\nicefrac{1}{2}) |N_j|$, and by Chernoff bounds, 
$$\Pr \left[B_1\right] 
\leq e^{-(\nicefrac{1}{2})^2 (\nicefrac{1}{2})|N_j|/2} = e^{-(1/16)|N_j|} \leq 1/4.$$
Similarly, since the paths $P_i$ for $i \in Q_j$ are disjoint, each $i \in Q_j$
appears in $Q'_j$ independently with probability $(\nicefrac{1}{2})(1 - \nicefrac{1}{2})^2 = 1/8$. 
Hence $\E[|Q'_j|] = (1/8)|Q_j|$, and by Chernoff bounds, 
$$\Pr\left[B_2 \right] \leq e^{-(\nicefrac{1}{2})^2(1/8)|Q_j|/2} = 
e^{-(1/64)|Q_j|} \leq e^{-(1/4) \beta |N_j|} \leq 1/4.$$ 

The ``good event''
$\bar{B_1}\cap\bar{B_2}$ implies that  $(1/4)|N_j| 
\leq |N'_j| \leq |N_j| \leq (1/(16 \beta)) |Q_j|$ and 
$|Q'_j| \geq (1/16)|Q_j|$, and consequently $|Q'_j| \geq \beta |N'_j|$.  Also,  $\Pr[\bar{B_1}\cap\bar{B_2}] \ge1 - 1/4 - 1/4 = \nicefrac{1}{2}$.
Notice that  $|N'_j|$ is the initial matching size of \unwaugpath 
instance $\mathcal{A}_j$ while each $i \in Q'_j$ corresponds to an unweighted
$3$-augmentation $(o_i, e_{i+1},o_{i+1})$ with respect matching $N'_j$.
Also notice that those $3$-augmentations are vertex-disjoint since for all 
$i \in Q'_j$,  the augmentations $(e_i, o_i,e_{i+1},o_{i+1},e_{i+2})$ 
are edge disjoint (by \Cref{lem:augment-final}).

Hence we have,
\begin{align*}
 \E[|\Aug_j|] 
 \geq 	\Pr \left[ \bar{B_1}\cap\bar{B_2} \right] \cdot \underbrace{ \E \left[|\Aug_j| \;\Big
 		\rvert \; \bar{B_1}\cap\bar{B_2}   \right]}_{ \geq (\beta^2/32)|N'_j| 
 		\text{ by \Cref{lem:unw-blackbox}}} 
  \geq \frac{1}{2} \frac{\beta^2}{32} |N'_j| 
  		\geq  \frac{\beta^2}{64} \frac{1}{4}|N_j| 
  		= \frac{\beta^2}{256} |N_j| \geq \frac{\beta^2}{256}|Q_j|. 
\end{align*}
The last inequality holds because $|Q_j| \leq |N_j|$ (each $i \in Q_j$ 
is associated with a unique edge in $N_j$, namely $e_{i+1}$).
\end{proof}

We are now ready to show that the total gain of the augmentations over all
weight classes is high.  Recall that this and 
\Cref{lem:large-opt-with-small-excess,lem:gain-bounds,lem:augment-final,lem:unweighted-succeeds}
hold under the assumption that $w'(M') \leq (2^{-12}) \cdot \optwgt$.

\begin{lemma} \label{lem:aug-paths-final}
  The total expected gain of weight
  %$\E \left[ \sum_{j = 1}^k \Delta'_j \right]$ 
	we get by doing the augmentations in \cref{algstep:augmentations} in 
	\augpaths is at least $8 c \cdot \optwgt$ for some sufficiently small constant
	$c > 0$.
\end{lemma}

\begin{proof}

Let $\Delta' = \sum_{j \in [k]}|\Aug_j'| \cdot 2^{j-1} \alpha$, thus the total 
gain of all the augmentations is at least $\Delta'$ (for each augmentation where
the middle edge belongs to $N_j$, we gain at least $ 2^{j-1} \alpha$).
Let $\Delta = \sum_{j \in [k]}|\Aug_j| \cdot 2^{j-1}\alpha$. 
Recall that by definition of $\Aug_j'$'s, each augmentation in $\Aug_j$
can block at most $2$ other augmentations of lower weight classes $\Aug_{j-1},
\Aug_{j-2}, \dots, \Aug_{1}$.
Thus if we consider a term $2^{j-1} \alpha$ in the summation $\Delta$, it can
eliminate at most $2 \alpha \sum_{j' = 1}^{j-2} 2^{j'-1} \leq 2 \cdot  
2^{j-1} \alpha$ worth of other terms in the summation $\Delta'$. 
Thus we have that $\Delta' \geq (1/3) \cdot \Delta$, and hence it is sufficient to 
show that $\E[\Delta] \geq 24 c  \cdot \optwgt$, which would imply that
$\E[\Delta'\geq 8c\cdot \optwgt]$.

First notice that $\Delta =  \sum_{j \in [k]}|\Aug_j'| \cdot 2^{j-1}\alpha \geq  
\sum_{j \in [k] : |Q_j| \geq 16 \beta |N_j|} |\Aug_j| \cdot 2^{j-1} \alpha$. 
Thus we have
	\begin{align}
		\E[\Delta]  &\geq \sum_{ \substack{j \in [k] \\ |Q_j| \geq 16 \beta |N_j|}} 
			\E \left[ |\Aug_j| \right] \cdot 2^{j-1} \alpha
		\geq   \sum_{\substack{j \in [k] \\ |Q_j| 
		\geq 16 \beta |N_j|}} \frac{\beta^2}{512} |Q_j| 2^{j - 1} \alpha 
		= \frac{\alpha \beta^2}{1024} \left( \sum_{\substack{j \in [k] \nonumber \\ 
			|Q_j| \geq 16 \beta |N_j|}}  |Q_j| 2^{j} \right)\\
		&\geq \frac{\alpha \beta^2}{1024} \left( \sum_{j = [k]}  |Q_j| 2^{j}  
			- \sum_{j \in [k]} 16 \cdot 2 \cdot \beta |N_j| 2^{j-1}\right). 
			\label{eq:gain-lb}
\end{align}

Observe that $\sum_{j=1, \dots, k} 16 \cdot 2 \cdot \beta |N_j| 2^{j-1} = 
32 \beta \sum_{j \in [k]}|N_j|2^{j-1} \leq 32 \beta w(M_0) \leq 32 \beta \cdot 
\optwgt \leq (0.001) \optwgt$ for $\beta \leq (1/16000)$.

We now lower bound $\sum_{j = [k]}  |Q_j| 2^{j}$.
We have $\sum_{j \in [k]}|Q_j|2^j \geq \sum_{j \in [k]} \sum_{i \in Q_j} 
w(e_{i+1})$ because for each  $i \in Q_j$, the middle edge $e_{i+1}$ of $P_i$,
belongs to the weight class $N_j$ and hence $w(e_{i+1}) \leq 2^j$.
But by third part of \Cref{lem:gain-bounds} we have $\sum_{i \in Q_j}w(e_{i+1}) \geq
\sum_{i \in Q_j} (1/3) \cdot g(P_i)$, and consequently
\begin{align*}
\sum_{j = [k]}  |Q_j| 2^{j} \geq (1/3) \cdot \sum_{j \in [k]} 
\sum_{i \in Q_j} g(P_i) 
= (1/3) \sum_{i \in Q} g(P_i) \geq (1/3) (0.02) \cdot \optwgt > 
(0.003) \cdot \optwgt.
\end{align*}
The first inequality of the last line follows from 
\Cref{lem:augment-final}.

Combining these bounds with inequality~\eqref{eq:gain-lb} yields
$E[\Delta] \geq (\alpha \beta^2 / 1024) (0.003 \cdot \optwgt - 0.001 \cdot \optwgt) = 
(\alpha \beta^2 / 1024)(0.002) \cdot \optwgt$, and thus $E[\Delta'] 
\geq (\alpha \beta^2 / (3 \cdot 1024))(0.002) \cdot \optwgt$.

We earlier set $\alpha = 0.02$.
To finish the proof, set $\beta = 1/16000$ and  
$c = (\nicefrac{1}{8}) \cdot (\alpha \beta^2 / (3 \cdot 1024))(0.002)$.
\end{proof}

% \todo[inline]{
% The constant is super small and I hope nobody will mind. \Cooley
% }

\Cref{lem:aug-paths-final} implies that if the weight of
$M'$ with respect to excess weights $w'$ is small, then in expectation
we recover a good matching through augmentations; together with 
\Cref{obsv:excess-matching}, this gives the following.

\begin{lemma} \label{thm:weighted-aug-paths}
There exists a constant $c > 0$ such that the following holds:
if \augpaths is initialized with a matching $M_0$ satisfying 
$(\nicefrac{1}{2} - 4c) \cdot \optwgt \leq w(M_0) \leq (\nicefrac{1}{2} + 4c) \cdot \optwgt$,
%then conditioned on the event that
and if the input edge stream $\tilde{E}$ contains a matching of weight
at least $(1 - 0.001) \cdot \optwgt$, the expected weight $\E[w(\hat{M})]$ 
of the output $\hat{M}$ of \augpaths is at least $(\nicefrac{1}{2} + 4c) \cdot \optwgt$.
\end{lemma}

We finally note the following lemma on the space complexity of \augpaths.

\begin{lemma} \label{lem:space-augpaths}
The algorithm \augpaths uses $O(n \poly ( \log n))$ memory.
\end{lemma}

\begin{proof}
The algorithm runs at $O(\log n)$ copies of the unweighted algorithm 
\unwaugpath which in turn takes $O(n)$ memory per copy.
Furthermore, the $(1/4)$-approximation algorithm for weighted matching used in \augpaths can 
be implemented using the $(\nicefrac{1}{2} - \epsilon)$-approximation algorithm given by 
Paz and Schwartzman~\cite{Paz2017}, which uses $O(n \poly (\log n))$ memory.
Apart from that, \augpaths only needs $O(n)$ memory to store the
initial matching $M_0$.
\end{proof}

\subsubsection{Main Algorithm for \texorpdfstring{$(\nicefrac{1}{2} + c)$}{(1/2 + c)}-Approximate Matching}

Now that we know how to tackle the difficult case of finding weighted
$3$-augmenting paths, we shift our focus back to the main algorithm.  See
\textsc{Random-Arrival-Matching} in \Cref{alg:rand-arrival}.

We quickly recap.  \textsc{Random-Arrival-Matching} runs the local-ratio method
for the first $p = O(1/\log n)$ fraction of the edge stream and maintains vertex
potentials.  Then it runs two algorithms in parallel for the rest of the stream:
One is the algorithm \augpaths we described in
\Cref{sec:weighted-augpaths}.  The other algorithm merely stores all
edges that \emph{would} have been added to the local-ratio stack if we had
continued to run it till the end.

\begin{algorithm}[htb]
	\DontPrintSemicolon
	\SetKwInOut{Global}{Global}
  \SetKwInOut{Input}{Input}
  \SetKwInOut{Output}{Output}
  \SetKwFunction{push}{Push}
  \SetKwFunction{pop}{Pop}
  \SetKwFunction{add}{Add}
  \SetKwFunction{augalg}{WAP}
  \SetKwFunction{keepall}{Local-Ratio-Residual}
  \SetKwFunction{initaug}{Initialize-2-3-Augmentations}
  \SetKwFunction{tryaug}{Add-2-3-Augmentation-Edge}
  \SetKwFunction{approx}{Factor-4-Approx}
  \SetKw{KwAnd}{and}
  \SetKw{KwOr}{or}
  \Input{Number of vertices $2 n$, number of edges $m$, a stream of edges $E$,
  		a weight function $w:E \to \R^+$ where $G=(V, E, w)$ is weighted 	
  		graph, and an instance \augalg of the weighted augmenting paths algorithm
  		\augpaths.}
  \Output{A matching $M$ of $G$.}
  \Global{Stack $S$ of edges, set $T$ of edges, a vertex potential
  		vector $\alpha \in \R^{V}$.}
  \BlankLine
  Let $E = (e_1 = (u_1, v_1), e_2 = (u_2, v_2), \dots, e_m = (u_m, v_m))$.
	Let $S = [ \, \, ]$ and let $T = [ \, \, ]$.
  Let $\alpha_v \leftarrow 0$ for all $v \in V$. \;
	Let $p \gets 100/\log n$. \;
  \For{$i \leftarrow 1$ \KwTo $p \cdot m$}{
    Let $w'(e_i) = w(e_i) - \alpha_{u_i} - \alpha_{v_i}$. \label{line:w'-definition}\;
    \If{$w'(e_i) > 0$}{
      \push{$S$, $e_i$} \;
      $\alpha_{u_i} \leftarrow \alpha_{u_i} + w'(e_i)$ \;
      $\alpha_{v_i} \leftarrow \alpha_{v_i} + w'(e_i)$ \;
    }
  }
  Let $M_0$ be the matching computed by unwinding stack $S$. 
  \label{algstep:mathcing-m0}\;
  \augalg.\textsc{Initialize}\initialize{$M_0$} \;
  \For{$i \leftarrow p \cdot m + 1$ \KwTo $m$}{
  	\lIf{$w(e_i) > \alpha_u + \alpha_v$}{\add{T, $e_i$}}
    \augalg.\textsc{Feed-Edge}\feededge{$e_i$}
	}
  \SetKwProg{function}{function}{}{\KwRet}
	
	Let $M_1$ be the maximum matching in $T$ with respect to weights $w''(e = (u,v)) = w(e) - \alpha_u - \alpha_v$.\label{algstep:residual-matching} \; 
	\While{$S$ is not empty}{
		$e = \gets $ \pop{S} \;
		\lIf{the endpoints of $e$ are not matched in $M_1$}{\add{$M_1$, $e$}}
	}
	Let $M_2 = $ \augalg.\textsc{Finalize}\finalize{} \;
	
	\KwRet the better of $M_1$ and $M_2$ \;

  \caption{Outline of the algorithm \ram.}
  \label{alg:rand-arrival}
\end{algorithm}

\subsubsection*{Analysis of the main algorithm}

We will first show that the expected weight of the matching returned by \ram is
at least $(\nicefrac{1}{2} + c) \cdot \optwgt$, where $c$ is the constant given by
\Cref{thm:weighted-aug-paths}.  We consider three cases based on the weight of
$M_0$ which is computed in \cref{algstep:mathcing-m0}.

\paragraph{Case 1:} The weight of $M_0$ is at least
$(\nicefrac{1}{2} + 4c) \cdot \optwgt$, in which case we have nothing more to do.

\paragraph{Case 2:} The weight of $M_0$ is at most $(\nicefrac{1}{2} - 4c)$. 
For this case, we show that the matching $M_1$ computed by \ram has a weight 
of at least  $(\nicefrac{1}{2} + 4c) \cdot \optwgt$ in the following lemma.

\begin{lemma} \label{lem:m1-is-good}
	If $w(M_0)  \leq (\nicefrac{1}{2} - 4c) \cdot \optwgt$, then $w(M_1) \geq (\nicefrac{1}{2} + 4c) \cdot \optwgt$.
\end{lemma}

\begin{proof}
Let $A = S \cup \{e = (u, v) \in E: w(e) \leq \alpha^\ast_u + \alpha^\ast_v \}$ 
where 	$\alpha^\ast$ is the vertex potential vector after seeing the first 
$p$ fraction of the edges.  
Then, with respect to the graph $G' = (V, A, w)$, the local-ratio stack $S$ 
contains a $\nicefrac{1}{2}$-approximate matching.
Suppose that $w(M_0) = (\nicefrac{1}{2} - \gamma) \cdot \optwgt$ where $\gamma \geq 4c$. 
Then the optimal matching of $G'$ is at most $(1 - 2 \gamma)\cdot \optwgt$, which means 
that the graph with the remaining edges, with respect 
to the weight function $w''(e = (u, v)) = w(e) - \alpha^\ast_u - \alpha^\ast_v$, 
has a matching of weight at least $2 \gamma \cdot \optwgt$. 

Thus the matching $M_1$ computed before \cref{algstep:residual-matching} has
a weight of at least $2 \gamma \cdot \optwgt$ with respect to the weight function $w''$.

We now show that, by unwinding the stack $S$ in the while loop in \cref{algstep:residual-matching}, 
we can increase the weight by at least $(\nicefrac{1}{2} - \gamma) \cdot \optwgt$. 
Let $M''$ be any matching in $G'$. 
By following the same lines of Ghaffari~\cite{Ghaffari2017}, who gave a more intuitive analysis of 
$(\nicefrac{1}{2} - \epsilon)$-approximate algorithm by Paz and Schwartzman~\cite{Paz2017}, 
we show the following: 
There is a way to delegate the weights of $M''$-edges on to the edges of the matching $M_1$ computed by
\cref{alg:rand-arrival}, such that each edge $e \in M_1$ takes at most $2 \cdot w(e)$ delegated weight.

Fix some edge $e_r = (u_r, v_r)$ in $G'$ and consider the time we push it on to the stack $S$. 
With a slight abuse of notation, let $G$ denote the remaining graph before pushing $e_r$ on to the stack,
and $w$ is the weight function at that time.
Let $G_r$ be the graph after the removal of $e_r$ and let $w_r$ be the
updated weight function. 
Let $M_r$ be the snapshot of $M_1$ just before we pop $e_r$ out of the stack,
and let $M$ be the snapshot of $M_1$ after popping out $e_r$ and processing it.
By induction, assume that in $G_r$, there is a way to delegate the weights $w_r$ of
$M''$-edges on to the edges of $M_r$ such that each edge $e \in M_r$ takes at most $2 \cdot w_r(e)$ delegated weight.
The base case is just before we start processing the stack, and the claim is trivially true as all $M''$-edges have zero
weight at this point.

To conclude the inductive proof, we now show that in $G$, we can delegate the weights $w$ of $M''$-edges on to the 
edges of $M$ such that each edge $e \in M$ takes at most $2 \cdot w(e)$ delegated weight.

In $M''$, there can be at most two edges $e_1, e_2$ incident to the edge $e_r$. 
(It may happen that $e_r$ is in $M''$ so that we have exactly one such edge.)
By inductive hypothesis, for each $e_i \in \{e_1, e_2\}$ we have already found a way to delegate 
the weight $w_r(e_i) = w(e_i) - w(e_r)$ on to $M_r$ edges. 
We need to find room to delegate at most $(w(e_1) - w_r(e_1)) + (w(e_2) + w_r(e_2)) = 2 \cdot w(e_r)$ more weight.
When we pop $e_r$ out of the stack, we have the following two cases:
\begin{enumerate}
\item At least one of the endpoints $u_r$ or $v_r$ of edge $e_r$ is matched in $M_r$ with some edge $e'_r$.
Thus edge $e'_r$ has taken at most $2 w_r(e'_r) = 2 (w(e'_r) - w(e_r))$ amount of delegated weight at the moment. 
But in $G$, edge $e'_r$ can take up to $2 \cdot w(e'_r)$ weight, hence we have room for $2 \cdot w(e_r)$ on $e'_r$.
\item Both endpoints $u_r$ and $v_r$ of edge $e_r$ are unmatched in $M_r$, so that we add $e_r$ to our matching as a new edge. 
Therefore it has its full capacity of $2 \cdot w(e_r)$ remaining for the delegated weight.
\end{enumerate}
Thus we have room to delegate a weight of $2 \cdot w(e_r)$ in both the cases, and thus the 
step of processing the stack $S$ in the while loop in \cref{algstep:residual-matching} increases the weight of 
$M_1$ by at least $\nicefrac{1}{2} w(M'')$ where $M''$ is any matching in $G'$. 
Setting $M''$ to be the maximum matching in $G'$, we get $w(M_1) 
\geq  \nicefrac{1}{2} \cdot (1 - 2\gamma) \cdot \optwgt + 2 \gamma \cdot \optwgt 
 = (\nicefrac{1}{2} + \gamma) \cdot \optwgt \geq (\nicefrac{1}{2} + 4c)  \cdot \optwgt$.
\end{proof}

\paragraph{Case 3:} The matching $M_0$ is such that $(\nicefrac{1}{2} - 4c) \cdot \optwgt 
\leq w(M_0) \leq (\nicefrac{1}{2} + 4c) \cdot \optwgt$. 
For this case,  we already proved that the expected weight is at least 
$(\nicefrac{1}{2} + 4c) \cdot \opt$ if the last $(1 - p)$ fraction of the stream contains 
a matching of weight at least $(1 - 0.001) \cdot \optwgt$.

Now we put together \Cref{lem:m1-is-good} with 
\Cref{thm:weighted-aug-paths} and prove the following theorem.

\begin{theorem}
For sufficiently large $n$, the expected weight of the matching returned by 
the algorithm \ram is at least $(\nicefrac{1}{2} + c) \cdot \optwgt$.
\end{theorem}

\begin{proof}
Let $\tilde{M}$ be the output of \ram.
Let $\mathcal{E}$ denote the event that the last $(1 - p)$ fraction of 
the stream contains a matching of weight at least $(1 - 0.001) \optwgt$.
Let $\optWM$ be a fixed optimal matching of the graph.
Let $E_1$ denote the first $p$ fraction of the edges in the graph.
By the random order arrival property, we have that
$\E[w(\optWM \cap E_1)] = p w(\optWM) = p \optwgt$. 
To see this, notice that each edge in $E_1$ has equal chance of being one 
of the edges of $\optWM$, and then the result follows from the linearity 
of expectation.
Thus by Markov's inequality $\Pr[w(\optWM \cap E_1) \geq 0.001 \optwgt] < 
p/0.001 < c$ for sufficiently large $n$ as $p = O(1/\log n)$.
Thus $\Pr[\mathcal{E}] \geq 1 - 4c$.
Also, due to \Cref{lem:m1-is-good} and 
\Cref{thm:weighted-aug-paths} $\E[w(\tilde{M})| \mathcal{E}] 
\geq (\nicefrac{1}{2} + 4c) \optwgt$.
This yields that
\begin{align*}
\E[w(\tilde{M})] \geq \Pr[\mathcal{E}] \cdot \E[w(\tilde{M})| \mathcal{E}] 
                 \geq (1 - c)(\nicefrac{1}{2} + 4c) \optwgt 
  = (\nicefrac{1}{2} + 7c/2 -  4c^2) \optwgt 
\geq (\nicefrac{1}{2} + c) \optwgt.
\end{align*}
\end{proof}

What remains now is to bound the memory requirement of \ram.
We know from \Cref{lem:space-augpaths} that the instance \texttt{WAP} of 
\augpaths used in \ram uses at most $O(n \poly (\log n))$ memory.
Thus we only need to show that both the stack $S$ and the set $T$ also
use $O(n \poly (\log n))$ memory.

\paragraph{Bounding the size of $S$:}
% \footnote{We thank Sofya Vorotnikova for
%     presenting this proof in the workshop ``Communication Complexity and
%     Applications, II (17w5147)'' at the Banff International Research Station.
%     Note that this effectively shows that if we run the local ratio algorithm
%     over a random-order stream, the expected space usage is $O(n\poly(\log n))$,
%     and we get a $(\nicefrac{1}{2})$-approximation.}

Consider a state of the local-ratio algorithm where we have added some
edges to the stack and suppose that vertex potentials are $\alpha'_v$
for all $v \in V$.
For an edge $e = (u,v)$, let $w'(e) = w(e) - \alpha'_u - \alpha'_v$.
Let $E'$ be the set of remaining edges for which $w'(e) > 0$.
The next edge added to the stack by the local ratio algorithm is equally
likely to be any edge from $E'$.

For a vertex $v \in V$, let $d'_v$ be the number of edges incident to $v$
in $E'$. 
Consider a random edge $X$ selected as follows.
First pick vertex $v \in V$ with probability proportional to $d'_v$, and then
pick a uniformly random edge incident to $v$ in $E'$.
It is easy to see that $X$ is a uniformly random edge of $E'$.

Now fix a vertex $v$ and order the edges in $E'$ that are incident to $v$ in 
increasing order of $w'$.
Notice that if the local-ratio algorithm sees the $i$-th edge in ordering,
then it will be added to stack and, and since its weight get subtracted from 
each of the other incident edges, the weights of at least $i-1$ other edges
go below zero. This means that at least $i$ gets removed from $E'$ in the 
perspective of the local-ratio algorithm.
Let $R$ be the set of removed edges and let $E''$ be the set of remaining edges
after adding the next edge to $S$.
Then by the above reasoning, we have,
\begin{align*}
\E \left[ R \right] &\geq \sum_{v \in V} \Pr_X[\text{pick $v$}] \cdot 
\left( \sum_{i \in [d'_v]} \Pr_X[ \text{picking $i$-th edge incident to $v$} ] 
\cdot i \right) \\
	&= \sum_{v \in V} \frac{d'_v}{2 |E'|} \left( \sum_{i \in [d'_v]} 
	\frac{1}{d'_v} \cdot i \right)
	= \frac{1}{2|E'|} \sum_{v \in V} \frac{d'_v(d'_v + 1)}{2} 
	\geq \frac{1}{4|E'|} \sum_{v \in V} (d'_v)^2.
\end{align*}
But by Cauchy-Shwartz inequality, since $\sum_{v\in V}d'_v = 2|E'|$, we have 
that $$ \underbrace{\left(\sum_{v \in V} 1^2 \right)}_{n}  
\left(\sum_{v \in V}(d'_v)^2\right) \geq \left( \underbrace{ \sum_{v \in V} 
1 \cdot d'_v }_{2|E'|} \right)^2,$$ or equivalently,
$\sum_{v \in V} (d'_v)^2 \geq 4|E'|^2/ n$.
Hence $\E \left[R \right] \geq |E'| / n$ and the expected number of 
remaining edges, $\E \left[ E'' \right]$, at most $|E'|(1 - 1/n)$.

This yields that after picking $100 n \log n$ edges, the expected number of 
remaining edges is at most $|E|(1 - 1/n)^{100 n \log n} \leq 1/n^3$, and thus
by the Markov's inequality, size of $|S|$ is $O(n \log n)$ with high 
probability.

\paragraph{Bounding the size of $T$:}
We next show that $|T| \in O(n \poly (\log n))$ with high probability. To bound the size of $T$ at the end of the algorithm, we define events $B_{v,t}$ similarly to how  we defined $A_{v,t}$ in the proof of \cref{lem:sizes1}.

Recall that $A_{v,t}$ are defined to capture the number of unmatched neighbors of a vertex $v$ after processing the first $t$ edges. 
Define $B_{v,t}$ to be the event that at least $\log^2{n}$ edges $e$ incident to $v$ satisfy $w'_t(e) > 0$, where $w'_t(e)$ denote the value of $w'(e)$ just after processing the $t$-th edge of the stream. (Recall that $w'(e) = w(e) - \alpha_u - \alpha_v$ as defined in \cref{line:w'-definition} of \cref{alg:rand-arrival}.) In the rest of the proof, we show that $\prob{B_{v, t} | B_{v, t - 1}} \le \rb{1 - (\log^2{n})/(m - t + 1)}$, after which the claim follows as in the proof of \cref{lem:sizes1} for $p = 100 / \log{n}$.

If $B_{v, t}$ occurs, let $C_{v, t}$ be the set of edges corresponding to the $\log^2{n}$ largest \emph{positive} values $w'_t(e)$ over all the edges $e$ incident to $v$. Then, if $B_{v, t - 1}$ occurs and if the $t$-th edge from the stream is from $C_{v, t-1}$, then $B_{v, t}$ can not occur. Since each edge $e \in C_{v,t-1}$ is such that $w'_{t-1}(e) > 0$, $e$ appears in the stream after position $t-1$. So, given that $B_{v, t - 1}$ occurs, the probability that the $t$-th edge from the stream is in $C_{v,t-1}$ is $|C_{v,t-1}| / (m - t + 1)$, and hence
\[
	\prob{B_{v, t} | B_{v, t - 1}} \le \rb{1 - \frac{\log^2{n}}{m - t + 1}},
\]
as desired.

This gives the following lemma on the size of the stack $S$ and set $T$.

\begin{lemma} \label{lem:size-st}
Given that the edges arrive in a uniformly random order, with high probability,
both the local-ratio stack $S$ and the set $T$ will contain 
$O(n \poly (\log n))$ edges.
\end{lemma}

\Cref{lem:space-augpaths} and \Cref{lem:size-st} yield that the
algorithm \ram uses $O(n \poly (\log n))$ memory with high probability.

%%% Local Variables:
%%% mode: latex
%%% TeX-master: "000-main_random_order"
%%% End:

\section{\texorpdfstring{$(1-\eps)$}{(1-epsilon)}-Approximate Maximum Weighted Matching}
\label{sec:mpctot}

In this section, we reduce the problem of finding weighted augmenting paths in general graphs to that of finding \emph{unweighted} augmenting paths in \emph{bipartite} graphs.
Our reduction yields a ${(1 - \eps)}$\nobreakdash-approximation maximum weighted matching algorithm that can be efficiently implemented in both the multi-pass streaming model and the MPC model. 
We formalize this result as \cref{lemma:improving-the-matching-weight} below. Throughout the section, we use $\optWM$ to denote some fixed maximum weighted matching in the input graph, and we assume that edge weights are positive integers bounded by $\poly(n)$.

\begin{restatable}[General weighted to bipartite unweighted]{theorem}{theoremmainmpc}\label{lemma:improving-the-matching-weight}
	Let $\optWM$ be a maximum weighted matching and $M$ be any weighted matching such that $w(M) < w(\optWM) / (1 + \eps)$ for some constant $\eps$. There exists an algorithm that in expectation augments the weight of $M$ by at least $\eps^{O(1 / \eps^2)} \cdot w(\optWM)$ which can be implemented
	\begin{enumerate}
	\item in $U_M$ rounds, $O(m/n)$ machines per round, and $O_\eps(n \poly(\log n))$ memory per machine, where $U_M$ is the number of rounds used by a $(1-\delta)$-approximation algorithm for bipartite unweighted matching that uses $O(m/n)$ machines per round and $O_\delta(n \poly(\log n))$ memory per machine in the MPC model, and
	\item in $U_S$ passes and $O_\eps(n \poly (\log n))$ memory, where $U_S$ is the number of passes used by a $(1-\delta)$-approximation algorithm for bipartite unweighted matching that uses $O_\delta(n \poly (\log n))$ memory in the multi-pass streaming model,
      \end{enumerate}
      where $\delta = \eps^{28 + 900/\eps^2}$.
Using the algorithm of Ghaffari et al.~\cite{ghaffari2018improved} or that of Assadi et al.~\cite{assadi2017coresets}, we get that $U_M = O_\eps(\log\log n)$, and using the algorithm of Ahn and Guha~\cite{ahnguha}, we get that $U_S = O(\log\log(1/\delta)/\delta^2) = O((1/\eps)^{56+1800/\eps^2}\log(1/\eps))$.
\end{restatable}

It is easy to see that \cref{thm:mpc} follows directly from \cref{lemma:improving-the-matching-weight}. 
If the current matching is not $(1 - \eps)$-approximate, after a single run of the algorithm guaranteed by \cref{lemma:improving-the-matching-weight}, the weight of the current matching improves by at least $\eps^{O(1/\eps^2)} \cdot \optwgt$ in expectation.
Hence, it is sufficient to repeat the same algorithm for $(\nicefrac{1}{\eps})^{O(1/\eps^2)}$ iterations to get (in expectation) a $(1 - \eps)$-approximation.
Since each iteration can reuse the memory used by the previous iteration, the space requirement of the multi-pass streaming model and the memory-per-machine requirement in the MPC model remain unchanged.

As explained in \cref{sec:intro-mpcandstreaming}, a quick summary of our proof technique for \cref{lemma:improving-the-matching-weight} is as follows:
First we show that, if the initial matching $M$ we have is not close to optimal, then there exists a large-by-weight fraction of short augmentations, and these augmentations can be divided into several classes where augmentations in each class have comparable edge-weights and gains. 
We then show how to find weighted augmentations in each such class by reducing it to a bipartite matching problem. The reduction encompasses our layered graph construction and the filtering technique. Finally, we combine the augmentations recovered by this method in a greedy manner to significantly improve the current matching, and this yields the proof of \cref{lemma:improving-the-matching-weight}.

In the rest of this section, we elaborate on each of these steps: In \cref{sec:improving-matching-weight}, we first introduce the concept of augmentation classes to capture groups of short augmentations that have similar edge-weights and gains. Then, we formally state the two intermediate results: the first on augmentation classes containing augmentations that contributes to an overall gain of $\Omega(\eps^2) \cdot w(\optWM)$ (\cref{lemma:many-short-augmentations}) and the second on the existence of an efficient procedure to find many augmentations of those classes using a reduction to the unweighted bipartite setting (\cref{lemma:weight-class}). 
\cref{lemma:many-short-augmentations} and \cref{lemma:weight-class} now imply a simple algorithm for proving \cref{lemma:improving-the-matching-weight}: Run the algorithm given by \cref{lemma:weight-class} for each augmentation class (of geometrically increasing weight), and then greedily pick non-conflicting augmentations starting with the augmentation class of the highest weight.
We then analyze this algorithm and prove \cref{lemma:improving-the-matching-weight} assuming that \cref{lemma:many-short-augmentations} and \cref{lemma:weight-class} hold (whose proofs appear later).

In \cref{sec:short-augmentations}, we prove \cref{lemma:many-short-augmentations}. In fact, we prove a technical lemma that guarantees many-by-weight short augmenting paths and cycles of significant gains that also satisfy several additional constraints on edge weights (thus it is stronger than \cref{lemma:many-short-augmentations}). 
This  lemma, while implying \cref{lemma:many-short-augmentations}, also assists in proving \cref{lemma:weight-class}, as the additional constraints on edge weights of the augmentations make sure that many of those augmentations are captured by our reduction.

In the more involved \cref{sec:finding-short-augmentations}, we present the precise construction of the layered graphs we introduced in \cref{sec:intro-mpcandstreaming}, and we explain our filtering technique in detail. We then show how exactly the unweighted augmenting paths in the layered graphs relate to weighted augmenting paths and cycles of the original graph. 
Finally, in \cref{sec:combine}, we put together the results from \cref{sec:short-augmentations} and \cref{sec:finding-short-augmentations} to prove \cref{lemma:weight-class}.

Throughout the analysis, we assume that $\eps < 1/16$, and also extensively use the following definitions.
We begin with the definition of alternating paths and cycles.

\begin{definition}[Alternating paths and cycles]
Let $M$ be a matching. A path $P$ is said to be \emph{alternating} if its edges alternate between $M$ and $E \setminus M$. The first edge of $P$ can be in $M$ or $E \setminus M$.
\\
Similarly, a cycle $C$ is alternating if its edges alternate between $M$ and $E \setminus M$. 
\end{definition}
Observe that from the definition, an alternating cycle has even length and an alternating path can be of even or odd length. 

In our analysis, we sometimes consider alternating paths such that an endpoint of a path $P$ is incident to a matched edge $e$ that is not on the path $P$. For instance, let $P = v_1 v_2 v_3$ and $\{v_2, v_3\} \in M$, and suppose that there is another edge $\{v_0, v_1\} \in M$. Now, if we wish to add $\{v_1, v_2\}$ to the matching, we should remove both $\{v_0, v_1\}$ and $\{v_2, v_3\}$. Hence, adding some edges of a path to the matching might involve removing some edges which are not on the path. To capture this scenario, we define the following notion:
\begin{definition}[Matching neighborhood]
Let $C$ be an alternating path or an alternating cycle with respect to $M$. Then, by $\Nmatching{C}$, we denote all the edges of the matching $M$ incident to the vertices of $C$, including those lying on $C$ itself.
Note that if $C$ is a cycle, then $\Nmatching{C} = C \cap M$.
\end{definition}

For completeness, we also define the usual notions of \emph{applying an augmentation} and the \emph{gain of an augmentation} below. 

\begin{definition}[Applying augmentation]
	Let $C$ be an alternating path or an alternating cycle. Let $A = \Nmatching{C}$ and $B = C \setminus M$. Then, by \emph{applying} $C$ we define the operation in which $A$ is removed from $M$ and $B$ is added to $M$.
\end{definition}

\begin{definition}[Gain of augmentation]
	Let $C$ be an alternating path or an alternating cycle. Then, the \emph{gain} $\gain(C)$ of $C$ denotes the increase in the matching weight if $C$ is applied.
\end{definition}

Note that an augmentation usually means an alternating path or cycle whose unmatched edges have a larger total weight than that of the edges in its matching neighborhood. However we sometimes consider cases where each individual `augmentation' does not satisfy this, but collectively they do. (For example consider the
single edge alternating paths $v_1v_2$ and $v_3v_4$ where $v_1$ and $v_4$ are unmatched vertices and $v_2$ is
matched to $v_3$. If $w(\{v_1,v_2\}) = w(\{v_3,v_4\}) = 2$ and $w(\{v_2,v_3\}) = 3$, then applying both the augmentations gives a gain of one whereas applying either one of them individually is not beneficial.)

%\input{600-unweighted-matching}
%\section{Type of augmenting paths}
\subsection{The Main Algorithm}
\label{sec:improving-matching-weight}

In this section, we present our main algorithm and prove \cref{lemma:improving-the-matching-weight} assuming the two intermediate results that we prove in the later sections.
The first one claims that if the current matching is not $(1 - \eps)$-approximate, then there exist many-by-weight short vertex-disjoint augmentations (\cref{lemma:many-short-augmentations}) that have comparable gains and edge weights.
The second one claims that, for a given weight $W$, we can efficiently find many-by-weight short augmentations whose edge weights and gains are comparable to $W$ (\cref{lemma:weight-class}). 
We begin by defining augmentation classes, which are collections of augmentations whose gains and individual edges are similar in weight.

\begin{definition}[Augmentation class]\label{definition:augmentation-class}
Fix a weight $W$, and let $M$ be the current matching. 
By the \emph{augmentation class} of $W$ we refer to the collection of all augmentations (not necessarily vertex-disjoint) such that each augmentation $C$ has the following properties:
	\begin{enumerate}
			\item The weight of each edge of $C$ is between $\layerepspower W$ and $2W$.
		\item The gain $\gain(C)$ of $C$ is at most $2W$.
    \item When the weight of each edge in $\Nmatching{C}$ (recall that $\Nmatching{C}$ is the matching neighborhood of $C$) is rounded up and the weight of each unmatched edge in $C$ (i.e., $C\setminus \Nmatching{C}$) is rounded down to the nearest multiple of $\layerepspower W$, the gain of such $C$ is at least $\layerepspower W$. 
    \item The augmentation $C$ consists of at most $64 / \eps^2 + 1$ vertices.
	\end{enumerate}
\end{definition}

By the third property, the gain of an augmentation $C$ without any rounding is also at least $\eps^{12} W$. 
The following theorem says that if $M$ is not close to optimal, then there is a collection of vertex-disjoint augmentations, each of which belongs to some augmentation class, and collectively they have a large gain; we prove this theorem in \cref{sec:short-augmentations}. 

\begin{theorem}[Significant weight in augmentation classes]\label{lemma:many-short-augmentations}
  Let $M$ be a matching such that $w(M) < (1 - \eps) \optwgt \le \optwgt/(1 + \eps)$ (i.e., $M$ is
  not $(1 - \eps)$-approximate).
  Then, there exists a collection $\cC$ of vertex-disjoint augmentations with the following properties:
	\begin{itemize}
		\item Each augmentation $C \in \cC$ is in the augmentation class
                of $(1 + \eps^4)^i \le w(C)$ for at least one $i \in \bbN$.
		\item It holds that $\gain(\cC) \ge \eps^2 \optwgt/200$.
	\end{itemize}
\end{theorem}

In the following result, we essentially claim that if a given augmentation class does not already contain many-by-weight edges of $M$, then there is an efficient procedure that finds many-by-weight vertex-disjoint augmentations in that class.
\begin{theorem}[Single augmentation class]\label{lemma:weight-class}
	%Let $M$ be the current matching. Assume that $w(M) < \woptWM / (1 + \eps)$. Let $\cC$ be the family of augmentations as defined by \cref{lemma:shortpaths}. Consider a family of weights of the form $(1 + \eps^4)^i$, for $i \in \bbN$, and consider an assignment of each of the elements of $\cC$ to one or more of those weights.
	
	Let $M$ be the current matching. Assume that $w(M) < \optwgt / (1 +
        \eps)$. Let $\cC_W$ denote a collection of vertex-disjoint augmentations
        belonging to the augmentation class of $W$. Define $w(M_W)$ to be the total weight of the edges of $M$ with weights in $[\layerepspower W, 2 W]$.
	Then there is an algorithm that, given $W$, outputs a collection $\cA_{W}$ of vertex-disjoint augmentations ($\cA_W$ is not necessarily a subset of $\cC_W$) having the following properties:
	\begin{enumerate}[(A)]
		\item\label{item:cAW-is-in-aug-class} $\cA_W$ is a subset of the augmentation class of $W$.
%		\item\label{item:weights-gain} For each augmentation $C \in \cA_{W}$ it holds $\layerepspower W \le \gain(C) \le 2 W$.
%		\item\label{item:weights-path-length} Each augmentation of $\cA_{W}$ consists of at most $64 / \eps^2 + 1$ vertices.
		\item\label{item:weights-partial-gain} In expectation, $\gain(\cA_{W}) \ge \eps^{c / \eps^2} (\gain(\cC_W) - \eps^{10} w(M_{W}))$, for some constant $c$.
	\end{enumerate}

	This algorithm can be implemented in $U_M$ MPC rounds with $O_{\eps}(n \log{n})$ memory per machine, and $U_S$ passes and $O_\eps(n \poly (\log n))$ memory in the streaming model, where $U_M$ and $U_S$ are defined in~\cref{lemma:improving-the-matching-weight}.
\end{theorem}

Let $\cC$ be the family of augmentations as defined in \cref{lemma:many-short-augmentations}, so applying $\cC$ increases the matching weight by $\eps^2 \optwgt/200$. Consider all the weights of the form $(1 + \eps^4)^i$, for $i \in \bbN$. Property~\eqref{item:weights-partial-gain} of \cref{lemma:weight-class} implies that there is an algorithm that for those weights finds augmentations whose sum of gains, when applied \emph{independently}, is in expectation at least $\eps^{O(1/\eps^2)} \gain(\cC)$ up to some additive loss. (This additive loss is significant only if there is already a significant weight in the matching $M_{W}$.) However, even if that additive loss is negligible, when those augmentations are applied simultaneously they might intersect.

But, we still manage to find a set of non-intersecting augmentations of significant total gain by following a simple greedy strategy; we consider augmentation classes in decreasing order of weight and apply only those augmentations that do not intersect with previously applied ones. This approach retains a considerable fraction of the gain since the augmentations we consider are short (thus, for a given augmentation, the number of conflicting augmentations in a given augmentation class is small), and since the weights of augmentation classes, and consequently, the maximum gains of augmentations in those classes are geometrically decreasing.

\cref{alg:mpc} implements this approach, and we analyze it next to prove our main result, \cref{lemma:improving-the-matching-weight}, assuming that we already have  \cref{lemma:many-short-augmentations} and \cref{lemma:weight-class}.

\newcommand{\epsilonA}{\eps^4}
\newcommand{\txtmax}{\operatorname{max}}
\newcommand{\imax}{i_{\txtmax}}

\begin{algorithm}
	\DontPrintSemicolon
  \SetKwInOut{Input}{Input}
  \SetKwInOut{Output}{Output}
  \SetKwFunction{push}{push}
  \SetKwFunction{add}{add}
  \SetKw{KwAnd}{and}
  \SetKw{KwOr}{or}

  \Input{A weighted graph $G$,
		approximation parameter $\eps$,
		the current matching $M$}  %Done to match the style of Algorithm 2.
  \Output{A matching of $G$}

  \BlankLine
	
  % Let $\imax$ be the smallest integer such that for every $W$ whose augmentation
  % class contains an augmentation of $\cC$ defined in
  % \cref{lemma:many-short-augmentations}, $W \le (1 +
  % \eps^4)^{\imax + 1}$.
  $\imax \gets \lceil \log_{1+\eps^4}\left((64/\eps^2+1)\cdot\max\{w(e) : e \in E\} \right)\rceil$.

  \tcc{One MPC round or one pass can be spent to compute $\max\{w(e) : e \in E\}$.}
  \label{line:define-imax}
	
	$\cW \gets \{ (1 + \eps^4)^i : i = 0, 1, \dots, i_{\txtmax} \}$\label{line:generate-weights}

	%$R \gets O\rb{\eps^{-c / \eps^2}}$, for some constant $c > 0$

	%\For{$R$ steps \label{line:the-main-alg-main-loop}} {		
				
		\For{each $W \in \cW$ in parallel} {
			Let $\cA_W$ be the set of augmentations that the algorithm of \cref{lemma:weight-class} outputs for the augmentation class $W$. \label{line:apply-to-each-W}
		}		
%		
%		\tcc{Resolving the conflicts between $\cA_W$ for distinct $W$.}
%		
%		$\istar \eqdef \lceil 20 \log_{1 + \eps^4} (1 / \eps) \rceil$ \label{line:define-istar}
%		
%		$\cWbase \eqdef \{(1 + \eps^4)^i : 0 \le i < \istar\}$ \label{line:define-cWbase}
%		
%		$\cF(\hW) \eqdef \{(1 + \eps^4)^{j \istar} \hW \ |\ j \in \bbN\}$, for each $\hW \in \cWbase$ \label{line:define-cFW}
%		
%		\For{each $\hW \in \cWbase$ in parallel \label{line:for-hW-in-cWbase}} {
%			
%			$\cA^{\hW} \gets \emptyset$ \label{line:init-cAW}
%			
%			\For{each $W \in \cF(\hW)$ in decreasing order} {
%				Add to $\cA^{\hW}$ all the augmentations of $\cA_{W}$ that do not intersect with the augmentations that are already in $\cA^{\hW}$. \label{line:update-cAW}
%			}
%		}
%		
%		Let $\cA^{\Wstar}$, for $\Wstar \in \cWbase$, be one that maximizes gain. Apply $\cA^{\Wstar}$ to $M$. \label{line:apply-one-weight}
%	%}

	$\hat{\cA} \gets \emptyset$ \;
	\For{each $W \in \cW$ in decreasing order} {
		\For{each augmentation $C$ in $\cA_W$}{
			Add $C$ to $\hat{\cA}$ if $C$ does not conflict with any other augmentation in $\hat{\cA}$.
		}
	}
	
	\Return the matching obtained after applying the augmentations in $\hat{\cA}$ to $M$.
  \caption{Algorithm $\MainAlg$ for improving matching weight as described by \cref{lemma:improving-the-matching-weight}}
  \label{alg:mpc}
\end{algorithm}

\newcommand{\gainAll}{W^+_{\operatorname{all}}}
\begin{proof}[Proof of \cref{lemma:improving-the-matching-weight}]
The theorem follows from the analysis of \cref{alg:mpc}. 
%Recall that $\cW$ is the set of augmentation classes considered by \cref{alg:mpc}, and $\cA_W$ is the set of augmentations returned for a augmentation class $W$ by the algorithm of \cref{lemma:weight-class}. 
Recall that $\cC_W$ is the augmentations of $\cC$ ($\cC$ is defined in \cref{lemma:many-short-augmentations} and satisfies $\gain(\cC) \geq \eps^2 \optwgt / 200$) that are also in the augmentation class $W$ and $M_W$ is the set of matching edges whose weights are between $\eps^{12}W$ and $2 W$. 

Let $\gainAll$ be the total gain of all augmentations that the algorithm finds in \cref{line:apply-to-each-W}. I.e., $\gainAll = \sum_{W \in \cW} \gain(\cA_W)$, where $\cW$ is the set of weights of all augmentation classes considered by the algorithm. By \cref{lemma:weight-class}, we have that
$\gain(\cA_W) \geq  \eps^{c/\eps^2} \left(\gain (\cC_W) - \eps^{10}w(M_W) \right)$, which yields
\begin{equation}
\gainAll \geq \eps^{c/\eps^2} \left( \sum_{W \in \cW} \gain (\cC_W) - \eps^{10} \sum_{W \in \cW} w(M_W) \right).
\label{eq:totalgain}
\end{equation}

Notice that for two weights $W_1$ and $W_2$,
if $\eps^{12} W_1 > 2 W_2$, then $M_{W_1}$ and $M_{W_2}$ do not intersect. 
Since we consider weights of the form $(1 + \eps^{4})^i$, any matching edge can be in $M_W$ for at most $\lceil \log_{1 + \eps^4} (2/\eps^{12})\rceil \le 1/\eps^{6}$ (we assumed $\eps < 1/16$) different weights $W$. This yields that 
$$\eps^{10}  \sum_{W \in \cW} w(M_W) < \eps^{10} (1/\eps^{6}) w(M) \leq (\eps^2/256) w(M),$$
where the last inequality follows from the assumption that $\eps < 1/16$.

On the other hand, by \cref{lemma:many-short-augmentations}, the term $\sum_{W \in \cW} \gain (\cC_W)$ is at least $\eps^2 \optwgt / 200$. Substituting in \cref{eq:totalgain}, we get,
\begin{equation*}
\gainAll \geq \eps^{c/\eps^2} \left( \eps^2 \optwgt / 200 - (\eps^2/256) w(M) \right) \geq \eps^{c'/\eps^2} w(\optWM)
\end{equation*}
for some constant $c' > 0$.

Now fix some augmentation class $W_i = ( 1 + \eps^4)^i$ and an augmentation $C$ in $\cA_{W_i}$.
By definition, $\gain(C) \geq \eps^{12}( 1 + \eps^4)^i $, and for any other augmentation class $W_j = (1 + \eps^4)^{j}$, the maximum gain of any augmentation in $\cA_{W_j}$ is at most $2 W_j = 2 (1 + \eps^4)^j$. 
If we apply $C$, it blocks at most $64/\eps^2 + 1$ other augmentations in each of the augmentation classes below it. 
Thus the total gain of the blocked augmentations if $C$ is applied is at most
\begin{align*}
\sum_{j < i} (64/\eps^2 + 1) \cdot 2 \cdot (1 + \eps^4)^j &\leq  (130/\eps^2) (1 + \eps^4)^{i-1} \sum_{j = 0}^{\infty} (1 + \eps^4)^{-j} \\
&= (130/\eps^2) (1 + \eps^4)^{i-1} \frac{1}{1 - 1/(1 + \eps^4)}  \\
&= (130/\eps^6) (1 + \eps^4)^i \leq (130/\eps^{18})  \gain(C),
\end{align*}
and this means that the gain we retain by our greedy strategy, $\gain(\hat{A})$, is at least  $\gainAll / (1 + 130/\eps^{18}) \geq \eps^{c''/\eps^2} w(\optWM)$ for some constant $c'' > 0$.

\paragraph{MPC implementation:}

Since the maximum edge weight is $\poly(n)$, the number of different augmentation classes we consider, i.e., $\imax + 1$, is $O(\log_{1 + \eps^4} n)$. Hence can implement \cref{line:apply-to-each-W}  in $O(m/n)$ machines with $O_\eps(n \poly(\log n))$ memory by running the algorithm of \cref{lemma:weight-class} (i.e., \Cref{alg:fixed-weight}) in parallel for each augmentation class.

For each augmentation class $W$, the collection of augmentations $\cA_W$ is vertex disjoint, and hence requires $O(n)$ memory. Thus all the collection $\cA_W$ for all augmentation classes require $O_\eps(n \poly(\log n))$ memory, and hence they can be collected in a single round into a single machine, and the greedy strategy can be run in that machine.
%
%The algorithm first generates a list $\cW$ of all the weights of short augmentations that we consider. In particular, as we will see in the sequel and as it is apparent from the definition of augmentation classes, each such augmentation has length at most $65 / \eps^2$. Let $\wmax$ be the maximum edge-weight. We set
%\[
%	i_{\txtmax} \eqdef \lceil \log_{1 + \eps^4} \rb{65 \wmax / \eps^2} \rceil.
%\]
%From \cref{lemma:many-short-augmentations}, this definition of $i_{\txtmax}$ implies that the algorithm considers weights whose augmentation classes retain large gain.
%
%We assume that the smallest weight is $1$ (which can always be achieved by scaling) and that $\wmax \le n^3$. Observe that the second assumption is without loss of generality, as we can remove all the edges which are factor $n^3$ or more smaller than $\wmax$. Those edges affect $\optWM$ by at most a $1/n^2$ fraction. This all implies that $|\cW| \in O_{\eps}(\log{n})$.
%
%The implementation of \cref{line:apply-to-each-W} is provided by \cref{lemma:weight-class}.
%
%It only remains to comment on the implementation of \crefrange{line:init-cAW}{line:update-cAW}. As each $\cA_{W}$ consists of vertex-disjoint augmentations, the size of all the augmentations over all the elements of $\cF(\hW)$ is $O(n |\cW|) \in O_{\eps}(n \log{n})$. So, we deliver all those augmentations to a single machine and apply a sequential algorithm over them to obtain $\cA^{\hW}$. 

\paragraph{Streaming implementation:}
The implementation is quite straightforward for the streaming setting.  For each $W \in \cW$, an instance is created, in which \Cref{alg:fixed-weight} is run.  Note that there are $O(\log_{1+\eps^4}n)$ such instances.  All the outputs ($|\cW|$ of them) are then stored.  The greedy conflict resolution that is done afterwards can be done using these stored outputs without using any pass over the stream.  So the number of passes used is same as that used by \Cref{alg:fixed-weight}, and memory used is $O(\log_{1+\eps^4}n)$ times that used by \Cref{alg:fixed-weight} (see \cref{lemma:weight-class}).
\end{proof}

\subsection{Existence of Many-by-weight Short Augmentations}
\label{sec:short-augmentations}

In this section we show that if the current matching is not a $(1 - \eps)$-approximate one, then there exists a large-by-weight number of short vertex-disjoint augmentations. Moreover, we show that many of those augmentations $C$ have the following properties: the weight of each edge of $C$ (matched or unmatched) is $\Omega(\poly(\eps) \cdot w(C))$ (Properties~\ref{item:property-large-edges} and~\ref{item:property-large-matched-edges} of \cref{lemma:shortpaths}); and, $\gain(C)$ is large (Property~\ref{item:property-path} of \cref{lemma:shortpaths}). This implies that $C$ belongs to some augmentation class, e.g., to an augmentation class of $w(C)$ rounded down to $(1 + \eps^4)^i$ (a formal argument of this appears after the statement of the lemma). Hence, the following lemma implies that the augmentation classes all combined contain a collection of vertex-disjoint augmentations of large weight. 

%As we describe in the sequel, we exploit this property by enumerating over all possible types of short augmentations and then applying a type that has the largest gain.
%\commentb{I would suggest adding a couple more sentences here elaborating the importance of this (i.e., the additional properties B, C) in our algorithm.}

\begin{restatable}{lemma}{lemmashortpaths}\label{lemma:shortpaths}
  Let $M$ be a matching such that $w(M)\le w(\optWM)/(1 + \eps)$ where $\eps \le 1/16$.  
  Then there exists a collection $\cC$ of vertex-disjoint augmentations with the following properties:
  \begin{enumerate}[(A)]
		\item\label{item:property-length} Each $C \in \cC$ is such that $C \cup \Nmatching{C}$ consists of at most $4 / \eps$ edges.
		%Moreover, if $C$ is not a cycle, then $C$ starts and ends in $M$-edges.
		\item\label{item:property-large-edges} For every $C \in \cC$ and every edge 
			$e \in C \cap \optWM$, $w(e) \ge (\eps^2/64) \cdot w(C)$.
		\item\label{item:property-large-matched-edges} For every $C \in \cC$ and every edge 
			$e \in C \cap M$, $w(e) \ge (\eps^6/64) \cdot w(C)$.
    \item\label{item:property-path} For every $C \in \cC$, we have that
			\[
				w(C\cap \optWM) \ge (1+\eps/8)\cdot w(\Nmatching{C}).
			\]
    \item\label{item:property-collection} The sum of gains of the elements of $\cC$ is at least $\eps^2 w(\optWM)/200$. That is
			\[
				\sum_{C \in \cC}\left( w(C\cap \optWM) - w(\Nmatching{C}) \right)\ge \eps^2 w(\optWM)/200.
			\]
  \end{enumerate}
\end{restatable}
The proof of~\cref{lemma:shortpaths} is a simple adaptation of the proof of the known fact that a  matching has many short augmentations if its value is less than $(1-\epsilon)$ times the value of an optimal matching. It is provided in~\cref{sec:proof-of-short-paths}.

We now formally argue that the above lemma implies~\cref{lemma:many-short-augmentations}. Recall the statement of that theorem: if $w(M) \leq w(\optWM)/(1+\eps)$ then there exists a collection $\cC$ of vertex-disjoint augmentations with the following properties:
\begin{itemize}
  \item Each augmentation $C \in \cC$ is in the augmentation class
              of $(1 + \eps^4)^i \le w(C)$ for at least one $i \in \bbN$.
  \item It holds that $\gain(\cC) \ge \eps^2 w(\optWM)/200$.
\end{itemize}
The second item is the same as Property~\eqref{item:property-collection} and the first item follows because, for $C \in \cC$, if we let $W$ be $w(C)$ rounded down to the closest power of $(1+\eps^4)$ then the following holds:
\begin{enumerate}
  \item The weight of each edge of $C$ is between $\layerepspower W$ and $2W$ by selection of $W$ and Properties~\eqref{item:property-large-edges},\eqref{item:property-large-matched-edges}.
  \item The gain $\gain(C)$ of $C$ is at most $w(C) \leq 2W$.
  \item When the weight of each matched edge (i.e., an edge in $M$) of $\Nmatching{C}$ (recall that $\Nmatching{C}$ is the matching neighborhood of  $C$) is rounded up and the weight of each unmatched edge of $C$ is rounded down to the nearest multiple of $\layerepspower W$, the gain of such $C$ is at least $\layerepspower W$. This holds because by Properties~\eqref{item:property-length} and~\eqref{item:property-path} we have that the gain \emph{after} the rounding is at least
      \begin{align*}
      	w(C\cap \optWM) - w(\Nmatching{C})   - \layerepspower W \cdot 4/\eps \gg \layerepspower W\,. 
      \end{align*}

    \item $C$ consists of at most $4/\eps \leq 64 / \eps^2 + 1$ vertices by Property~\eqref{item:property-length}.
\end{enumerate}
Hence, $C$ is in the augmentation class of $(1 + \eps^4)^i \le w(C)$ for at least one $i \in \bbN$.

As can be seen in the above calculations,~\cref{lemma:shortpaths} is more restrictive than that required by the definition of an augmentation class. The reason is as follows.~\cref{lemma:shortpaths} shows the existence of very structured short augmentations that have a large total gain. However, no procedure for finding those augmentations is given.  In the proof of~\cref{lemma:weight-class} we will give such a procedure that efficiently finds augmentations that satisfy looser guarantees than those of~\cref{lemma:shortpaths}.  These relaxed properties of augmentations correspond to the definition of augmentation classes.  The more restrictive guarantees of~\cref{lemma:shortpaths} are then used  to show that, for each augmentation class, the efficient procedure finds in expectation a set of vertex-disjoint augmentations with a gain comparable to that promised by~\cref{lemma:shortpaths} (see~\cref{lemma:paths-and-cycles-in-layered}).

% Observe that \cref{lemma:shortpaths} describes augmentations with properties more restricted than those guaranteed to exist by \cref{lemma:many-short-augmentations}. We now comment on the reason to state \cref{lemma:many-short-augmentations} in addition to \cref{lemma:shortpaths}. In our analysis, \cref{lemma:shortpaths} serves to show existence of many-by-weight short augmentations. In \cref{sec:finding-short-augmentations} we show how to find many-by-weight short augmentations (not necessarily those defined by \cref{lemma:shortpaths}). To achieve that, we design a structure (that we call \emph{layered graphs}) that captures the augmentations defined by \cref{lemma:shortpaths} by considering augmentations with smaller gain and smaller edge-weights than stated by \cref{lemma:shortpaths}. These relaxed properties of augmentations correspond to the definition of augmentation classes and respectively to \cref{lemma:many-short-augmentations}.

%The proof of \cref{lemma:shortpaths} is provided in \cref{sec:proof-of-short-paths}.

%%% Local Variables:
%%% mode: latex
%%% TeX-master: "000-main_random_order"
%%% End:

\subsection{Finding Short Augmentations}
\label{sec:finding-short-augmentations}

In this section, we dive in to the details of the construction of our layered graphs and the filtering technique we introduced in \cref{sec:intro-mpcandstreaming}. For this, we first parameterize the graph in terms of a random bipartition and the current matching (\cref{section:parametrization}). Then, in \cref{section:layered-graph}, we present the formal definition of a layered graph, and in \cref{section:tau-pairs}, we explain the filtering technique. Later, in \cref{sec:existence-of-short-augmentations}, we show that our construction captures many of the paths described by \cref{lemma:shortpaths}.

\subsubsection{Graph parametrization}\label{section:parametrization}
As a reminder, our goal is to reduce the problem of finding weighted augmentations to the problem of finding unweighted augmenting paths. As the first step in this process, we randomly choose a bipartite subgraph of the input graph. The graph obtained in this way is referred to as \emph{parametrized}. We now describe this step.

\paragraph{Bipartiteness:}
	%After directing the edges
	Given $V$, we construct two disjoint sets $L$ and $R$ by uniformly at random assigning each vertex of $V$ to either $L$ or $R$.
	
	We then consider only those edges whose one endpoint is in $L$ and the other is in $R$, and define
	\begin{itemize}
		\item $A \eqdef M \cap (L \times R)$, i.e., $A$ consists of the matching edges that connect $L$ and $R$,
		\item $B \eqdef (E \setminus M) \cap (L \times R)$, i.e., $B$ consists of the unmatched edges that connect $L$ and $R$.
	\end{itemize}

\paragraph{Parametrized graph:}
	We say that a given graph is \emph{parametrized} if each vertex is assigned to $L$ or $R$ as described above. Given graph $G = (V, E)$ and matching $M$, we use $G^P = \partupple$ to denote its parametrization.

\subsubsection{Layered graph}\label{section:layered-graph}
	We now introduce the notion of \emph{layered graph}, that plays a key role in enabling us to turn an algorithm for finding unweighted augmenting paths into an algorithm for finding weighted augmentations. We provide an example of such graphs in \cref{figure:layered-graphs}.
	
	\begin{figure}
		\begin{center}
			\includegraphics[scale=0.7]{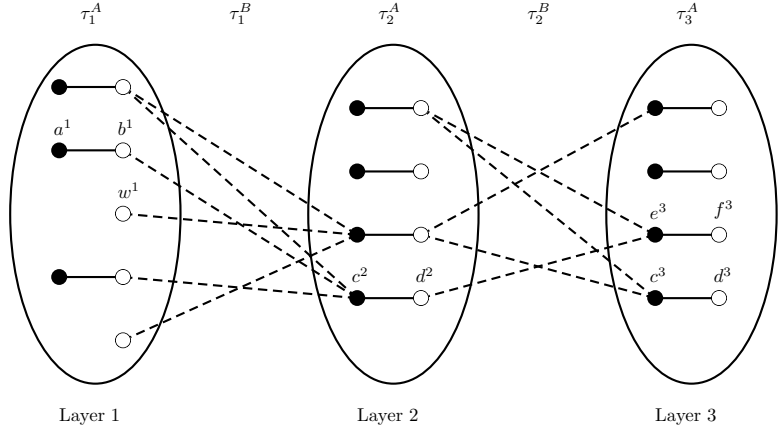}
		\end{center}
		\caption{\label{figure:layered-graphs} This figure depicts a layered graph $\Layered$ consisting of $3$ layers. In this example, we show only those vertices that have at least one edge of $\Layered$ incident to it. Full segments represent matched and dashed segments represent unmatched edges filtered in $\Layered$. The black vertices are in $L$ while the white vertices are in $R$. Pictorially, we think of a layered graph evolving from left to right. Notice that, since $(c, d)$ appears in both the $2$-nd and the $3$-rd layer, $\tau_2^A$ equals $\tau_3^A$.}
	\end{figure}

	\begin{definition}[Weighted layered graph]\label{definition:layered-graph}
		Let $G^P = \partupple$ be a parametrized graph. Recall that $A$ is a subset of matched and $B$ is a subset of unmatched edges. Let $\tau^A \in \bbR_{\ge 0}^{k + 1}$ and $\tau^B \in \bbR_{\ge 0}^{k}$ be two sequences of non-negative multiples of $\layerepspower$. Let $w : A \cup B \to \bbR_{\ge 0}$ be a weight function, and $W$ be a positive weight. Then, we use $\Layered(\tau^A, \tau^B, W, G^P) = (\VLayered, \ELayered)$ to denote \emph{layered graph} which is defined in two stages. First, we define $\VLayered$ and $\ELayered$ as follows
		\begin{itemize}
			\item $\VLayered = \{v^t, 1 \le t \le k + 1: v \in V\}$; in other words, $\VLayered$ represents the vertex set $V$ copied $k + 1$ times. We use \emph{Layer $t$} to refer to the $t$-th copy of the vertices of $V$.
			\item $\ELayered = X \cup Y$, where $X$ and $Y$ are defined as follows
				\begin{itemize}
					\item $X = \{\{u^t, v^t\} : 1 \le t \le k + 1, \{u, v\} \in A, \text{ and } w(\{u, v\}) \in (\rb{\tau^A_t - \layerepspower} W, \tau^A_t W]\}$, i.e., among the edges in layer $t$ we keep only those whose weight is relatively close from below to the threshold value $\tau^A_t W$.
					\item $Y = \{\{u^t, v^{t + 1}\} : 1 \le t \le k, \{u, v\} \in B, u \in R, v \in L, \text{ and }  w(\{u, v\}) \in [\tau^B_t W, \rb{\tau^B_t + \layerepspower} W)\}$, i.e., among the edges connecting layer $t$ and layer $t+1$ keep only those that are in $B$ (i.e., unmatched), that go from $R$ in layer $t$ to $L$ in layer $t+1$, and whose weight is relatively close from above to the threshold value $\tau^B_t W$.
				\end{itemize}
		\end{itemize}
		% In the second stage, we filter some edges that are included in the definition of $Y$.  Let $Y^i = Y \cap \{v^i : v \in V\}\times\{v^{i+1} : v\in V\}$, i.e., the set of $Y$-edges between $i$th layer and $(i+1)$th layer.
		% \begin{itemize}
		% 	\item \textbf{Filtering step for $\bigcup_{i=2}^{k-1}Y^i$.} If $e\in \bigcup_{i=2}^{k-1}Y^i$ is such that $e$ has zero or one $X$-edge incident on it, then remove $e$.
		% 	\item \textbf{Filtering step for $Y^1\cup Y^k$.} Let $e=\{u^1,v^2\}\in Y^1$.  Remove $e$ if one of the following holds:
                %         \begin{itemize}
                %           \item if $e$ has no $X$-edge incident on it in the second layer;
                %           \item if $e$ has no $X$-edge incident on it in the first layer \emph{and} $u$ is matched in $M$ or $\tau^A_1>0$.
                %         \end{itemize}
                %           (Conversely, keep $e$ only if it has two $X$-edges incident on it \emph{or} an $X$-edge incident on it in the second layer and $u$ is unmatched in $M$ and $\tau^A_1 = 0$.)  Analogously filter edges in $Y^k$.
		% \end{itemize}
		In the second stage, we filter some of the vertices from $\Layered$.
		\begin{itemize}
			\item \textbf{Filtering step for intermediate layers.} For $i \in \{2,3,\ldots, k\}$ and $v \in V$, remove $v^i$ if it is unmatched in $X$.
			\item \textbf{Filtering step for the first and the last layer.} For every vertex $v^1$ such that $v \in R$ and $v^1$ has no matched edge in layer $1$ incident to it: keep $v^1$ only if $v$ is not incident to $M$ and $\tau^A_1 = 0$; otherwise remove $v^1$ from $\VLayered$.  Analogously process every vertex $v^{k+1}$, i.e., if $v \in L$ and $v^{k + 1}$ has no matched edge in layer $k + 1$ incident to it: keep $v^{k + 1}$ only if $v$ is not incident to $M$ and $\tau^A_{k + 1} = 0$; otherwise remove $v^{k + 1}$ from $\VLayered$.
		\end{itemize}
	\end{definition}
		When it is clear from the context, we use only $\Layered$ to denote $\Layered(\tau^A, \tau^B, W, G^P)$. We refer a reader to \cref{figure:layered-graphs} for an illustration of layered graphs.
		
		Now we elaborate on why some of the vertices of $\Layered$ are filtered. First, we use layered graphs to find weighted augmentations via unweighted augmenting paths. The main idea here is to set $\tau^A$ so that sum of its elements is less than the sum of the elements of $\tau^B$. Intuitively, it guarantees that any alternating path that passes through all the layers could be used to improve the matching weight. Now, unlike in unweighted, in the weighted case a path can be weighted-augmenting even if on the path lay more matched than unmatched edges, e.g., path $a^1 b^1 c^2 d^2 e^3 f^3$ in \cref{figure:layered-graphs} if $w(\{a^1, b^1\}) + w(\{c^2, d^2\}) + w(\{e^3, f^3\}) < w(\{b^1, c^2\} + w(\{d^2, e^3\})$. The point of the first and the last layer of $\Layered$ is exactly to capture this type of scenarios. However, sometimes there is a vertex in one of these layers, e.g., the first layer, that does not have any matched edge in $\Layered$ incident to it, as it is the case with $w^1$ in \cref{figure:layered-graphs}. This might happen for two reasons. First, $w$ is not incident to any matched edge in $G$, in which case we keep $w^1$ only if $\tau^A_1 = 0$. (This is the same as saying that $w^1$ is incident to a zero-weight matched edge.) The second case if $w$ is incident to a matched edge $e$ in $G$, but $w(e) \notin (\rb{\tau^A_1 - \layerepspower} W, \tau^A_1 W]$. In this case, we should remove $w^1$ from $\Layered$ as otherwise it might not capture the case outlined above. For similar reasons vertices are removed from the last layer of $\Layered$.  Furthermore, to make sure that a matching returned by the unweighted algorithm gives us augmenting paths that pass through all layers (exactly once), we remove the vertices left unmatched by $X$ in the intermediate layers.  Thus we have no free vertices in the intermediate layers, therefore an augmenting path must start or end only in the first or the last layer.

	%\paragraph{Remark:} Layered graph is undirected. However, conveniently, tail-matched and head-matched types impose a natural direction of the edges. Therefore, although undirected, sometimes we will refer to these edges as arcs by which we will also fix their type. For instance, if we refer to a matched edge $\{u, v\}$ as arc $(u, v)$ it means that $u$ is tail-matched and $v$ is head-matched.

\subsubsection{Filtering -- Properties of \texorpdfstring{$(\tau^A, \tau^B)$}{(tauA, tauB)} Pairs}
\label{section:tau-pairs}

Recall that layered graphs are defined with respect to $(\tau^A, \tau^B)$ pairs. Furthermore, such a pair determines which edges are kept in and between layers of the corresponding layered graph.

Observe that each layered graph has a property that a path passing through all the layers is an alternating path. So, it is useful to think of paths passing through all the layers as our candidates for weighted augmentations. Naturally, we would like that each candidate for weighted augmentations have a certain property, e.g., that the sum of weights of the unmatched edges is larger than the sum of weights of the matched edges. We control these properties by imposing some restrictions on the $(\tau^A, \tau^B)$ pairs that we consider. Next, we list those restrictions, and their summary is provided in \cref{table:good-tau-pairs}.

Recall that $\tau^A$ corresponds to matched and $\tau^B$ corresponds to unmatched edges. We look for short augmenting paths, so we set the length of $\tau^A$ to be $O(1/\eps^2)$. The exact value is provided in \cref{table:good-tau-pairs}, property~\eqref{item:cycles-length-tauA}.

In our final algorithm, we look for augmenting paths in the graph obtained from $\Layered$ by removing all the edges in the first and the last layer. Furthermore, we require that those paths pass through all the layer. Hence, we require that $|\tau^A| = |\tau^B| + 1$ (property~\eqref{item:cycles-length-tauB}).

As described earlier and as implied by the definition of layered graphs, we bucket the weights of edges in multiples of $\layerepspower W$. To reflect that, we set each entry of $\tau^A$ and $\tau^B$ to be of the form $\layerepspower k$, for $k \in \bbN$ (property~\eqref{item:cycles-form-of-tauA-and-tauB}). Furthermore, we require that each unmatched edge we consider has a sufficiently large weight. To expresses that, we require that each entry of $\tau^B$ is at least $2 \layerepspower$. Similarly, each matched edge which is not an end of a path is required to have non-negligible weight (property~\eqref{item:cycles-lower-bound-of-tauA-and-tauB}).

Recall that our goal is to consider augmentations whose weight is close to $W$ (from the conditions, each augmentation has weight at least $2 \layerepspower W$). Hence, we upper-bound the total sum of the weights of the edges corresponding to $\tau^B$ (property~\eqref{item:cycles-upper-bound-sum-tauB}).

Finally, we want to ensure that each augmentations leads to an increase in the weight. To that end, we require that the set of weights of the edges corresponding to $\tau^B$ is by at least $\layerepspower W$ larger than those corresponding to $\tau^A$ (property~\eqref{item:cycles-diff-tauB-and-tauA}). Observe that from this property and property~\eqref{item:cycles-upper-bound-sum-tauB} it implies that the sum of the weights of the edges corresponding to $\tau^A$ is upper-bounded by $(1 + \eps^4 - \layerepspower) W$.

\begin{table}[ht]\algfontsize
\begin{mdframed}[style=myframe]
A pair $(\tau^A, \tau^B)$ of sequence is called \emph{good} if it has the following properties:

\begin{enumerate}[(A)] 
	\item\label{item:cycles-length-tauA} The sequence $\tau^A$ consists of at most $\tfrac{2}{\eps} \cdot \tfrac{16}{\eps} + 1$ elements;
	\item\label{item:cycles-length-tauB} The sequence $\tau^B$ has one element less than the sequence $\tau^A$;
	\item\label{item:cycles-form-of-tauA-and-tauB} Each entry of $\tau^A$ and each entry of $\tau^B$ is a non-negative multiple of $\layerepspower$;
	\item\label{item:cycles-lower-bound-of-tauA-and-tauB} Each entry of $\tau^B$ and each $\tau^A_i$, whenever $1 < i < |\tau^A|$, is at least $2 \layerepspower$;
	%\item\label{item:cycles-upper-bound-sum-tauA} $\sum_i \tau^A_i \le 1 + \eps^2$;
	\item\label{item:cycles-upper-bound-sum-tauB} $\sum_i \tau^B_i \le 1 + \eps^4$;
	\item\label{item:cycles-diff-tauB-and-tauA} $\sum_i \tau^B_i - \sum_i \tau^A_i \ge \layerepspower$.
\end{enumerate}

\end{mdframed}
\caption{The definition of \emph{good} $(\tau^A, \tau^B)$ pairs.}
\label{table:good-tau-pairs}
\end{table}

\subsubsection{Short Augmentations in Layered Graphs}
\label{sec:existence-of-short-augmentations}
In this section, our goal is to show that each short augmentations of $\cC$ as defined by \cref{lemma:shortpaths} appears among the layered graphs our algorithm constructs.

We begin by showing that any alternating path in a layered graph could be, informally speaking, decomposed into a collection of ``meaningful'' augmentations in $G$. Namely, observe that an alternating path in a layered graph when translated to $G$ might contain cycles. In general, it might not be possible to augment a path intersecting itself. Nevertheless, we show that our layered graph is defined in such a way that every (not necessarily simple) path in $G$ obtained from a layered graph can be decomposed into cycles and paths each of which alone can be augmented.

	\begin{lemma}[Decomposition on a path and even-length cycles]\label{lemma:decomposition-on-paths-and-cycles}
		Let $G^P = \partupple$ be a parametrized graph. Let $P$ be an alternating path in $\Layered(\tau^A, \tau^B, W, G^P)$. Let $S$ be the path obtained from $P$ by replacing each vertex $v^t$ by $v$. (Note that $S$ might not be a simple path.) Then, $S$ can be decomposed into a single simple path and a set of cycles. Furthermore, the edges in the path and the edges in each of the cycles alternate between $A$ and $B$.
	\end{lemma}
	\begin{proof}
		In this proof, we orient the edges of $\Layered$ as follows. Each edge $e = \{u^t, v^{t + 1}\}$ connecting layer $t$ and layer $t + 1$ is oriented from $u^t$ to $v^{t + 1}$. Each edge $e = \{x^s, y^s\}$ within a layer, where $x^s \in L$ and $y^s \in R$, is oriented from $x^s$ to $y^s$. Observe that in this way the head of each matched arc is in $R$ while the tail is in $L$. Also, the head of each unmatched arc is in $L$ while the tail is in $R$. Let $\arc{\Layered}$ denote the resulting oriented graph.
		
		Observe that $P$ corresponds to a directed path $\directedP$ in $\arc{\Layered}$. Let $\directedS$ be the path obtained from $\directedP$ by replacing each vertex $v^t$ by $v$. Hence, disregarding the orientation in $\directedS$ results in $S$.
	
		Let $\arc{C}$ be a cycle obtained by adding an arc between the last and the first vertex of $\directedP$. We will call that arc special. Let $\directedS'$ be obtained from $\arc{C}$ by replacing each vertex $v^t$ by $v$.
		
		First, observe that each node in $\directedS'$ has in-degree equal to its out-degree. Hence, $\directedS'$ is an Eulerian graph. So, $\directedS'$ can be decomposed into arc-disjoint union of cycles. Let $\cC$ be that collection of cycles excluding the cycle containing the special arc. Let $\arc{Q}$ be the path obtained by removing the special arc from the corresponding cycle of the decomposition. Note that by removing the special arc from $\directedS'$ we obtain $\directedS$. Hence, $\cC$ and $\arc{Q}$ represent a decomposition of $\directedS$. Our goal is to show that each cycle of $\cC$ and $\arc{Q}$ are alternating.
		
		Towards a contradiction, assume that there is a vertex $v$ of a cycle of $\cC$ or of $\arc{Q}$ such that the incoming and the outgoing arc both belong to $A$ or both belong to $B$. Then, $v$ should be both in $L$ and in $R$, which is in a contradiction with the parametrization. Hence, the lemma holds.
	\end{proof}

  We now use~\cref{lemma:shortpaths}  to prove that specifically designed layered graphs contain many-by-weight vertex-disjoint augmentations. Specifically, we show that every augmentation considered in that lemma appears in at least one layered graph. 

\begin{restatable}[]{lemma}{lemmapathandcyleslayered}
\label{lemma:paths-and-cycles-in-layered}
	Let $\cC$ be a collection of augmentations as defined by \cref{lemma:shortpaths}. Consider an augmentation $C \in \cC$. Then, there exists a parametrization $G^P$, a choice a good pair $(\tau^A, \tau^B)$, and $W$ so that $\Layered(\tau^A, \tau^B, W, G^P)$ contains a path $S$ passing through all the layers so that when \cref{lemma:decomposition-on-paths-and-cycles} is applied on $S$ it results in a decomposition containing $C$. Furthermore, $W$ equals $(1 + \eps^4)^i \le w(S)$, for some integer $i \ge 0$.
\end{restatable}

% Here we only prove the case where $C$ is a cycle. The proof is very similar for the case where $C$ is a path. For the sake of completeness, we also give the proof for this case in the Appendix.

\begin{proof}

We break the proof of \cref{lemma:paths-and-cycles-in-layered} into two cases: $C$ is a cycle, and $C$ is a path. The proof is similar in both cases, and here we only present the proof for the case where $C$ is a cycle. For the completeness, we present the proof for the case where $C$ is a path in \cref{sec:proofs-missing-from}.

\paragraph{When $C$ is a cycle:}

	We split the proof into three parts. First, we fix a parametrization of the graph, then define a layered graph based on this parametrization. And finally, we show that $C$ appears in the layered graph.
	
	\paragraph{Parametrization:}
	Observe that $C$ has even length, and let $C = v_1 \ldots v_{2 t} v_1$. Note that $t \le 2 / \eps$. Without loss of generality, assume that $\{v_1, v_2\} \in M$. Consider a parametrization $G^P$ of the graph in which
	%each edge $\{u_i, u_{i + 1}\}$, for all $1 \le i < 2 t$, is oriented as $(v_i, v_{i + 1})$, and $\{v_{2 t}, v_1\}$ is oriented as $(v_{2 t}, v_1)$. Furthermore, assume that $G^P$ is such that
	$v_i \in R$ for each even $i$, while $v_i \in L$ for each odd $i$. By the definition, $G^P$ contains $C$.
	
	Let $a_1, \ldots, a_t$ be the matched edges of $C$ appearing in that order, with $a_1 = \{v_1, v_2\}$. Similarly, let $b_1, \ldots, b_t$ be the unmatched edges of $C$ appearing in that order, with $b_1 = \{v_2, v_3\}$.
	
	\paragraph{Layered graph:} Let $d \eqdef 16 / \eps$. We define a layered graph $\Layered(\tau^A, \tau^B, W, G^P)$ so that it contains a (non-simple) alternating path which starts at $a_1$, goes around $C$ for $d$ times, and ends at $a_1$. More formally, $\Layered$ contains an alternating path $S$ passing through all the layers of the form
	\[
		S = \underbrace{a_1 b_1 \ldots a_{t} b_{t}}_{\text{repeated $d$ times}} a_1.
	\]
	Note that $S$ consists of $2 d t + 1$ many edges. Also, when \cref{lemma:decomposition-on-paths-and-cycles} is applied to $S$ it outputs a collection of cycles in which $C$ appears $d$ times.
	
	Define $W$ to be the largest value of the form $(1 + \eps^4)^i \le w(S)$, where $i \ge 0$ is an integer. Note that, $W$ has the form as stated by lemma. Sequences $\tau^A$ and $\tau^B$ are defined as follows:
	\begin{itemize}
		\item Sequence $\tau^A$ has length $d t + 1$ and $\tau^B$ has length $d t$.
		\item For every $a_i$ and for every integer $j$ such that $(j \equiv i \mod t)$, set $\tau^A_j$ to be the smallest $\layerepspower k$ such that $k$ is an integer and $\tau^A_j W \ge w\rb{a_i}$.
		\item For every $b_i$ and for every integer $j$ such that $(j \equiv i \mod t)$, set $\tau^B_j$ to be the largest $\layerepspower k$ such that $k$ is an integer and $\tau^B_j W \le w\rb{b_i}$.
	\end{itemize}
	
	\paragraph{Correctness:}
	We now show that $(\tau^A, \tau^B)$ is a good pair. To that end, show that $\tau^A$ and $\tau^B$ as defined above have all the properties stated by \cref{table:good-tau-pairs}.
	
	Property~\eqref{item:cycles-length-tauA}-\eqref{item:cycles-form-of-tauA-and-tauB} are ensured by the construction. It is easy to verify that $\Layered(\tau^A, \tau^B, W, G^P)$ contains $S$.
	
	We next show that property~\eqref{item:cycles-lower-bound-of-tauA-and-tauB} holds as well. First, recall that from \cref{lemma:shortpaths}, for every $e \in C$ it holds $w(e) \ge (\eps^6 / 64) w(C)$ (for the elements $C \setminus M$ we have even stronger guarantee). Then, we have
	\[
		w(e) \ge \frac{\eps^6}{64} w(C) \ge \frac{\eps^6}{64} \frac{w(S)}{d + 1} \stackrel{d = 16/\eps; \eps \le 1/16}{\ge} 2 \layerepspower w(S) \ge 2 \layerepspower W,
	\]
	and the property~\eqref{item:cycles-lower-bound-of-tauA-and-tauB} follows by the definition of $\tau^A$ and $\tau^B$.
	
	%Property~\eqref{item:cycles-upper-bound-sum-tauA} follows from the following sequence of inequalities
	% \begin{eqnarray*}
	% 	\sum_i \tau^A_i W & \le & w(S) + \layerepspower (t d + 1) W \\
	% 	& \le & (1 + \eps^4 + 2 \layerepspower t d) W \\
	% 	& \stackrel{t \le 2 / \eps; d = 16 / \eps}{\le} & (1 + \eps^4 + 64 \eps^{10}) W \\
	% 	& \stackrel{\eps \le 1/16}{\le} & (1 + \eps^2) W.
	% \end{eqnarray*}
	
	To show Property~\eqref{item:cycles-upper-bound-sum-tauB}, we observe that
$
		\sum_i \tau^B_i W \le w(S) \le (1 + \eps^4) W,
$
	implying $\sum_i \tau^B_i \le (1 + \eps^4)$.
	
	The entries of $\tau_A$ and $\tau_B$ represent discretized edge-weights of $C \cap M$ and $C \cap \optWM$, respectively. Observe that $\tau^B W$ lower-bounds the edge-weights of $C \cap \optWM$, while $\tau^A W$ upper-bounds the edge-weights of $C \cap M$. We will show that even when the weights are discretized, the difference between the weighted and unweighted edges of $S$ is significant. To that end, we compare $\sum_i \tau^A_i W$ and $\sum_i \tau^B_i W$. First, we have
	\begin{equation}\label{eq:bound-on-sum-tauB}
		\sum_i \tau^B_i W \ge d (w(C \cap \optWM) - t \layerepspower W).
	\end{equation}
	We also have
	\begin{eqnarray}
		W & \le & w(S) \le (d + 1) w(C) \nonumber \\
		& = & (d + 1) \rb{w(C \cap M) + w(C \cap \optWM)} \nonumber \\
			& \le & 2 (d + 1) w(C \cap \optWM) \nonumber \\
			& \le & 4 d w(C \cap \optWM). \label{eq:bound-on-TW}
	\end{eqnarray}
	Combining~\eqref{eq:bound-on-sum-tauB} and~\eqref{eq:bound-on-TW} leads to
	\begin{eqnarray}
		 \sum_i \tau^B_i W & \ge & d w(C \cap \optWM) (1 - 4 \layerepspower t d) \nonumber \\
			 & \stackrel{t \le 2/ \eps; d = 16 / \eps}{\ge} & d w(C \cap \optWM) (1 - 16 \cdot 8 \eps^{10}) \nonumber \\
			 & \stackrel{\eps \le 1/16}{\ge} & d w(C \cap \optWM) \rb{1 - 8 \eps^9}. \label{eq:final-bound-on-sum-tauB}
	\end{eqnarray}
	Next, observe that from~\eqref{eq:bound-on-TW} and $\eps \le 1 / 16$ we have
	\begin{eqnarray}
		\sum_i \tau^A_i W & \le & (d + 1) (w(C \cap M) + \layerepspower t W) \nonumber \\
		& \stackrel{\text{from \cref{lemma:shortpaths}}}{\le} & (d + 1) \rb{\frac{w(C \cap \optWM)}{1 + \eps / 8}  + \layerepspower t W} \nonumber \\
		& \le & \frac{(d + 1) (1 + 8 \layerepspower t d)}{1 + \eps / 8} w(C \cap \optWM) \nonumber \\
		& \stackrel{t \le 2 / \eps; d = 16 / \eps}{\le} & \frac{(d + 1) (1 + \eps^8)}{1 + \eps / 8} w(C \cap \optWM). \label{eq:final-bound-on-sum-tauA}
	\end{eqnarray}
	From~\eqref{eq:final-bound-on-sum-tauB}, \eqref{eq:final-bound-on-sum-tauA} and the definition of $d$ we derive
	\begin{eqnarray*}
		\sum_i \tau^B_i W - \sum_i \tau^A_i W
		& \ge & \rb{(16 / \eps) \rb{1 - 8 \eps^9} - \frac{(1 + 16 / \eps) (1 + \eps^8)}{1 + \eps / 8}} w(C \cap \optWM) \\
		& = & \frac{(2 + 16 / \eps) \rb{1 - 8 \eps^9} - (1 + 16 / \eps) (1 + \eps^8)}{1 + \eps / 8} w(C \cap \optWM) \\
		& = & \frac{1 - 16 \eps^7 - 129 \eps^8 - 16 \eps^9}{1 + \eps / 8} w(C \cap \optWM)\,,
	\end{eqnarray*}
which is $\ge \layerepspower W$ because $\eps \le 1/16$.
	The last chain of inequalities implies
	\begin{equation}\label{eq:difference-tauA-tauB}
		\sum_i \tau^B_i - \sum_i \tau^A_i \ge \layerepspower,
	\end{equation}
	hence showing that Property \eqref{item:cycles-diff-tauB-and-tauA} holds as well.
	
	\paragraph{When $C$ is a path:}

We defer the proof  to \cref{sec:proofs-missing-from} as it is very similar to the previous case. 
\end{proof}

%%% Local Variables:
%%% mode: latex
%%% TeX-master: "000-main_random_order"
%%% End:

\subsection{Combining the Results}
\label{sec:combine}
We are now ready to prove~\cref{lemma:weight-class}, and we start with the algorithm (\cref{alg:fixed-weight}) that is used to prove this theorem.

\begin{algorithm}
	\DontPrintSemicolon
	\SetKwInOut{Global}{Global}
  \SetKwInOut{Input}{Input}
  \SetKwInOut{Output}{Output}

  \Input{A weighted graph $G$ \\
		Approximation parameter $\eps$\\
		Weight $W$}
  \Output{Augmentations corresponding to $W$}

  \BlankLine
	
	Partition the vertex set into $L$ and $R$ by assigning each vertex to one of the sets uniformly at random and independently. Let $G^P$ be the resulting parametrized graph. \label{line:fix-parametrization}
  
  Let $\cT$ be the set of all good $(\tau^A, \tau^B)$ pairs, where good pairs are defined in \cref{table:good-tau-pairs}.

	\For{each $(\tau^A, \tau^B) \in \cT$ in parallel} {
		Define $\Layered'$ to be $\Layered(\tau^A, \tau^B, W, G^P)$ with the edges from the first and the last layer removed.

    Let $M' = \UnwBipMatching(\Layered', \delta)$ be the matching returned by a $(1-\delta)$-approximation bipartite unweighted matching algorithm (recall that $\Layered'$ is bipartite).
		
		Let $M_{\Layered'}$ be the matching $M$ restricted to $\Layered'$.
		
		Let $\cP$ be the collection of augmentations in $M'\cup M_{\Layered'}$. \label{line:augmentations-cP}
		
		Let $\cApair$ be a set of augmentation in $G$. Initially, $\cApair$ is the empty set.\label{line:init-cApair}
		
		\For{each $P \in \cP$\label{line:loop-define-cApair}} {
			Apply \cref{lemma:decomposition-on-paths-and-cycles} to $P$, i.e., decompose $P$ into a union of even-length cycles and a simple path in $G$. Let $\cC$ be that decomposition. \label{line:apply-lemma-decomposition}
			
			Choose an augmentation $C \in \cC$ that has the largest gain among the elements of $\cC$.
			
			If $C$ does not intersect any element of $\cApair$, add $C$ to $\cApair$. \label{line:add-to-C-if-possible}
		}
	}

	Let $\cA_W$ be a $\cApair$ set that maximizes gain over all $(\tau^A, \tau^B)$ pairs. \label{line:take-max-cA}
	
	\Return $\cA_W$
  
  \caption{Algorithm used by \cref{lemma:weight-class}}
  \label{alg:fixed-weight}
\end{algorithm}

%As the next step, we connect the guarantee provided by $\UnwMatching$ and the properties of good pairs to lower-bound the gain that a layered graph provides. 
\begin{lemma}\label{lemma:guarantee-for-one-layered}
	Let $\Layered$ be a layered graph constructed by \cref{alg:fixed-weight}. Use $\gain(\Layered)$ to denote the maximum gain obtained by applying some vertex-disjoint augmenting paths of $\Layered$ where each of the paths passes through all the layers of $\Layered$. Define $\Layered'$ as the graph obtained by removing the edges in the first and the last layer of $\Layered$. Let $w(M_{\Layered'})$ be the total weight of the matching edges in $\Layered'$. Let $\cApair$ be the set of augmentations as obtained at \crefrange{line:init-cApair}{line:add-to-C-if-possible}. Then
%Let $\cP$ be the augmenting paths obtained by running $\UnwMatching$ on $\Layered'$ at \cref{line:augmentations-cP}. $\cP$ applied in $G$ as described within the loop at \cref{line:loop-define-cApair} leads to a gain of at least
	\[
		\gain(\cApair) \ge \eps^{20} \rb{\frac{(1-\delta) \gain(\Layered)}{2} - \frac{\delta w(M_{\Layered'})}{\layerepspower}}\,.
	\]
\end{lemma}
\begin{proof}
	Let $M_{\Layered'}$ denote the matching edges in $\Layered'$. Then, from the definition of $\tau^A$ and $\Layered$, we have
	\begin{equation}\label{eq:bound-on-X-edges}
		|M_{\Layered'}| \le \frac{w(M_{\Layered'})}{\layerepspower W}.
	\end{equation}
	
	Let $\cC$ be a collection of augmenting paths in $\Layered$ that have gain $\gain(\Layered)$ and let $\cP$ be a collection of vertex-disjoint paths found at \cref{line:augmentations-cP}. Observe that, as the first and the last layer of $\Layered$ consist of matched edges, each augmenting path passing through all the layers in $\Layered$ corresponds to an augmenting path in $\Layered'$ passing through all the layers as well, and vice-versa.
  Also, any augmenting path in $M'\cup M_{\Layered'}$ must pass through all the layers because there cannot be a free vertex with respect to $M_{\Layered'}$ except in the first and the last layer (see the filtering step in \cref{definition:layered-graph}).  So we have that $|M'| = |M_{\Layered'}| + |\cP|$.
  Since $\UnwBipMatching$ returns a $(1-\delta)$-approximate matching,
  \[
    |M_{\Layered'}| + |\cC| \le \frac{|M'|}{1-\delta} = \frac{|M_{\Layered'}| + |\cP|}{1-\delta}\,,
    \]
which, after simplification, gives
	\begin{equation}\label{eq:bound-on-len-cP}
		|\cP| \ge ((1-\delta) |\cC| - \delta |M_{\Layered'}|).
	\end{equation}

Next, from the definition of $\tau^A$ and $\tau^B$, each augmenting path in $\Layered$ increases the weight of the matching by at most $\sum_i (\tau_i^B + \layerepspower) W \le 2 W$. So, we have $|\cC| \ge \gain(\Layered) / (2 W)$, that together with \eqref{eq:bound-on-X-edges} and \eqref{eq:bound-on-len-cP} implies % and the guarantee of $\UnwMatching$ stated in \cref{table:unw-matching} 
	\begin{equation}\label{eq:lower-bound-on-cA}
		|\cP| \ge \frac{(1-\delta) \gain(\Layered)}{2 W} - \frac{\delta w(M_{\Layered'})}{\layerepspower W}% = \frac{\layerepspower \beps W - 2 \delta W_M}{2 \layerepspower \tW}.
	\end{equation}
	
	In the rest of the proof, we use the lower-bound on $|\cP|$ to lower-bound $\gain(\cApair)$.
	
	First, consider a path $P \in \cP$. When $P$ is translated to $G$ (\cref{line:apply-lemma-decomposition}), it might intersect itself and not being augmenting. From \cref{lemma:decomposition-on-paths-and-cycles}, $P$ can be decomposed into a collection of augmenting cycles and an augmenting path in $G$. Let $\cD_P$ be the collection of components in this decomposition. From the definition of $\tau^A$ and $\tau^B$ we have that $P$ has gain at least $\layerepspower W$. Also, each $\tau^A_i$ and each $\tau^B_j$ is multiple of $\layerepspower$. This further implies that there is at least one component in $\cD_P$ that has gain at least $\layerepspower W$. This implies that for every path of $\cP$ there is an augmentation in $G$ that has gain at least $\layerepspower W$.
	
	However, notice that although the paths in $\cP$ are vertex-disjoint, when they are translated to $G$ they might share some vertices. This comes from the fact that in $\Layered$ the vertices of $G$ are copied in every layer. Now we want to account for these overlaps. First, each vertex of $G$ is copied $|\tau^A| + |\tau^B| + 1$ many times in $\Layered$. Hence, $|\tau^A| + |\tau^B| + 1$ many paths of $\Layered$ can intersect at the same vertex of $G$. Furthermore, each path in $\Layered$ consists of $|\tau^A| + |\tau^B| + 1$ vertices. Therefore, each component in $G$ obtained from a path of $\cP$ overlaps with at most $(|\tau^A| + |\tau^B| + 1)^2 \le 1 / \eps^8$ many other such components. This together with~\eqref{eq:lower-bound-on-cA} implies that for $\cApair$ as defined at \crefrange{line:init-cApair}{line:add-to-C-if-possible} we have
	\begin{eqnarray*}
		\gain(\cApair) & \ge & \layerepspower \cdot \eps^8 W  |\cP| \\
		& \ge & \layerepspower \cdot \eps^8 W \rb{\frac{(1-\delta) \gain(\Layered)}{2 W} - \frac{\delta w(M_{\Layered'})}{\layerepspower W}} \\
		& \ge & \eps^{20} \rb{\frac{(1-\delta) \gain(\Layered)}{2} - \frac{\delta w(M_{\Layered'})}{\layerepspower}},
	\end{eqnarray*}
	as desired.
\end{proof}
\begin{proof}[Proof of \cref{lemma:weight-class}]
% In our analysis, we apply \cref{theorem:McGregor}, that is parametrized by the length of short augmentations that our algorithm is searching for. For our purpose, it suffices to search for augmentations of length up to $64 / \eps^2 - 1$. Hence, we define
% 	\begin{equation}\label{eq:definition-of-beps}
% 		\beps \eqdef \eps^{-128 / \eps^2},
% 	\end{equation}
% 	that, for $2i + 1 < 64 / \eps^2$, lower-bounds $\beta$ defined in \cref{theorem:McGregor}.

	Let $\cC$ be the family of augmentations as defined by \cref{lemma:shortpaths}.
	%By applying \cref{lemma:guarantee-for-one-layered}, we lower-bound the expected gain of $\cC$ preserved by a fixed parametrization and for a given weight.
	From \cref{lemma:paths-and-cycles-in-layered}, for every $C \in \cC$ there exists a parametrization of $G$, weight $W$, and a layered graph defined with respect to $W$ and considered by \cref{alg:fixed-weight} in which $C$ appears and passes through all its layers.\footnote{In our analysis, given a layered graph we only consider paths that pass through all the layers, i.e., only those paths that have at least one vertex in each of the layers. For the sake of brevity, we will omit specifying that a path passes through all the layers and, instead, only say that a path appears in a layered graph.} Let $\cC_W \subseteq \cC$ be the subcollection of $\cC$ appearing in layered graphs defined with respect to $W$. \cref{alg:fixed-weight} fixes a parametrization of $G$ and then constructs layered graphs with respect to that parametrization. $C$ appears in a layered graph if its vertices are properly assigned to $L$ and $R$. Recall that each vertex gets assigned to one of the two sets with probability $1/2$ and independently of other vertices. Hence, the probability that $C \in \cC$ remains in a random parametrization is at least $2^{-|C|} \ge 2^{-65 / \eps^2}$. This, implies that the expected gain obtained by applying all the augmentations of $\cC_W$ that remain in one parametrization is at least $2^{-65 / \eps^2} \gain(\cC_W)$. Our goal now is to show that \cref{alg:fixed-weight} finds augmentations whose gain is ``close'' to this remained gain.

	\cref{alg:fixed-weight} finds augmentations in all the layered graphs independently (\cref{line:augmentations-cP}) and, hence, those augmentations might overlap. Furthermore, even a single augmentation from a layered graph when translated to $G$ might intersect itself. In both of these cases, our aim is to resolve overlap-conflicts while retaining large gain.
	
	Note that the number of layered graphs for a constant $\eps$ is $O(1)$. Hence, to show that \cref{alg:fixed-weight} retains large gain, it suffices to show that for a fixed $(\tau^A, \tau^B)$ pair the following is achieved:
	\begin{itemize}
		\item \cref{alg:fixed-weight} finds many-by-weight augmentations of the corresponding layered graph.
		\item \cref{alg:fixed-weight} translates those augmentations to $G$ so to retain most of their gain (\crefrange{line:apply-lemma-decomposition}{line:add-to-C-if-possible}).
	\end{itemize}
	Notice that these properties are essentially guaranteed by \cref{lemma:guarantee-for-one-layered}. So, it remains to count the number of layered graphs and apply \cref{lemma:guarantee-for-one-layered} to conclude the proof. To that end, for a fixed $W$, let $\Layered$ be a layered graph that maximizes the gain. Let $\gain(\Layered)$ be the maximum gain that can be obtained by applying vertex-disjoint augmenting paths of $\Layered$. We next lower-bound $\gain(\Layered)$.
	
	Observe that there are at most $(2 / \layerepspower + 2)^{65 / \eps^2}$ distinct $(\tau^A, \tau^B)$ pairs. (In this bound, the term ``$+2$'' comes from the fact that $\tau^A_i$ can be zero, and from the fact that a layer might not exist in which case we think that it has value $-1$.) Hence, in expectation over all parametrization, we have
	\begin{equation}\label{eq:bound-gain-layered}
		\gain(\Layered) \ge 2^{-65 / \eps^2} (2 / \layerepspower + 2)^{-65 / \eps^2} \gain(\cC_W) \ge \eps^{900 / \eps^2} \gain(\cC_W).
	\end{equation}
	
	\paragraph{Proving Properties~\ref{item:cAW-is-in-aug-class} and \ref{item:weights-partial-gain}:} As in the statement of \cref{lemma:guarantee-for-one-layered}, $\Layered'$ is obtained by removing the edges from the first and the last layer of $\Layered$, and $\cApair$ is obtained at \crefrange{line:init-cApair}{line:add-to-C-if-possible}. We will show that $\cApair$ satisfies the required properties. From it will follow that $\cA_W$ returned at \cref{line:take-max-cA} satisfies those properties as well. Property~\ref{item:cAW-is-in-aug-class} follows by the definition of layered graphs and our discussion above. So it remains to show that Property~\ref{item:weights-partial-gain} holds as well.
	
	As a reminder, $\cApair$ corresponds to $\Layered$ that maximizes the gain among all the layered graphs for $W$. From \cref{lemma:guarantee-for-one-layered} and~\eqref{eq:bound-gain-layered} we have that in expectation
	% \begin{equation}\label{eq:gain-final-derivation}
	% 	\gain(\cApair) \ge \eps^{20} \rb{\frac{(1-\delta) \eps^{900 / \eps^2} \gain(\cC_W)}{2} - \frac{\delta w(M_{\Layered'})}{\layerepspower}}.
	% \end{equation}
	% Plugging-in the definition of $\beps$ from \eqref{eq:definition-of-beps} into~\eqref{eq:gain-final-derivation}, we derive
	\begin{equation}\label{eq:gain-final-derivation-beps}
		\gain(\cApair) \ge (1-\delta)\eps^{21+900 / \eps^2} \gain(\cC_W) - \eps^{8} \delta w(M_{\Layered'}).
	\end{equation}
	Let $w(M_{W})$ be the weight of the matched edges of $G$ such that each edge has weight between $\layerepspower W$ and $2 W$. Notice that a matched edge of $G$ appears at most $32 / \eps^2 + 1 \le 1 / \eps^4$ many times in $\Layered$. Recall that each matching edge in $\Layered'$ has weight at least $\layerepspower W$. Hence,
	\begin{equation}\label{eq:bound-on-wMLayered}
		w(M_{\Layered'}) \le w(M_{W}) / \eps^4.
	\end{equation}
	Letting
	\[
		\delta \eqdef \eps^{28 + 900 / \eps^2}
	\]
	from~\eqref{eq:gain-final-derivation-beps} and~\eqref{eq:bound-on-wMLayered} we obtain
	\[
		\gain(\cA_W) \ge \gain(\cApair) \ge \eps^{22+ 900/ \eps^2} \gain(\cC_W) - \eps^{32 + 900 / \eps^2} w(M_{W}).
	\]
	This proves that Property~\ref{item:weights-partial-gain} holds as
        well.
        
\paragraph{MPC implementation:}
\Cref{alg:fixed-weight} can be implemented in $U_M$ MPC rounds in the following way.	
	\cref{line:fix-parametrization} is implemented by collecting all the vertices to one machine, call that machine $\mu$, and randomly assigning them to $L$ and $R$ (in the way as described in \cref{section:parametrization}). Then, the edge-set of $G$ is distributed across the machines, while the vertex sets $L$ and $R$ are sent to each of those machines. Notice that $\mu$ cannot send directly $L$ and $R$ to each of the machines, as it would result in outgoing communication of $\mu$ being at least $n \machines$ bits (recall that $\machines$ denotes the number of machines) which could be much larger than the memory of $\mu$ (see \cref{sec:prelim} for details on the bound on the communication in each round). So, distributing $L$ and $R$ to each of the machines is performed in two steps as follows. First, $\mu$ locally splits $L \cup R$ into $\machines$ sets, so that each set has $\lceil n / \machines \rceil$ or $\lfloor n / \machines \rfloor$ vertices. Notice that in our case, $n \ge \machines$ and hence each of the sets is non-empty. Then, these sets are sent to the $\machines$ machines -- one set per machine. In the second step, each machine sends its set to each of the other machines. Since we assumed that the memory per machine is at least $n$, the total incoming and outgoing communication of a machine in this step does not exceed its memory.  In the similar way, we can make sure that each machine knows the current matching $M$.
	
%	After the vertices are distributed to each of the machines, based on $L$ and $R$ for each edge is decided whether it remains in the parametrization or not.Those edges that do not remain in the parametrization are ignored in the corresponding execution of \cref{alg:mpc}.

Then all $(\tau^A, \tau^B)$ pairs are generated by each machine.  For constant
$\eps$, there are at most $O(1)$ many such pairs.  For each $(\tau^A, \tau^B)$,
each machine can then generate its part of $\Layered'$ as follows.  Each vertex is replicated many times, where copy $v^{W, (\tau^A, \tau^B), t}$ corresponds to the parameters: weight $W$, a good pair $(\tau^A, \tau^B)$, and the layer $t$ it belongs to. Let $e = (u, v)$ be a parametrized edge of $G^P$. The edge $e$ is replicated \emph{locally} to each layer for which it satisfies the weight requirements.  If $e=\{u^i,v^{i+1}$ is not a matching edge, then we need to check if one of $u^i$ and $v^{i+1}$ is removed in the filtering step (see the description of layered graphs in~\cref{section:parametrization}).  These checks are straightforward because each machine knows $M$.

% In this way, we create each layered graph while only having the access to all
% good pairs and a parametrization of $G$. Each edge is replicated
% $O_{\eps}(\istar)$ times, where $\istar$ is defined as in
% \Cref{line:define-istar} of \cref{alg:mpc}. This construction of layered
% graphs can be implemented in $O(1)$ MPC rounds. This together with the round and
% memory complexity of $\UnwMatching$ concludes the analysis of the MPC
% implementation of \cref{alg:fixed-weight}.

After that $\UnwBipMatching$ is called for each $(\tau^A, \tau^B)$, which uses
$U_M$ MPC rounds and $O_\eps(n)$ memory per machine, because $\delta$ is a function of only $\eps$.
Irrespective of how $\UnwBipMatching$ stores its output, $\cP$ can be collected on
a fixed machine, which then does the remaining processing, and redistributes the
output $\cA_W$.

\paragraph{Streaming implementation:}
\Cref{alg:fixed-weight} can be implemented in $U_S$ passes as follows.
Random assignment to $L$ and $R$ can be done initially and stored.  Then $O_\eps(1)$ pairs $(\tau^A, \tau^B)$ are generated, and for each pair, $\UnwBipMatching$ is then called, which uses $U_S$ passes and $O_\eps(n\poly(\log(n)))$ memory.  When an edge $e$ arrives in the stream, it is fed to those instances of $\UnwBipMatching$ for which it appears in some layer.  This happens if the edge $e$ and neighboring matching edges $e_1$ and $e_2$ satisfy weight and orientation (with respect to $L$ and $R$) requirements (see~\cref{section:parametrization}).
Outputs of all the instances are then collected together after which the further processing is straightforward.
\end{proof}
%This finishes the proof of \cref{lemma:weight-class}.

%%% Local Variables:
%%% mode: latex
%%% TeX-master: "000-main_random_order"
%%% End:

%%% Local Variables:
%%% mode: latex
%%% TeX-master: "000-main_random_order"
%%% End:

\section*{Acknowledgments}
B.~Gamlath, S.~Kale and O.~Svensson were supported by ERC Starting grant 335288-OptApprox.
S.~Mitrovi{\' c} was supported in part by the Swiss NSF grant P2ELP2\_181772. Part of this work was carried out while S.~Mitrovi{\' c} was visiting ETH.

\bibliographystyle{alpha}
\bibliography{ref}

\appendix

\section{Proofs Missing from Earlier Sections}
\label{sec:proofs-missing-from}
\subsection{Proof of~\Cref{lem:unw-blackbox}}
\lemunwblackbox*
\begin{proof}
  The algorithm maintains a support set $S$ greedily.  We use a parameter
  $\lambda$ that depends on $\beta$.  Whenever we see an edge $uv$ such that $u$
  is an unmatched vertex and $v$ is a matched vertex, we add it to $S$ if degree
  of $u$ in $S$ is less than $\lambda$ and degree of $v$ in $S$ is less than
  $2$.  In the end, we greedily find vertex disjoint $3$-augmentations and
  return them.

  Let $E_3\subseteq M$ be the set of $3$-augmentable edges, so $|E_3| \ge \beta
  |M|$.  We call an edge $vw$ in $E_3$ a \emph{bad} edge, if one of the following
  happens in $S$:
  \begin{itemize}
    \item There is no edge incident to $v$ or $w$.
    \item There is exactly one edge incident to each of $v$ and $w$, but it is
    to the same vertex (which, gives us a triangle, not an augmentation).
  \end{itemize}
  We can individually augment all edges in $E_3\setminus E_B$, which we call
  \emph{good} edges, and we denote this set of good edges by $E_G$.  The crucial
  observation is that a for a bad edge $vw$, one of the edges on its
  $3$-augmenting path $avwb$ was not added to $S$ by the algorithm.  Which means
  that one of $a$ and $b$ already had $\lambda$ edges incident to it.  Hence,
  $\lambda |E_B| \le |S| \le 4|M|$, because each edge in $M$ can have at most
  $4$ support edges incident to it.  This gives
  $\lambda(|E_3| - |E_G|) \le 4|M|$.  Using $|E_3| \ge \beta |M|$ and algebraic
  simplification, we get that $|E_G| \ge (\beta - 4/\lambda)|M|$.  When we
  greedily augment using $S$, for each augmentation $avwb$, we may potentially
  lose up to $2\lambda$ augmentations, because we cannot use the support edges
  incident to $a$ or $b$ any more, otherwise we lose the vertex-disjointness
  property of the $3$-augmenting paths that we return.  Therefore, the number of
  $3$-augmentations that we return is at least
  \[
    \frac{|E_G|}{2\lambda} \ge \left( \frac{\beta}{2\lambda} -
      \frac{2}{\lambda^2}\right)|M|\,,
  \]
  which finishes the proof if we use $\lambda = 8/\beta$.
\end{proof}

\subsection{Proof of~\Cref{lem:kmmlem1}}
%\begingroup
%\def\thelemma{\ref{lem:kmmlem1}}
%\begin{lemma}
%  Let $\advntg \ge 0$, $M'$ be a maximal matching in $G$, and $M^*$ be a maximum
%  matching in $G$ such that $|M'| \le (\nicefrac{1}{2} + \advntg) |M^*|$.  Then
%  the number of $3$-augmentable edges in $M'$ is at least
%  $(\nicefrac{1}{2} - 3\advntg)|M^*|$, and the number of non-$3$-augmentable
%  edges in $M'$ is at most $4\advntg|M^*|$.
%\end{lemma}
%\addtocounter{lemma}{-1}
%\endgroup
\lemmakmm*
\begin{proof}
  Let the number of $3$-augmentable edges in $M'$ be $k$. For each
  $3$-augmentable edge in $M'$, there are two edges in $M^*$ incident on it.
  Also, each non-$3$-augmentable edge in $M'$ lies in a connected component of
  $M' \cup M^*$ in which the ratio of the number of $M^*$-edges to the number
  of $M'$-edges is at most $\nicefrac{3}{2}$.  Hence, \allowdisplaybreaks
  \begin{align*}
    |M^*| &\le 2k + \frac{3}{2} (|M'| - k)
    &&\text{ since there are $|M'|-k$ non-$3$-augmentable edges}\,,\\
          &\le 2k + \frac{3}{2} \left(\left(\frac{1}{2} + \advntg \right)|M^*| -
            k\right)
    &&\text{ because $|M'| \le (\nicefrac{1}{2} + \advntg) |M^*|$}\,,\\
          &= \frac{1}{2}k + \left(\frac{3}{4} + \frac{3}{2}\advntg
            \right)|M^*|\,,
  \end{align*}
  which, after simplification, gives $k \ge (\nicefrac{1}{2} - 3\advntg) |M^*|$.  And the
  number of non-$3$-augmentable edges in $M'$ is
  $|M'| - k \le |M'| - (\nicefrac{1}{2} - 3\advntg)|M^*| \le (\nicefrac{1}{2} +\advntg - \nicefrac{1}{2} +
  3\advntg)|M^*| = 4\advntg|M^*|$.
\end{proof}

%%% Local Variables:
%%% mode: latex
%%% TeX-master: "000-main_random_order"
%%% End:

\subsection{Proof of \cref{lemma:shortpaths}}
\label{sec:proof-of-short-paths}
\lemmashortpaths*
%\begin{proof}[Proof of \cref{lemma:shortpaths}]
\begin{proof}
	We first provide a proof in which Property~\ref{item:property-large-matched-edges} is ignored, and Property~\ref{item:property-path} is replaced by a more strict property
	\begin{equation}\label{eq:modified-property-path}
		\text{For every $C$ in the collection, we have that $w(C\cap \optWM) \ge (1+\eps/4)\cdot w(\Nmatching{C})$.}
	\end{equation}

	We construct $\cC'$ having this modified set of properties. After that, we show how to obtain $\cC$ from $\cC'$.
	
	\paragraph{Constructing $\cC'$:} 
	{\bf (Property~\ref{item:property-large-matched-edges} ignored. Property~\ref{item:property-path} replaced by Property~\eqref{eq:modified-property-path})}
  Without loss of generality, assume that $M \cup \optWM$ is a union of cycles\footnote{
  By adding zero-weight edges, one can assume that $M$ and $\optWM$ are two perfect matchings.  
  Then the edges in $M \cap \optWM$ can be considered as pair of different edges that 
  form a $2$-cycles.}, say $C_1, C_2, \ldots$.  
  Label edges of $\optWM$ using the set $[|\optWM|]$, i.e.,
  $\{1, 2, \ldots, |\optWM|\}$, in such a way that $\optWM$-edges in a cycle $C_i$ get
  labels that ``respect'' the cycle order.  To elaborate, first number the
  $\optWM$-edges in $C_1$ starting at an arbitrary edge, in the cyclical order, as
  $1, 2, \ldots, |C_1|/2$.  Then continue on to $C_2$, and start with
  $|C_1|/2+1$, and so on.

  Now, for $i \in [4/\eps]$, let $\optWM_{-i}$ be the matching obtained by removing
  edges $\optWM_i = \{i, i+4/\eps, i+8/\eps,\ldots \}$ from $\optWM$.  
  For any $i$, the set of edges $M\cup \optWM_{-i}$ is a union of vertex disjoint paths
  or cycles, each of which has length at most $4/\eps$, and each path starts and ends 
  in an $M$-edge.  
  If we pick $i$ uniformly at random from $\{1, \dots, 4/\eps \}$, then $\EE[\optWM_{-i}] 
  = (1-\eps/4)w(\optWM)$.  
  Thus, by the probabilistic method, there is some $i$ for which 
  $w(\optWM_{-i}) \geq (1 - \eps/4) w(\optWM)$.
  We show the existence of the desired set $\cC$ using $\optWM_{-i}$.

  Let $\tcC := \{H_1, H_2, \ldots, H_k\}$ be the collection of paths and cycles in
  $M\cup \optWM_{-i}$. 
  Construct $\cC_A$ as follows: Start $\cC_A$ being empty.
  For each $H \in \tcC$, split $H$ into pieces by removing each edge $e \in H \cap \optWM$ 
  such  that $w(e) < (\eps^2/64) w(H)$, and add the pieces to $\cC_A$. 
	Notice that if $C'$ is path obtained from a path or cycle $C \in \tcC$ after removing
	some $\optWM$-edges, $C'$ must start and end in $M$-edges, and the removal of such
	edges can only decrease the path length.
	Furthermore, if $e \in C' \cap \optWM$ is a remaining edge, 
	then $w(e) \geq (\eps^2/64) w(C) \geq (\eps^2/64) w(C')$. 
	Thus $\cC_A$ satisfies Property~\ref{item:property-length}
	and Property~\ref{item:property-large-edges}.
		
	First, note that from the way we constructed path $C \in \cC_A$ it holds $\Nmatching{C} = C \cap M$. Now, let $\cCbad \eqdef \{C \in \cC_A : w(C\cap \optWM) < (1+\eps/4) w(C\cap M)$ and
	let $\cC' \eqdef \cC_A \setminus \cCbad$. Hence, $\cC'$ satisfies Property~\eqref{eq:modified-property-path}. Also, since $\cC'$ is a sub-collection of $\cC_A$, $\cC'$ satisfies  Property~\ref{item:property-length} and Property~\ref{item:property-large-edges}.

	\paragraph{Proving Property~\ref{item:property-collection} for $\cC'$:}
	What remains is to show that Property~\ref{item:property-collection} holds for 
	$\cC'$ as well. 

	For an element $C \in \cC'$, we first show that the following holds
	\[
		w(C\cap \optWM) - w(C\cap M) \ge \eps w(C)/16.
	\]
	For $C$ satifying Property~\eqref{eq:modified-property-path}, we have
	\begin{eqnarray}
		w(C \cap \optWM) - w(C \cap M) & \geq & w(C \cap \optWM) - \frac{w(C \cap \optWM)}{1 + \eps/4} \nonumber \\
		& = & \frac{\eps/4}{1 + \eps/4} w(C \cap \optWM) \nonumber \\
		& \geq & (\eps/8) w(C \cap \optWM).\label{eq:helper-for-eq2}
	\end{eqnarray}
	This further implies
	\begin{align}\label{eq:2}
		w(C\cap \optWM) - w(C\cap M) & \stackrel{\text{by~\eqref{eq:helper-for-eq2}}}{\ge} \eps w(C\cap \optWM)/8 \stackrel{\text{by Property~\eqref{eq:modified-property-path}}}{\ge} \eps w(C)/16.
	\end{align}	

	Next, we upper-bound the total weight of the edges that were removed when
	constructing $\cC_A$ from $\tcC$. For any path or cycle $C \in \tcC$, the total
	removed weight from $C$ is at most $ (\eps^2/64) (4/\eps) w(C) \leq (\eps/16) w(C)$ 
	(recall that $|C| \leq 4/\eps$). Let $R$ be the set of all such removed edges.
	Then $$w(R) \leq \sum_{C \in \tcC} (\eps/16) w(C) \leq (\eps/16) w(M \cup \optWM_{-i}) 
	\leq (\eps/8) w(\optWM).$$
	
	Let $w(\cX \cap \optWM)$ denote $\sum_{C \in \cX} w(C \cap \optWM)$; therefore $w(\tcC\cap \optWM) \ge (1-\eps/4)w(\optWM)$. Notice that this implies
	\begin{eqnarray}
		w(\cC_A \cap \optWM) & = & w(\tcC \cap \optWM) - w(R) \nonumber \\
		& \ge & \rb{1-\eps/4 - \eps/8}w(\optWM) \nonumber \\
		& \geq & \rb{1-3\eps/8}w(\optWM). \label{eq:bound-on-cCA}
	\end{eqnarray}
	Now, we claim that
	\begin{equation}\label{eq:hypothesis}
		w(\cC' \cap \optWM) \ge \eps w(\optWM)/4.
	\end{equation}
	Towards a contradiction, assume that $w(\cC' \cap \optWM) < \eps w(\optWM)/4$. This implies
	\begin{eqnarray*}
		w(\cCbad \cap \optWM) & = & w(\cC_A \cap \optWM) - w(\cC' \cap \optWM) \\
		& > & (1- 3\eps/8) w(\optWM) - \eps w(\optWM)/4 \\
		& = & (1- 5\eps/8) w(\optWM)\,.
	\end{eqnarray*}
	From this we derive
	\begin{eqnarray*}
		w(\cCbad \cap M) & > & w(\cCbad \cap \optWM)/(1+\eps/4)  \\
		& > & (1 - \eps/4)(1 - 5 \eps/8) w(\optWM) \\
		& > & (1 - 7\eps/8) w(\optWM).
	\end{eqnarray*}
	The last chain of inequalities implies that
	\[
		w(M) \ge w(\cCbad \cap M) > (1-7\eps/8)w(\optWM) \ge \underbrace{(1-7\eps/8)(1+\eps)}_{\geq 1 \text{ for } \eps \leq 1/8} w(M) \ge w(M),
	\]
  which is a contradiction. 
  Therefore, \eqref{eq:hypothesis} holds.
	
	Now we can prove that Property~\ref{item:property-collection} holds for $\cC'$.
        \begin{equation}
		\sum_{C \in \cC'} \rb{w(C\cap \optWM) - w(\Nmatching{C})}  \stackrel{\text{by~\eqref{eq:2}}}{\ge}  \sum_{C \in \cC'} \eps w(C) / 16
		 \stackrel{\text{by~\eqref{eq:hypothesis}}}{\ge}  \eps^2 w(\optWM)/64. \label{eq:property-collection-for-cC'}
	\end{equation}	
\paragraph{Constructing $\cC$:}
By exhibiting $\cC'$, we showed that the lemma holds if Property~\ref{item:property-path} is replaced by~\eqref{eq:modified-property-path} and also when Property~\ref{item:property-large-matched-edges} is ignored. Now we prove that the lemma holds even if Property~\ref{item:property-path} is not replaced and Property~\ref{item:property-large-matched-edges} is taken into account. To that end, we obtain $\cC$ from $\cC'$ in the following way.

Initially, $\cC$ is empty. We consider each element $C \in \cC'$ separately and apply the following procedure:
\begin{itemize}
	\item Step 1: Remove all the edges $C \cap M$ violating Property~\ref{item:property-large-matched-edges}. Let $\cP$ be the obtained collection.
	\item Step 2: Add to $\cC$ all the elements of $\cP$ that satisfy Property~\ref{item:property-path}.
\end{itemize}
In the procedure above, the removed edge $e$ is never from $\optWM$, and removing an edge $e \in M$ from $C$ only increases the contribution of an edge from $\optWM$ to the remaining path. This implies that Property~\ref{item:property-large-edges} holds for the elements of $\cC$. So, for $\cC$ hold all the Properties~\ref{item:property-length}-\ref{item:property-path}. It remains to show that Property~\ref{item:property-collection} holds as well.

\paragraph{Proving Property~\ref{item:property-collection} for $\cC$:}
Let $C \in \cC'$ be an element decomposed into $\cP$ in Step~1. First, observe that it does not necessarily hold that the gain of $C$ equals to the sum of gains of the elements of $\cP$. The reason is that those edges from $C \cap M$ that are removed in Step~1 could be deducted twice when calculating the sum of gains of the elements of $\cP$. However, it is easy to upper-bound their negative contribution as follows. Recall that each $C$ has length at most $4/\eps$, so we can remove at most $4/\eps$ edges from each $C$. So, the gain-loss in $\cP$ compared to $C$ is at most $\tfrac{4}{\eps} \cdot \tfrac{\eps^6}{64} w(C)$.

Let $X \eqdef \gain(C)$. By Property~\eqref{eq:modified-property-path}, $X \ge \eps w(\Nmatching{C}) / 4$. Recall also that by~\eqref{eq:2} we showed that $X \ge \eps w(C) / 16$, hence
\begin{equation}\label{eq:lower-bound-on-X}
	X \ge \max\{\eps w(\Nmatching{C}) / 4, \eps w(C) / 16\}.
\end{equation}
When, in Step~1, $C$ is decomposed into $\cP$, by our discussion above, the sum of gains of the elements of $\cP$ is at least $X - \eps^5 w(C) / 8$. On the other hand, in Step~2, the algorithm removes all the elements of $C' \in \cP$ that have gain less than $\eps w(\Nmatching{C'}) / 8$. So, the total gain loss of $\cP$ due to Step~2 is
\[
	\sum_{C' \text{ is removed from $\cP$}} \eps w(\Nmatching{C'}) / 8 \le \eps w(\Nmatching{C}) / 8 + \eps^5 w(C) / 8.
\]

This implies that the elements of $\cP$ that are added to $\cC$ have gain at least
\[
	X - \eps^5 w(C) / 8 - \rb{\eps w(\Nmatching{C}) / 8 + \eps^5 w(C) / 8} \stackrel{\text{from~\eqref{eq:lower-bound-on-X}}; \eps < 1/16}{\ge} X / 3.
\]

We now conclude that after applying the above steps the sum of gains of the elements of $\cC$ are at least $1/3$ of that in $\cC'$. Therefore, Property~\ref{item:property-collection} for $\cC$ follows from~\eqref{eq:property-collection-for-cC'}.
\end{proof}

%%% Local Variables:
%%% mode: latex
%%% TeX-master: "000-main_random_order"
%%% End:

\subsection{Proof of~\Cref{lemma:paths-and-cycles-in-layered} (when $C$ is a Path)}

\lemmapathandcyleslayered*

\begin{proof}[Proof for the case when $C$ is a path]
	This is a simpler version of the proof for the case of cycles we presented in \cref{sec:finding-short-augmentations}. For the sake of completeness we provide its proof as well.
	
	\paragraph{Transformation of $Q$:} If $Q$ does not start by a matched edge, attach to the beginning of $Q$ an edge of weight $0$ and add that edge to the current matching. In a similar way alter $Q$ if it does not end by a matched edge. Notice that this does not change the gain of $Q$. After this transformation, we conveniently have that $Q \cap M$ equals $\Nmatching{Q} \cap M$.
	
	%Hence, without loss of generality and for the sake of analysis, we assume that $Q$ starts and ends by a matched edge.

	\paragraph{Parametrization:}
	Let $Q = v_1 \ldots v_{2 t}$. Taking into account the properties of $\cC$ and the transformation of $Q$, $Q$ has at most $4 / \eps + 2$ edges and $Q$ has odd length. So, $t \le 2 / \eps + 2 \le 4 / \eps$. Consider a parametrization $G^P$ of the graph in which
	%each edge $\{u_i, u_{i + 1}\}$, for all $1 \le i < 2 t$, is oriented as $(v_i, v_{i + 1})$. Furthermore, assume that $G^P$ is such that 
	$v_i \in R$ for each even $i$, while $v_i \in L$ for each odd $i$. By the definition, $G^P$ contains $Q$.
	
	Define $W$ to be the largest value such that $W = (1 + \eps^4)^i \le w(Q)$, for some integer $i \ge 0$. Let $a_1, \ldots, a_t$ be the matched edges of $Q$ appearing in that order, with $a_1 = \{v_1, v_2\}$. Similarly, let $b_1, \ldots, b_{t - 1}$ be the unmatched edges of $Q$ appearing in that order, with $b_1 = \{v_2, v_3\}$.

	\paragraph{Layered graph:} Now, we define a layered graph $\Layered(\tau^A, \tau^B, W, G^P)$ that contains $Q$ passing through all the layers.
	\begin{itemize}
		\item Sequence $\tau^A$ has length $t$ and $\tau^B$ has length $t - 1$.
		\item For every $a_i$, set $\tau^A_i$ to be the smallest $k \layerepspower$ such that $k$ is an integer and $\tau^A_j W \ge w\rb{a_i}$.
		\item For every $b_i$, set $\tau^B_i$ to be the largest $k \layerepspower$ such that $k$ is an integer and $\tau^B_j W \le w\rb{b_i}$.
	\end{itemize}	
	It is easy to verify that $\Layered(\tau^A, \tau^B, W, G^P)$ contains $Q$.
	
	\paragraph{Correctness:}
	All the properties~\eqref{item:cycles-length-tauA}-\eqref{item:cycles-upper-bound-sum-tauB} given in \cref{table:good-tau-pairs} follow directly by the properties of $Q$ (see \cref{lemma:shortpaths}) and the definition of $(\tau^A, \tau^B)$. (Note that property~\eqref{item:cycles-lower-bound-of-tauA-and-tauB} is not affected by appending zero-weight matched edges in the transformation of $Q$.) So, it remains to show that property~\eqref{item:cycles-diff-tauB-and-tauA} holds as well.
	
	As in the cycle case, we compare $\sum_i \tau^A_i W$ and $\sum_i \tau^B_i W$. We have
	\begin{equation}\label{eq:bound-on-sum-tauB-for-Q}
		\sum_i \tau^B_i W \ge w(Q \cap \optWM) - t \layerepspower W,
	\end{equation}
	and
	\begin{equation}\label{eq:bound-on-TW-for-Q}
		W \le w(Q) = w(Q \cap M) + w(Q \cap \optWM) \le 2 w(Q \cap \optWM).
	\end{equation}
	Using that $t \le 4 / \eps$ and combining~\eqref{eq:bound-on-sum-tauB-for-Q} and~\eqref{eq:bound-on-TW-for-Q} implies
	\begin{equation}\label{eq:final-bound-on-sum-tauB-for-Q}
		 \sum_i \tau^B_i W \ge w(Q \cap \optWM) \rb{1 - 2 \eps^{10}}.
	\end{equation}
	Next, observe that from~\eqref{eq:bound-on-TW-for-Q} and $t \le 4 / \eps$ we have
	\begin{eqnarray}
		\sum_i \tau^A_i W & \le & w(Q \cap M) + t \layerepspower W \nonumber \\
		 & \stackrel{\text{from \cref{lemma:shortpaths}}}{\le} & \frac{(1 + 4 \eps^{10})}{1 + \eps / 8} w(Q \cap \optWM) \nonumber \\
		 & \stackrel{\eps \le 1/16}{\le} & \frac{(1 + \eps^9)}{1 + \eps / 8} w(Q \cap \optWM).\label{eq:final-bound-on-sum-tauA-for-Q}
	\end{eqnarray}
	From~\eqref{eq:final-bound-on-sum-tauB-for-Q} and~\eqref{eq:final-bound-on-sum-tauA-for-Q} we derive
	\begin{eqnarray*}
		\sum_i \tau^B_i W - \sum_i \tau^A_i W 
		& \ge & \rb{\rb{1 - 2 \eps^{10}} - \frac{(1 + \eps^9)}{1 + \eps / 8}} w(Q \cap \optWM) \\
		& = & \frac{(1 + \eps / 8) \rb{1 - 2 \eps^{10}} - (1 + \eps^9)}{1 + \eps / 8} w(Q \cap \optWM) \\
		& = & \frac{\eps / 8 - \eps^9 - 2 \eps^{10} - \eps^{11} / 4}{1 + \eps / 8} w(C \cap \optWM)\,,
	\end{eqnarray*}
which is $\ge \layerepspower W$ because $\eps \le 1/16$.
	The last chain of inequalities implies
	\begin{equation}\label{eq:difference-tauA-tauB-for-Q}
		\sum_i \tau^B_i - \sum_i \tau^A_i \ge \layerepspower.
	\end{equation}
\end{proof}
%%% Local Variables:
%%% mode: latex
%%% TeX-master: "000-main_random_order"
%%% End:

\end{document}